\documentclass{lmcs} 
\pdfoutput=1
\usepackage[utf8]{inputenc}

\usepackage{lastpage}
\lmcsdoi{22}{1}{6}
\lmcsheading{}{\pageref{LastPage}}{}{}%
{Nov.~22,~2024}{Jan.~09,~2026}{}


\usepackage{color}
\usepackage{amsmath}
\usepackage{amssymb} 
\usepackage{wasysym}
\usepackage{graphicx}
\usepackage{xspace}
\usepackage{subfig}
\usepackage[export]{adjustbox}
\usepackage{stmaryrd}
\usepackage{enumitem}
\usepackage{tikz}
\usepackage{tikz-layers}
\usetikzlibrary{patterns}
\usetikzlibrary{decorations.markings}
\usepackage{xspace}
\usepackage{pgfplotstable}
\usepackage{booktabs}
\usepackage{latexsym}
\usepackage{listings}
\usepackage{numprint}
\usepackage{soul} 

\usepackage[T1]{fontenc}
\usetikzlibrary{automata, positioning, arrows}
\usetikzlibrary{arrows.meta}

\usetikzlibrary{decorations.pathmorphing}
\usetikzlibrary{calc}

\usepackage{hyperref}

\makeatletter
\DeclareRobustCommand{\cev}[1]{%
  \mathpalette\do@cev{#1}%
}
\newcommand{\do@cev}[2]{%
  \fix@cev{#1}{+}%
  \reflectbox{$\m@th#1\vec{\reflectbox{$\fix@cev{#1}{-}\m@th#1#2\fix@cev{#1}{+}$}}$}%
  \fix@cev{#1}{-}%
}
\newcommand{\fix@cev}[2]{%
  \ifx#1\displaystyle
    \mkern#23mu
  \else
    \ifx#1\textstyle
      \mkern#23mu
    \else
      \ifx#1\scriptstyle
        \mkern#22mu
      \else
        \mkern#22mu
      \fi
    \fi
  \fi
}

\makeatother

\newcommand{\cm}{CM}
\newcommand{\cmc}{CMC}

\newcommand{\copa}{CoPa}

\newcommand{\slcs}{{\tt SLCS}}
\newcommand{\slcsG}{{\tt SLCS}$_{\mkern1mu \gamma}$}
\newcommand{\slcsE}{{\tt SLCS}$_{\mkern1mu \eta}$}

\newcommand{\voxlogica}{{\tt VoxLogicA}}
\newcommand{\polylogica}{{\tt PolyLogicA}}
\newcommand{\polyvisualizer}{{\tt PolyVisualizer}}
\newcommand{\mcrl}{\texttt{mCRL2}}
\newcommand{\imgql}{\texttt{ImgQL}}

\newcommand{\SET}[1]{\{#1\}}
\newcommand{\ZET}[2]{\SET{#1 \,|\, #2}}
\newcommand{\pws}[1]{\mathbf{2}^{#1}}
\newcommand{\nats}{\mathbb{N}}
\newcommand{\reals}{\mathbb{R}}

\newcommand{\cnv}[1]{#1^{-}}
\newcommand{\dircnv}[1]{#1^{\pm}}

\newcommand{\plm}{$\pm$}
\newcommand{\upd}{$\uparrow\!\downarrow$}
\newcommand{\dwn}{$\downarrow$}

\newcommand{\relint}[1]{\widetilde{#1}}
\newcommand{\ap}{{\tt PL}}

\newcommand{\map}{\mathbb{F}}
\newcommand{\peval}[1]{\calV_{#1}}
\newcommand{\invpeval}[1]{{\calV}^{-1}_{#1}}

\newcommand{\sibis}{\sim_{\resizebox{0.2cm}{!}{$\triangle$}}}
\newcommand{\wsibis}{\approx_{\resizebox{0.2cm}{!}{$\triangle$}}}

\newcommand{\plmbis}{\sim_{\pm}}
\newcommand{\wplmbis}{\approx_{\pm}}

\newcommand{\lcceq}{\eqsign}

\newcommand{\form}{\Phi}
\newcommand{\ltrue}{{\tt true}}

\newcommand{\lneg}{\neg}

\newcommand{\slcsGeq}{\equiv_{\gamma}}
\newcommand{\slcsEeq}{\equiv_{\eta}}

\newcommand{\etga}{\calE}

\newcommand{\closure}{\calC}

\newcommand{\eqsign}{\rightleftharpoons}

\newcommand{\posTolts}{\LTS}
\newcommand{\posToltsA}{\posTolts_A}
\newcommand{\posToltsC}{\posTolts_{\mkern1mu C}}
\newcommand{\lts}{LTS}
\newcommand{\LTS}{\mathbb{S}}
\newcommand{\sact}{\mathbf{s}}
\newcommand{\dact}{\mathbf{d}}
\newcommand{\ftndact}{{\footnotesize\bf d}}
\newcommand{\Hudact}{{\Huge\bf d}}
\newcommand{\cact}{\mathbf{c}}
\newcommand{\ftncact}{{\footnotesize\bf c}}
\newcommand{\Hucact}{{\Huge\bf c}}

\newcommand{\ftncorridor}{\footnotesize\bf corridor}
\newcommand{\Hucorridor}{\Huge\bf corridor}
\newcommand{\ftngreen}{\footnotesize\bf green}
\newcommand{\Hugreen}{\Huge\bf green}
\newcommand{\ftnwhite}{\footnotesize\bf white}
\newcommand{\Huwhite}{\huge\bf white}
\newcommand{\sosrule}[2]{\displaystyle{\frac{\,\,\,\,#1\,\,\,\,}{\,\,\,\,#2\,\,\,\,}}}
\newcommand{\trans}[1]{\,\stackrel{#1}{\longrightarrow}\,}

\newcommand{\seq}{\sim}
\newcommand{\auxbeq}{\mathrel{\,%
  \raisebox{.3ex}{$\underline{\makebox[.7em]{$\leftrightarrow$}}$}\,}}
\newcommand{\beq}{\auxbeq_b}
\newcommand{\mcrltwo}{\texttt{mCRL2}}
\newcommand{\calC}{\mathcal{C}}

\newcommand{\calE}{\mathcal{E}}
\newcommand{\calF}{\mathcal{F}}

\newcommand{\calK}{\mathcal{K}}

\newcommand{\calP}{\mathcal{P}}

\newcommand{\calV}{\mathcal{V}}

\newcommand{\CYAN}[1]{\textcolor{cyan}{#1}}
\newcommand{\RED}[1]{\textcolor{red}{#1}}

\newcommand{\closedefi}{\hfill$\bullet$}
\newcommand{\closeex}{\hfill$\clubsuit$}
\newcommand{\closerem}{\hfill$\divideontimes$}
\newcommand{\sep}{\vert}

\newcounter{dgnot} 
\newenvironment{dgnot}[1][]{\refstepcounter{dgnot}\par\medskip
   \noindent \textbf{\RED{NfDiego~\thedgnot.}  #1} \rmfamily}{\medskip}

\newcounter{mknot} 
\newenvironment{mknot}[1][]{\refstepcounter{mknot}\par\medskip
   \noindent \textbf{\CYAN{NfMieke~\themknot.}  #1} \rmfamily}{\medskip}

\begin{document}

\title[Weak Simplicial Bisimilarity and Minimisation]{Weak Simplicial Bisimilarity and Minimisation for Polyhedral Model Checking
}

\titlecomment{{\lsuper*}This paper is an extended version of~\cite{Be+24}.}

\thanks{The authors are listed in alphabetical order, as they equally contributed to the work presented in this paper.}	


\author[Bezhanishvili]{Nick Bezhanishvili\lmcsorcid{0009-0005-6692-5051}}[a]
\author[Bussi]{Laura Bussi\lmcsorcid{0000-0003-1292-4086}}[b]
\author[Ciancia]{Vincenzo Ciancia\lmcsorcid{0000-0003-1314-0574}}[b]
\author[Gabelaia]{David Gabelaia\lmcsorcid{0000-0002-8317-7949}}[c]
\author[Jibladze]{Mamuka Jibladze\lmcsorcid{0000-0002-9434-9523}}[c]
\author[Latella]{Diego Latella\lmcsorcid{0000-0002-3257-9059}}[d]
\author[Massink]{Mieke Massink\lmcsorcid{0000-0001-5089-002X}}[b]
\author[de~Vink]{Erik P. de Vink\lmcsorcid{0000-0001-9514-2260}}[e]

\address{Institute for Logic, Language and Computation, University of Amsterdam, The Netherlands}	
\email{n.bezhanishvili@uva.nl}  

\address{Istituto di Scienza e Tecnologie dell'Informazione ``A. Faedo'', Consiglio Nazionale delle Ricerche, Pisa, Italy}	
\email{Laura.Bussi@cnr.it, Vincenzo.Ciancia@cnr.it, Mieke.Massink@cnr.it}  


\address{Andrea Razmadze Mathematical Institute, I. Javakhishvili Tbilisi State University, Georgia}	
\email{gabelaia@gmail.com, mamuka.jibladze@gmail.com}

\address{Formerly with Istituto di Scienza e Tecnologie dell'Informazione ``A. Faedo'', Consiglio Nazionale delle Ricerche, Pisa, Italy. Retired}	
\email{diego.latella@actiones.eu}

\address{Eindhoven University of Technology, The Netherlands}	
\email{evink@win.tue.nl}




\begin{abstract}
The work described in this paper builds on the polyhedral semantics of the \emph{Spatial Logic for Closure Spaces} (\slcs) and the geometric spatial model checker PolyLogicA. Polyhedral models are central in domains that exploit mesh processing, such as 3D computer graphics. A discrete representation of polyhedral models is given by cell poset models, which are amenable to geometric spatial model checking using \slcsE, 
a weaker version of \slcs.
In this work we  show that the mapping from polyhedral models to cell poset models preserves and reflects \slcsE. 
We also propose weak simplicial bisimilarity on polyhedral models and weak \plm-bisimilarity on cell poset models, 
where by ``weak'' we mean that the relevant equivalence is coarser than the corresponding
one for \slcs{,} leading to a greater reduction of the size of models and thus to more efficient model checking.

We show that the proposed bisimilarities enjoy the Hennessy-Milner property, i.e.\ two points are weakly simplicial bisimilar iff they are logically equivalent for~\slcsE. Similarly, two cells are weakly \plm-bisimilar iff they are logically equivalent in the poset-model interpretation of~\slcsE.
Furthermore we present a model minimisation procedure and prove that it correctly computes the minimal model with respect to weak \plm-bisimilarity, i.e.
with respect to logical equivalence of ~\slcsE. The procedure works via an encoding into \lts{s} and then exploits branching bisimilarity on those LTSs, exploiting  the minimisation capabilities as included in the \mcrl{} toolset. Various examples show the effectiveness of the approach.
\end{abstract}

\maketitle



\section{Introduction and Related Work}\label{sec:Introduction}

Spatial and spatio-temporal model checking have recently been
successfully employed in a variety of application areas, including
Collective Adaptive Systems~\cite{Ci+16a,Ci+18,AguzziAV24,AudritoDT24}, signal analysis~\cite{Ne+18},
image analysis~\cite{Ci+16,Ha+15,Ba+20}, and
polyhedral modelling~\cite{Be+22,Ci+23c,Be+24,Be+24a}. Interest in
these methods for spatial analysis is increasing in Computer
Science and in other domains, including initially unanticipated
ones, such as medical imaging~\cite{Be+19,Be+21}.

Spatial model checking is a global technique: it comprises the
automatic verification of properties, expressed in a suitable spatial
logic, such as the Spatial Logic for Closure Spaces (\slcs)~\cite{Ci+14,Ci+16}, for each point of a 
spatial model. The logic \slcs{} has been defined originally for   closure models, i.e. models based on
\v{C}ech closure spaces~\cite{Cec66}, a generalisation of topological
spaces, and model checking algorithms have been developed for finite
closure models also in combination with discrete time, leading to
spatio-temporal model checking~\cite{Ci+18}.  The spatial model
checker \voxlogica, proposed in~\cite{Be+19}, is very
efficient in checking properties of large images -- represented as
symmetric finite closure models -- expressed in~\slcs{}
\cite{Be+19,Be+19a,Be+21}. For example, the automatic segmentation via
a suitable \slcs{} formula characterising the white matter of the
brain in a 3D MRI image consisting of circa 12M voxels (i.e.\
$256 \times 256 \times 181$), requires approximately 10~seconds, using
\voxlogica{} on a desktop computer~\cite{Be+19a}.\footnote{Intel Core
  i9-9900K processor (with 8 cores and 16 threads) and 32GB of RAM\@. Note that
  \voxlogica{} checks such logical specifications for \emph{every}
  point in the model exploiting parallel execution, memoization, and
  state-of-the-art imaging libraries~\cite{Be+19}.}
 
In~\cite{Ci+22,Ci+23} several bisimulations for finite closure spaces have been studied, with the aim to improve the efficiency of model checking via model minimisation. These notions cover a spectrum from \cm-bisimilarity, an equivalence based on {\em proximity} --- similar to and inspired by topo-bisimilarity for topological models~\cite{vBB07} --- to \cmc-bisimilarity, \cm-bisimilarity specialisation for quasi-discrete closure models, and \copa-bisimilarity, an equivalence based on {\em conditional reachability}. Each of these bisimilarities has been equipped with its logical characterisation.

\begin{figure}
\centering
	\subfloat[]{\label{subfig:green_rooms}
		\includegraphics[valign=c,height=6em]{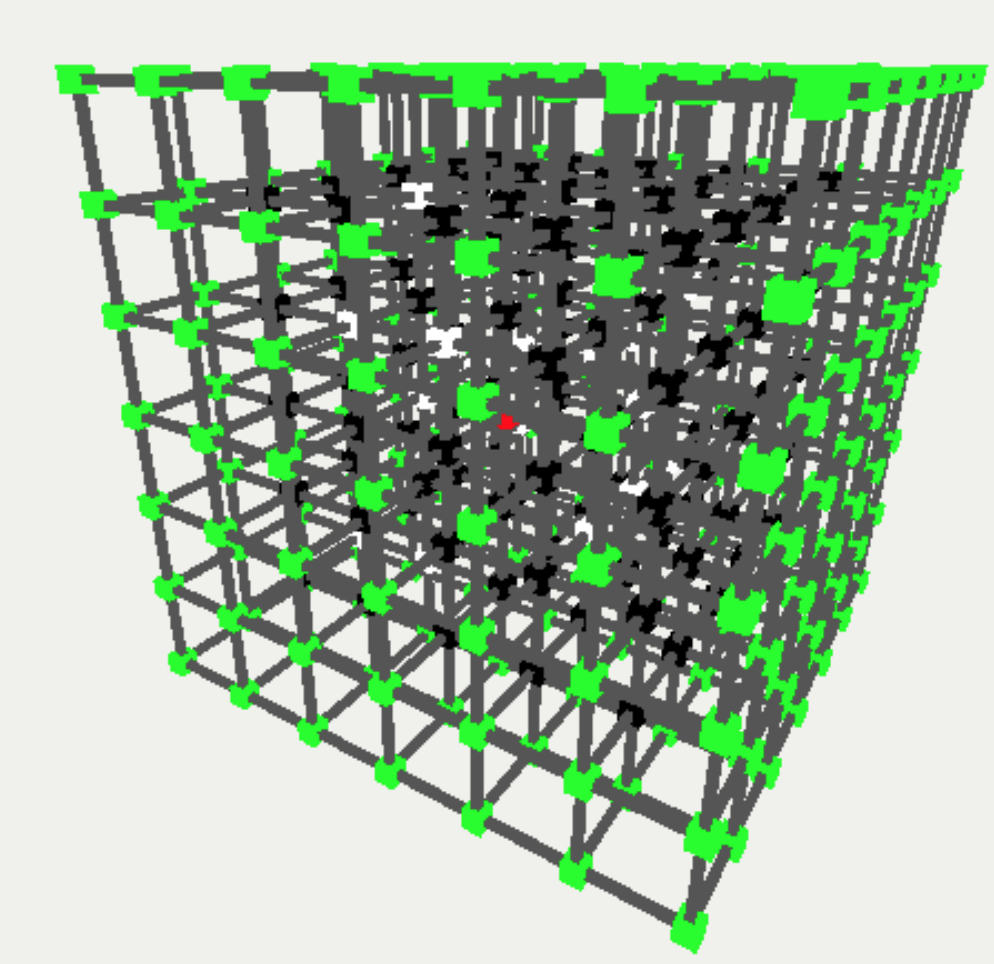}
	}
	\subfloat[]{\label{subfig:black_or_white_rooms}
		\includegraphics[valign=c,height=6em]{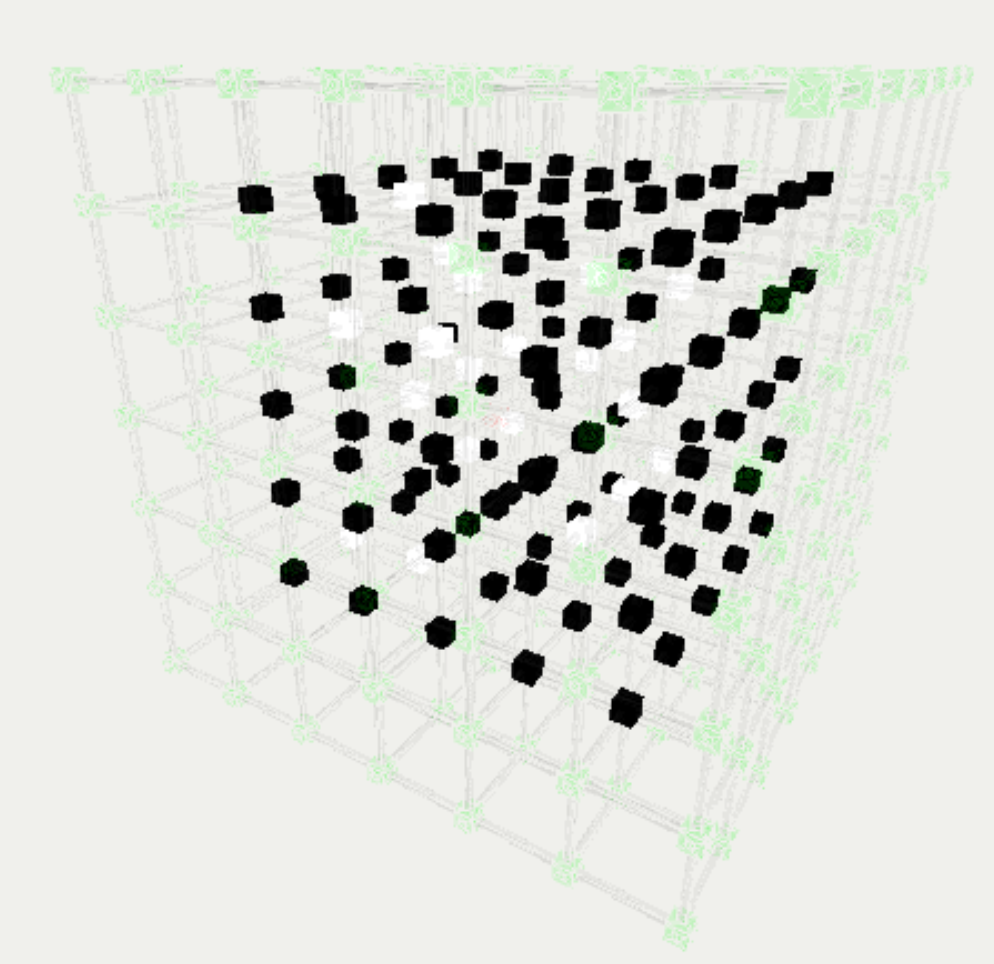}
	}
	\subfloat[]{\label{subfig:red_rooms}
		\includegraphics[valign=c,height=6em]{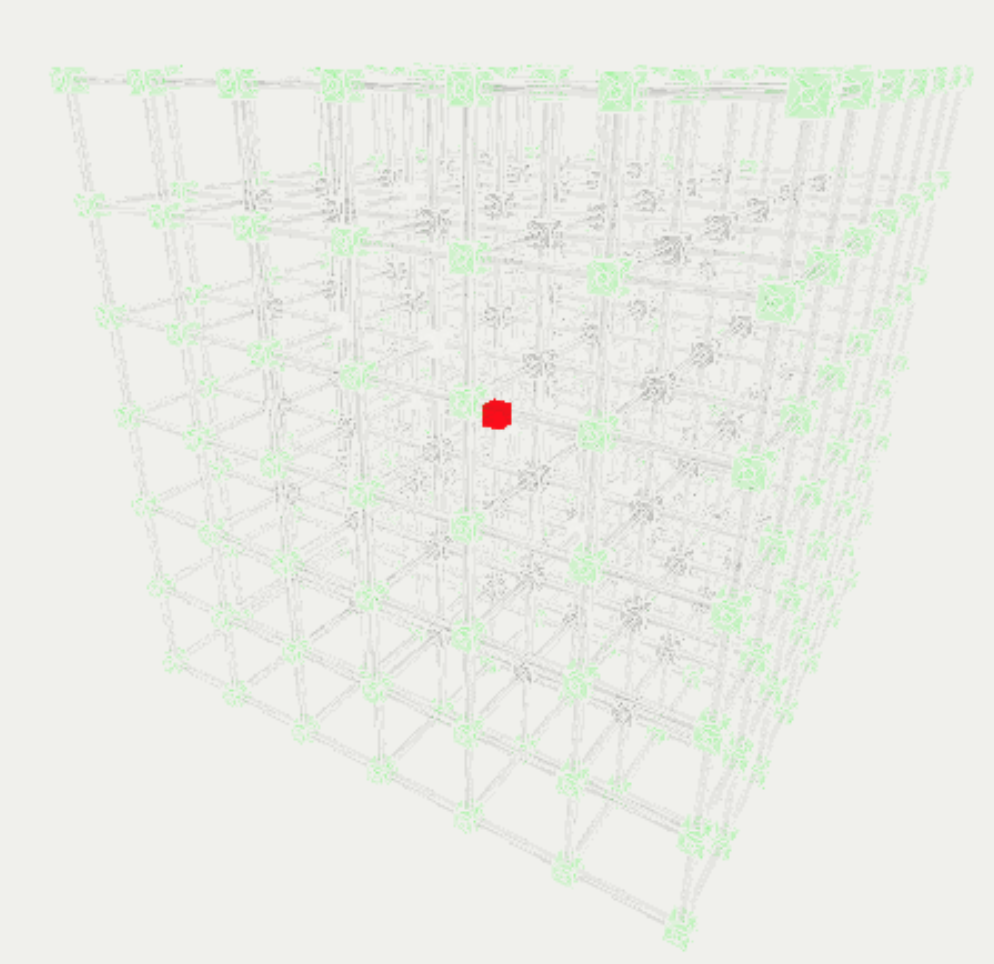}
	}
\caption{\label{fig:3Dmaze} 3D maze~(\ref{subfig:green_rooms}), black and white rooms~(\ref{subfig:black_or_white_rooms}) and red rooms~(\ref{subfig:red_rooms}) in the 3D maze (source~\cite{Be+22}).}
\end{figure}

The spatial model checking techniques mentioned above, targeting
grid-based structures, have been extended to \emph{polyhedral
  models}~\cite{Be+22,LoQ23}. 
  Polyhedra are subsets
in~$\reals^n$ generated by simplicial complexes, i.e.\ finite
collections of simplexes satisfying certain conditions. A simplex is
the convex hull of a set of affinely independent points. Given a set~$\ap$ of proposition letters, a polyhedral model is obtained from a
polyhedron by assigning a polyhedral subset to each proposition letter
$p \in \ap$, namely those points that ``satisfy'' proposition $p$.  Polyhedral
models in~$\reals^3$ can be used for (approximately) representing
objects in continuous 3D~space. This is typical of many 3D~visual
computing techniques, where an object is split into suitable geometric parts of
different size. Such ways of splitting of an object are known as mesh
techniques and include triangular surface meshes or tetrahedral volume
meshes (see~\cite{LevinePRZ2012}).
Interestingly, polyhedral models can conveniently be represented by
discrete structures, the so-called \emph{cell poset models}: each
point of the polyhedron is mapped to a (unique) ``cell'', i.e.\ an element of
the associated cell poset model. Cell poset models, being a particular case of Kripke models,  are amenable to discrete model checking.

 In~\cite{Be+22}, a variant of \slcs{} for polyhedral models, called \slcsG{} in the sequel, as well as a geometric model checking algorithm have been proposed. The latter has been implemented in the \polylogica{} model checker, together with \polyvisualizer, a tool for visualising and inspecting polyhedral models (see~\cite{Be+22} for details).
 Example~\ref{ex:3Dcube} below gives an idea of the framework of spatial model checking using
 \polylogica.

 \begin{exa}\label{ex:3Dcube}
Figure~\ref{subfig:green_rooms} shows a ``3D maze''  example originating from~\cite{Be+22}.
The maze
consists of ``rooms'' that are connected by ``corridors''. The rooms come in four colours:
white, black, green, and red for only one room.
The cells of  white, black, green, red  rooms satisfy (only) predicate letter $\mathbf{white}$, $\mathbf{black}$,
$\mathbf{green}$, $\mathbf{red}$, respectively. 
Predicate letter $\mathbf{corridor}$ is satisfied by (all and only the cells of) corridors.
The green rooms are all situated at the outer
boundary of the maze and represent the surroundings of the maze that can be reached via
an exit. The white, black, and red rooms and related corridors are situated inside the maze
and form the maze itself. Figure~\ref{subfig:black_or_white_rooms} shows all the white and black rooms. Figure~\ref{subfig:red_rooms} shows
the red room. 
The corridors between rooms are dark grey.
Valid paths through the maze should only pass by white/red rooms  and related corridors
to reach a green room without passing by black rooms or corridors that connect to black
rooms. 
All the images shown in Figure~\ref{fig:3Dmaze} are generated by \polylogica{}
and can be visualised (and inspected by) \polyvisualizer:
the result of a model checking session is presented by showing an image where
 the cells that satisfy the formula of interest are shown opaque, while the rest of the image 
 is shown transparent in the background. For instance, in Figure~\ref{subfig:black_or_white_rooms}
 the result of model checking  
 the simple \slcsG{} formula 
 $\mathbf{black}  \, \lor \, \mathbf{white}$ by \polylogica{} is shown, and similarly for 
 Figure~\ref{subfig:red_rooms} and formula $\mathbf{red}$. \closeex
 \end{exa}

\slcsG{} can express spatial properties of points lying in polyhedral models, and, in
particular, \emph{conditional reachability} properties. Besides
negation and conjunction, \slcsG{} provides the $\gamma$~reachability operator. Informally, a
point~$x$ in a polyhedral model satisfies the conditional reachability
formula $\gamma(\form_1,\form_2)$ if there is a topological path
starting from~$x$, ending in a point~$y$ satisfying $\form_2$, and
such that all the intermediate points of the path between $x$ and~$y$
satisfy~$\form_1$. Note that neither $x$ nor~$y$ is required to
satisfy~$\form_1$.  Many  interesting properties, such as
proximity (in the topological sense, i.e. ``being in the topological
closure of'') or ``being surrounded by'' can be expressed using
reachability (see~\cite{Be+22}). 

Moreover, in~\cite{Be+22} {\em simplicial bisimilarity} 
(denoted by $\sibis$ in the sequel)
has been proposed for polyhedral models, and it has been shown that it enjoys the Hennessy-Milner Property (HMP) with respect to \slcsG. 
In~\cite{Ci+23c} {\em \plm-bisimilarity} 
(denoted by $\plmbis$ in the sequel)
has been proposed for cell poset models, that also enjoys the HMP for \slcsG. 

In this paper we introduce a weaker version of conditional
reachability, denoted by~$\eta$.
A point~$x$ in a polyhedral model
satisfies the conditional reachability formula $\eta(\form_1,\form_2)$
if there is a topological path starting
from~$x$, ending in a point~$y$ satisfying~$\form_2$, and $x$~and all
the intermediate points of the path between $x$ and~$y$
satisfy~$\form_1$. Thus now $x$~is required to satisfy~$\form_1$. 
The operator $\eta$ can  be expressed
using~$\gamma$ and we will show that 
the logic where $\gamma$~has been replaced by~$\eta$ --- \slcsE, in the sequel ---
is strictly weaker than \slcsG{} in the sense
that it distinguishes fewer points than~\slcsG. Furthermore, as mentioned above,
\slcsG~can express proximity --- that boils down to the standard \emph{possibility}
modality~$\Diamond$ in the poset model interpretation --- whereas
\slcsE{}~cannot.
We show that the mapping from a
polyhedral model to its cell poset model preserves and reflects
\slcsE: a point satisfies a formula of~\slcsE{} if and only if the cell which it
is mapped to satisfies the formula\footnote{A similar feature was shown to hold for \slcsG\ in~\cite{Be+22}.}. This result paves the way to
the definition and implementation of model checking techniques for
\slcsE{} on polyhedral models, by working on their discrete
representations.

{\em Model reduction} for cell poset models, as a means for
{\em improving model checking efficiency} is our 
main concern in the present work. 
In particular, we are interested in techniques based on
{\em spatial} bisimilarity.
For that purpose we 
introduce {\em weak simplicial bisimilarity} on polyhedral models ($\wsibis$)  
showing that it enjoys the HMP with respect
to \slcsE{} --- $\wsibis$ coincides with the logical equivalence~$\slcsEeq$ as induced by \slcsE{} ---
and a notion of bisimulation equivalence for cell poset models, 
namely \emph{weak \plm-bisimilarity}
($\wplmbis$, to be read as \lq{}weak plus-minus\rq{} bisimilarity) such that
two points in the polyhedral model are weakly simplicial bisimilar if and only if their
cells are weakly \plm-bisimilar. We show that also on cell poset
models the HMP holds: $\wplmbis$ coincides with $\slcsEeq$.

The reason why we are interested in \slcsE{} 
is that  it characterises bisimilarities --- in the polyhedral model and the associated poset model --- that are coarser than simplicial
bisimilarity and \plm-bisimilarity, respectively (thence the adjective ``weak'' in the names of the two bisimilarities). This allows for greater model reduction, 
as we will see, for instance, in Example~\ref{ex:min} and Figure~\ref{fig:exa:MinRunExaE}. 
At the same time,   interesting reachability properties can be expressed in \slcsE,
as shown, for instance, by the following example.

\begin{exa}\label{ex:3DcubeFormulas}
Let us consider again the polyhedral model of Figure~\ref{subfig:green_rooms}.
Suppose we are interested in all those white rooms from which an
exit (i.e. green room) can be reached without passing by black rooms or
corridors connected to black rooms. Moreover, we want to know which route --- in the sense of rooms and corridors --- one can follow from each such white room for reaching an exit. We start by defining some auxiliary
formulas: a cell satisfies formula  
$
\eta(\mathbf{corridor},\mathbf{white}) \land 
\lneg\eta(\mathbf{corridor},\mathbf{green} \lor \mathbf{black} \lor \mathbf{red})
$ 
if it belongs to a corridor and from such a cell only (cells of) white rooms --- i.e. neither green, nor black, nor red --- can be reached
via the corridor.
For the sake of readability, we name such a formula $\mathtt{CorWW}$.
Formula $\mathtt{CorWG}$, defined as
$
\eta(\mathbf{corridor},\mathbf{white}) \land \eta(\mathbf{corridor},\mathbf{green}),
$
is satisfied by those cells of corridors between white and green rooms. Next, we
define formula $\mathtt{WtG}$ that characterises the cells of 
white rooms, corridors between white rooms, and corridors between white and green rooms, by which one can reach a green room, i.e. without passing by black rooms or corridors connected to black rooms:
$
\mathtt{WtG} = \eta((\mathbf{white} \lor \mathtt{CorWW} \lor \mathtt{CorWG}), \mathbf{green}).
$
Keeping in mind that in the answer to our model checking query we want to see
the green exits as well, we define the complete query $\mathtt{Q1}$ by
$
\mathtt{WtG} \lor \eta(\mathbf{green},\mathtt{WtG}).
$
The result of \polylogica{} applied on $\mathtt{Q1}$ and the ``maze'' is shown in Figure~\ref{subfig:connectionWhiteGreen}.

Suppose now we are interested in showing the
white rooms, and connecting corridors, from which both a green  room and the red room can be reached, without having to pass by black rooms (and related corridors), i.e. 
we want to show if and how one can reach an exit from the red room.
The  relevant query $\mathtt{Q2}$ is given by the formula
$
\eta((\mathtt{Q1} \lor \mathtt{CorWR}), \mathbf{red}) \lor
\eta((\mathbf{red} \lor \mathtt{CorWR}), \mathtt{Q1})
$
where  $\mathtt{CorWR}$ stands for
$
\eta(\mathbf{corridor}, \mathbf{white}) \land \eta(\mathbf{corridor}, \mathbf{red})
$.
The result of the model checking session is shown in Figure~\ref{subfig:whiteConnectsRedGreen}.

Finally, Figure~\ref{subfig:no_exit_rooms} shows the white rooms, and related corridors, from which it is {\em not} possible to reach a green room without having to pass by a black room and is the result of
model checking the formula $\mathtt{Q3}$ defined as 
$
(\mathbf{white} \lor \mathtt{CorWW}) \land \lneg\mathtt{WtG}.
$
\closeex
\end{exa}

\begin{figure}
\centering
	\subfloat[]{\label{subfig:connectionWhiteGreen}
		\includegraphics[valign=c,height=6em]{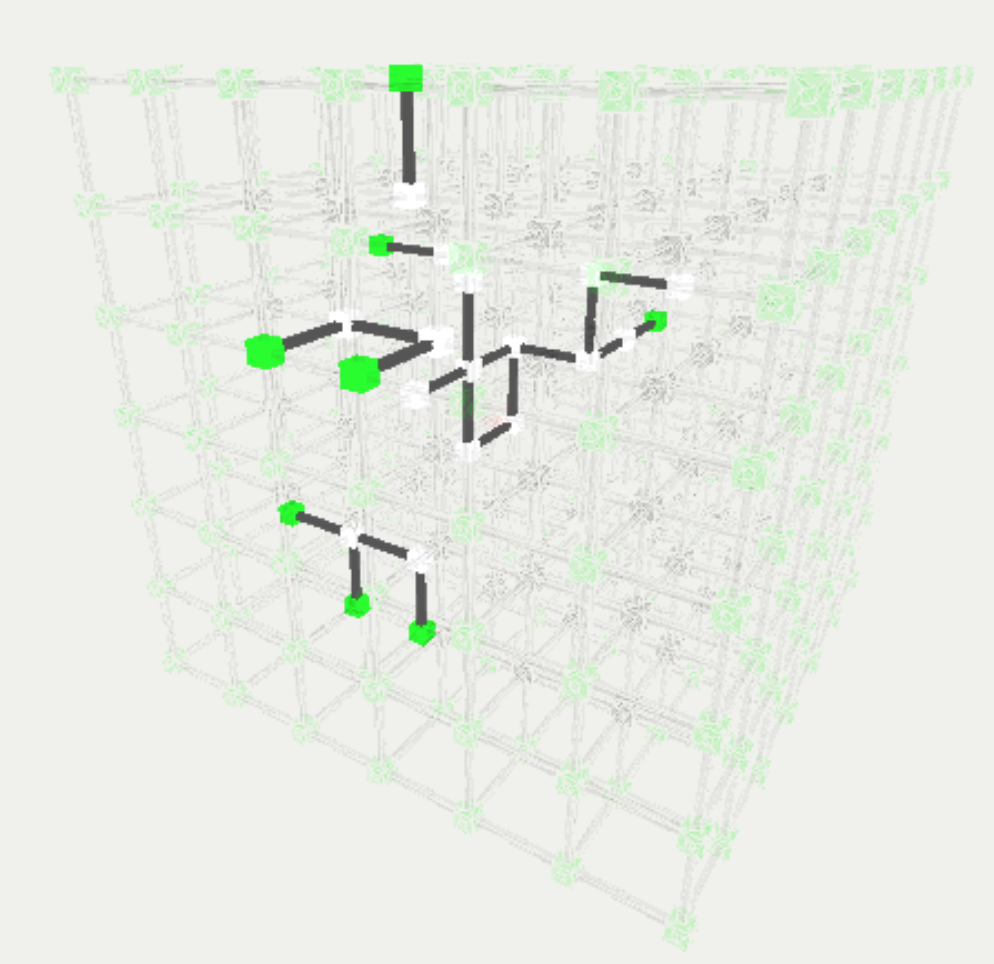}
	}
	\subfloat[]{\label{subfig:whiteConnectsRedGreen}
		\includegraphics[valign=c,height=6em]{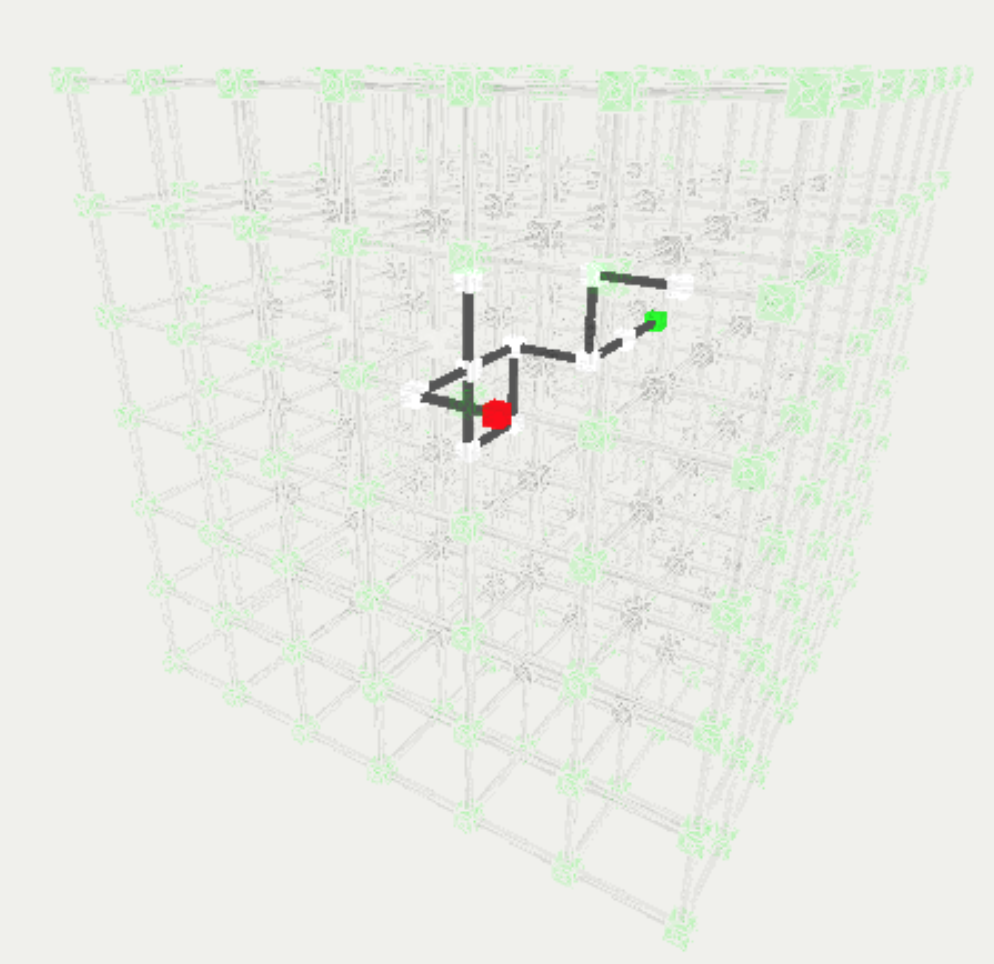}
	}
	\subfloat[]{\label{subfig:no_exit_rooms}
		\includegraphics[valign=c,height=6em]{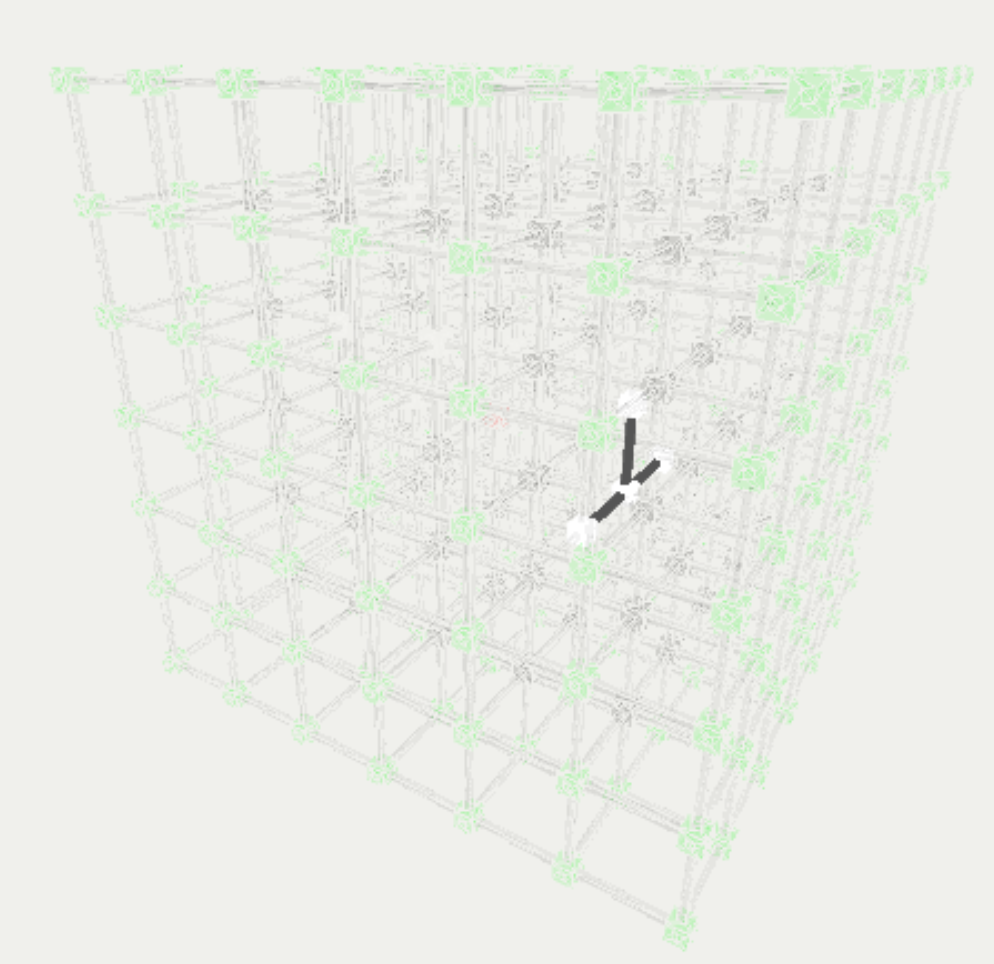} 
	}
\caption{\label{fig:3DmazeMC} 
Spatial model checking results of the properties $\mathtt{Q1}$ (\ref{subfig:connectionWhiteGreen}), $\mathtt{Q2}$ (\ref{subfig:whiteConnectsRedGreen}) and $\mathtt{Q3}$ (\ref{subfig:no_exit_rooms}) for the 3D maze of Figure~\ref{fig:3Dmaze}. 
(source:~\cite{Be+22}).}
\end{figure}

Building upon the theoretical results for \slcsE, weak simplicial bisimilarity and weak \plm-bisimilarity, 
we introduce a minimisation procedure based on weak \plm-bisimilarity, namely {\em weak \plm-minimisation}.  The
procedure uses an encoding of cell poset models into labelled
transition systems (\lts{s}) following an approach that is similar to that presented
in~\cite{Ci+23a} for finite closure models. More precisely, in the
case of cell poset models, there is a one-to-one correspondence
between the states of the \lts{} and the cells of the poset model.  It
is shown that two cells are weakly \plm-bisimilar in the poset model
if and only if they --- as states of the encoded \lts{} --- are branching bisimulation
equivalent.  This provides an effective way for computing the
equivalence classes for the set of cells, from which the minimal model
is built, on which \slcsE{} model checking can be safely performed.
In fact, efficient \lts{} minimisation tools are available for branching bisimulation, such as the one provided by the \mcrl\ toolset~\cite{Gr+17}. As we will see in Section~\ref{sec:Experiments}, this can lead to a drastic reduction of 
the size of the spatial model, thus increasing the practical efficiency of spatial model checking.
Figure~\ref{subfig:c3x5x4} shows an example of a maze, composed of 6,145 cells of three colours: white, green, and grey --- for corridors. This model is reduced to an LTS consisting of only 38 states, which is a reduction of two orders of magnitude. The different white, green and grey states of the minimised LTS represent the various equivalence classes of cells in the original polyhedral model. Even if this is a synthetic example, chosen on purpose for its symmetry properties, it illustrates the potential of the approach. Figure~\ref{subfig:c3x5x4min} only gives a first visual impression of spatial minimisation for polyhedra. We postpone the discussion of the details to Section~\ref{sec:Experiments}. 

\begin{figure}
  \begin{center}
  \subfloat[\label{subfig:c3x5x4}Maze]{\includegraphics[width=0.45\textwidth]{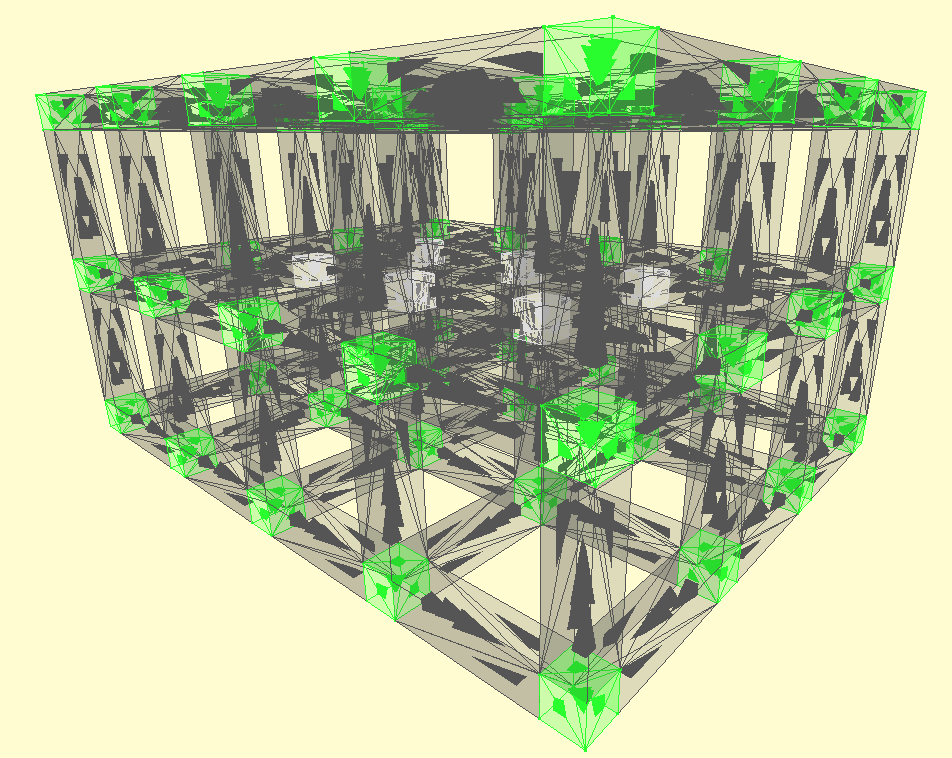}}\quad
  \subfloat[\label{subfig:c3x5x4min}Minimal LTS]{
  \resizebox{3in}{!}{

  }
  }
  \end{center}
  \caption{A maze (\ref{subfig:c3x5x4}) and its respective minimal LTS (\ref{subfig:c3x5x4min}).}\label{fig:LcubeLTS}
\end{figure}

In conclusion, in the present paper, we  focus on {\em model reduction} --- as a  way of improving model checking efficiency --- and spatial {\em reachability} --- rather than {\em proximity}.
In particular, we are interested in a framework for model reduction with the following features:
\begin{enumerate}
\item It should be  {\em sound} 
and {\em complete}, i.e. be based on a notion of bisimilarity that enjoys the Hennessy-Milner Property (HMP) so that completeness and soundness
 of the optimised model checking procedure --- via model reduction --- are guaranteed.
\item It should be {\em optimal} 
with respect to the logic of interest, in the  sense of yielding 
the minimal model with respect to the equivalence induced by the logic of interest,
but also a {\em useful} one. In this respect we have been inspired by the use of branching bisimilarity in the context of LTSs:
branching bisimilarity --- that is weaker than strong bisimilarity --- enjoys the HMP with respect to CTL$^*$ without $X$ (next) --- that is weaker than full CTL$^*$ --- and both the equivalence and its logical characterisation are widely used in  concurrency theory and its applications. In essence, weak simplicial bisimilarity in the context of spatial logic is a re-interpretation in space of branching bisimilarity in the context of temporal logic. Similarly, \slcsE{} can be seen as the spatial counterpart of CTL$^* \setminus X$.
\item It should exploit {\em existing tools} for minimisation via bisimulation, since at present powerful and efficient model minimisation techniques and tools are available for branching bisimilarity minimisation. 
\end{enumerate}
As we mentioned above, the fact that
logical equivalence $\slcsEeq$ is {\em coarser} than $\slcsGeq$ implies
 that poset model minimisation modulo $\slcsEeq$ results in models that can be 
smaller than those obtained modulo $\slcsGeq$, and this is one reason why we  focus on \slcsE{} in the present paper. 
As is to be expected, we do not have a general measure of the ``gain'',
in terms of percentage of reduction in the number of cells of the input models, when using
$\slcsEeq$ instead of $\slcsGeq$, because this depends on the specific model.

Furthermore,  we show that \slcsE{} is of interest for reasoning about reachability, which is an essential feature in topological structures, as illustrated by the 
examples presented in this paper.
There are also additional notions that can easily be expressed using the $\eta$ modality
such as  ``double reachability'' and  ``being surrounded''.
The former are properties like ``there is a path (from the point of interest) reaching --- while passing only through points satisfying~$\form_1$ --- a point satisfying 
$\form_2$ that can (also) be reached from a point satisfying $\form_3$ via a path passing through
points satisfying $\form_2$''. By exploiting the non-directionality of topological paths, this can be expressed by the following \slcsE{} formula:
$$
\eta(\form_1,\eta(\form_2,\form_3)).
$$
A formula like the above can be used for modelling an emergency egress situation --- e.g. in a building modelled as a polyhedral model --- in which, for instance, 
$\form_1$ characterises points in a building (such as the one schematised by the polyhedral model shown in Figure~\ref{fig:3Dmaze}) that are accessible to somebody to be rescued in that building (including the place where the person is located), but are not accessible 
to a rescue team;
$\form_3$ characterises the place where the rescue team is located while
$\form_2$ characterises points that are accessible to the rescue team (here we assume 
that $\form_3$ implies $\form_2$ --- if not, just replace $\form_3$ with $\form_3 \land \form_2$). 
The team and those to be rescued can thus meet in a point satisfying the nested $\eta$-formula $\eta(\form_2,\form_3)$.\\[0.5em]
The notion of ``being surrounded'' can be expressed using the $\eta$ modality as described below. 
We say that starting from a point $x$ that satisfies $\form_1$ one 
cannot ``escape'' from $\form_1$ without ``passing through'' $\form_2$ --- i.e. is ``surrounded'' by $\form_2$ --- if 
any path starting from $x$ and reaching a point that does not satisfy $\form_1$
must first pass through $\form_2$. More precisely, $x$ must satisfy $\form_1$ and there is
no path from $x$ leading to a point satisfying neither $\form_1$ nor~$\form_2$ without first passing 
through a point satisfying $\form_2$. In \slcsE{} this is captured by the following formula:
$$
\form_1 \land \lnot\eta(\lnot \form_2,\lneg(\form_1 \lor \form_2)).
$$
Note that if $x$ itself satisfies $\form_2$, then starting from $x$ one cannot escape from $\form_1$ without passing through $\form_2$.\footnote{
As we will see in Section~\ref{sec:slcsE}, the spatial properties discussed above can be expressed also in \slcsG{} (see Lemma~\ref{lem:etgaCorrectG}).
}

Below, we summarise the main  contributions of this paper:
\begin{itemize}
\item presentation of \slcsE{,} a spatial logic for polyhedral models which is weaker than \slcsG{;}
\item introduction of {\em weak simplicial bisimilarity} on polyhedral models ($\wsibis$) and 
showing that it enjoys the HMP with respect to \slcsE{;}
\item introduction of {\em weak \plm-bisimilarity} on cell poset models ($\wplmbis$) with the corresponding HMP result;
\item introduction of a novel cell poset model minimisation procedure based on 
weak \plm-bisimilarity --- and exploiting an encoding to \lts{s} and branching bisimilarity --- including the formal proof of its correctness;
\item proof-of-concept of the practical potential and effectiveness of this approach through a prototype toolchain and spatial model checking examples. It is shown that the cell poset models can be drastically reduced by several orders of magnitude.
\end{itemize}

The first three items above have been presented originally in~\cite{Be+24} where only some of the proofs of the relevant results where shown: in the present paper, all  proofs are presented in detail. The last
two items above are original contributions.

The paper is structured as follows. We provide a summary of
necessary background information in Section~\ref{sec:BackAndNotat}.
Section~\ref{sec:slcsE} introduces~\slcsE{} and addresses its
relationship with~\slcsG. It is also shown that \slcsE{} is preserved
and reflected by the mapping~$\map$ from polyhedral models to finite cell
poset models. Weak simplicial bisimilarity and weak \plm-bisimilarity
are defined in Section~\ref{sec:WeakBis} where it is also shown that
they enjoy the HMP with respect to the interpretation of \slcsE{} on polyhedral
models and on finite poset models, respectively. 
The minimisation procedure, based
on weak \plm-bisimilarity and exploiting its relationship with
branching bisimulation equivalence, is defined in
Section~\ref{sec:EtaMinimisation} where its correctness is also addressed. 
The procedure is currently implemented by means of an
experimental toolchain using \mcrltwo{} and is introduced in
Section~\ref{sec:toolchain}. Examples of use of the toolchain are
presented in Section~\ref{sec:Experiments}. Conclusions and a
discussion on future work are reported in
Section~\ref{sec:ConclusionsFW}. 

Finally, in 
Appendix~\ref{apx:DetailedProofs} detailed proofs are provided
and, in Appendix~\ref{apx:AdditionalExamples}, an additional minimisation
example is shown.

\section{Background and Notation}\label{sec:BackAndNotat}

In this section we introduce notation and recall necessary background information, the relevant
details of the language \slcsG, its polyhedral and poset models,
the truth-preserving map $\map$ between these models, simplicial bisimilarity and \plm-bisimilarity.

For sets $X$ and~$Y$, a function $f:X \to Y$, and subsets
$A \subseteq X$ and $B \subseteq Y$ we define $f(A)$ and~$f^{-1}(B)$
as $\ZET{f(a)}{a \in A}$ and $\ZET{a}{f(a) \in B}$, respectively.
The  \emph{restriction} of~$f$ on~$A$ is denoted by~$f|A$.
The powerset of~$X$ is denoted by~$\pws{X}$.
For a binary relation $R \subseteq {X \times X}$ we let
$\cnv{R} = \ZET{(y,x)}{(x,y)\in R}$ denote its converse and let
$\dircnv{R}$ denote $R \, \cup \cnv{R}$. For partial orders~$\preceq$
we will use the standard notation~$\succeq$ for~$\cnv{\preceq}$ and
$x \prec y$ whenever $x \preceq y$ and $x \neq y$ (and similarly
for~$x \succ y$). If $R$ is an equivalence relation on $A$, we let
$A{/R}$ denote the {\em quotient} of $A$ via $R$.
In the remainder of the paper we assume that a set~$\ap$ of
\emph{proposition letters} is fixed. The sets of natural numbers and
of real numbers are denoted by $\nats$ and~$\reals$, respectively. We
use the standard interval notation: for $x,y \in \reals$ we let
$[x,y]$ be the set $\ZET{r \in \reals}{x \leq r \leq y}$,
$[x,y) = \ZET{r\in \reals}{x \leq r < y}$, and so on. Intervals
of~$\reals$ are equipped with the Euclidean topology inherited
from~$\reals$. We use a similar notation for intervals over~$\nats$:
for $n,m \in \nats$, $[m;n]$ denotes the set
$\ZET{i \in \nats}{m \leq i \leq n}$,
$[m;n) = \ZET{i \in \nats}{m \leq i < n}$, and so on.
Finally, for topological space $(X,\tau)$ and $A\subseteq X$ we let $\closure_T(A)$ denote the  topological closure of $A$.

Below we recall some basic notions, assuming that the
reader is familiar with topological spaces, Kripke models, and
posets.

\subsection{Polyhedral Models and Cell Poset Models}
A \emph{simplex} $\sigma$ of dimension $d$ is the convex hull of a set 
$\SET{\mathbf{v_0},\ldots, \mathbf{v_d}}$ of
$d+1$~affinely independent points in~$\reals^m$, with $d \leq m$,
i.e.\
$\sigma = \ZET{ \lambda_0\mathbf{v_0} + \ldots +
  \lambda_d\mathbf{v_d}}{\lambda_0,\ldots,\lambda_d \in [0,1]\mbox{
    and } \sum_{i=0}^{d} \lambda_i = 1}$. For instance, a segment~$AB$
together with its end-points $A$ and~$B$
is a simplex in~$\reals^m$, for$~m \geq 1$.  Any subset of the set $\SET{\mathbf{v_0},\ldots, \mathbf{v_d}}$ of
points characterising a simplex~$\sigma$ induces a simplex~$\sigma'$ in turn,
and we write $\sigma' \sqsubseteq \sigma$, noting that
$\sqsubseteq$~is a partial order, e.g.\ 
$A \sqsubseteq A \sqsubseteq AB$, $B \sqsubseteq B \sqsubseteq AB$ and $AB \sqsubseteq AB$. 
The {\em barycentre} $b_{\sigma}$ of $\sigma$ is defined as  follows:
$
b_{\sigma} =
\sum_{i=0}^d \frac{1}{d+1}\mathbf{v_i}
$.

The \emph{relative interior} $\relint{\sigma}$ of a simplex~$\sigma$ is the
same as $\sigma$ ``without its borders'', i.e.\ the set
$\ZET{ \lambda_0\mathbf{v_0} + \ldots +
  \lambda_d\mathbf{v_d}}{\lambda_0,\ldots,\lambda_d \in (0,1]\mbox{
    and } \sum_{i=0}^{d} \lambda_i = 1}$. For instance, the open
segment~$\relint{AB}$, without the end-points $A$ and~$B$ is the
relative interior of segment~$AB$. The relative interior of a simplex
is often called a~\emph{cell} and is equal to the topological interior
taken inside the affine hull of the simplex.\footnote{But note that
  the relative interior of a simplex composed of just a single point
  is the point itself and not the empty set.}  
A  partial order is defined on cells:  we say that 
$\relint{\sigma_1} \preccurlyeq \relint{\sigma_2}$ if and only if
$\relint{\sigma_1} \subseteq \closure_T(\relint{\sigma_2})$ 
where, we recall, $\closure_T$ denotes the 
topological closure operator.
It is easy to see that 
$\preccurlyeq$ is indeed a partial order.
Note furthermore that $\sqsubseteq$ and $\preccurlyeq$ are compatible, in the sense that
$\relint{\sigma_1} \preccurlyeq \relint{\sigma_2}$ if and only if
$\sigma_1 \sqsubseteq \sigma_2$.
In the above
example, we have
$\relint{A}\preccurlyeq \relint{A} \preccurlyeq \relint{AB}, \relint{B}\preccurlyeq \relint{B} \preccurlyeq \relint{AB},$ and
$\relint{AB}\preccurlyeq \relint{AB}$.

A \emph{simplicial complex}~$K$ is a finite collection of simplexes
of~$\reals^m$ such that: (i) if $\sigma \in K$ and
$\sigma' \sqsubseteq \sigma$ then also $\sigma' \in K$; (ii) if
$\sigma, \sigma' \in K$ and $\sigma \cap \sigma' \not=\emptyset$, then
$\sigma \cap \sigma' \sqsubseteq \sigma$ and
$\sigma \cap \sigma' \sqsubseteq \sigma'$.
The \emph{cell poset} of simplicial complex~$K$ is
$(\relint{K},\preccurlyeq)$ where $\relint{K}$ is the set
$\ZET{\, \relint{\sigma}}{\sigma \in K}$, and $\preccurlyeq$ is the union of the partial orders on the cells of the simplexes of $K$. 

The polyhedron~$|K|$ of~$K$
is the set-theoretic union of the simplexes in~$K$. Note that
$|K|$~inherits the topology of~$\reals^m$ and that $\relint{K}$ forms a partition of polyhedron $|K|$.
Note furthermore that different simplicial complexes can give rise to the same polyhedron.

A \emph{polyhedral model} is a pair $\calP = (P,\peval{\calP})$
where $P=|K|$ for some simplicial complex $K$ and
$\peval{\calP}: \ap \to \pws{P}$ maps every proposition letter
$p \in \ap$ to the set of points of $P$ satisfying~$p$. It is
required that, for all $p \in \ap$, $\peval{\calP}(p)$ is always a
union of cells in~$\relint{K}$.  A  poset model is a triple
$\calF = (W,\preccurlyeq,\peval{\calF})$ where $(W,\preccurlyeq)$ is a poset that is
equipped with a valuation function
$\peval{\calF} : \ap \to \pws{W}\!$. Given a polyhedral model
$\calP = (P,\peval{\calP})$ with $P=|K|$, for some simplicial complex $K$, we say that
$\calF = (W,\preccurlyeq,\peval{\calF})$ is the \emph{cell
  poset model} of~$\calP$ relative to $K$ if and only if $W=\relint{K}$, $(\relint{K},\preccurlyeq)$
is the cell poset of~$K$, and, for all $\relint{\sigma}\in \relint{K}$,
we have: $\relint{\sigma} \in \peval{\calF}(p)$ if and only if
$\relint{\sigma} \subseteq \peval{\calP}(p)$. 
We will omit to specify ``relative to $K$'' if this is clear from the context.
  For all $x \in P$, we let $\map(x)$ denote the unique
  cell~$\relint{\sigma}\in \relint{K}$ such that $x \in \relint{\sigma}$. Note that
  $\map(x)$ is well defined, since $\relint{K}$ is a partition of $|K|$, and that
  $\map: P \to \relint{K}$ is a continuous function~\cite[Corollary
  3.4]{BMMP2018}. With slight overloading, we let
  $\map(\calP)$~denote the cell poset model of~$\calP$. 
  In the following,
  when we say that $\calF$ is a cell poset model, we mean that there
  exist a simplicial complex $K$ and a polyhedral model 
  $\calP=(|K|,\peval{\calP})$ such that $\calF = \map(\calP)$. 
  Finally, note
  that poset models are a subclass of Kripke models. 

Figure~\ref{fig:PolyhedronNoPathCompressed} shows a polyhedral
model. There are three proposition letters, $\mathbf{red}$,
$\mathbf{green}$, and $\mathbf{grey}$, shown by different colours
(\ref{subfig:PolyhedronNoPathCompressed}). The model is ``unpacked''
into its cells in
Figure~\ref{subfig:PolyhedronNoPathCellsCompressed}. The latter are
collected in the cell poset model, whose Hasse diagram is shown in
Figure~\ref{subfig:PolyhedronNoPathPosetCompressed}.

\begin{figure}[h]
\subfloat[]{\label{subfig:PolyhedronNoPathCompressed}
\resizebox{0.9in}{!}
{
\begin{tikzpicture}[scale=1.4,label distance=-2pt]
	    \tikzstyle{point}=[circle,draw=black,fill=white,inner sep=0pt,minimum width=4pt,minimum height=4pt]
	    \node (p0)[point,draw=red,label={270:$B$}] at (0,0) {};
	    	\filldraw [red] (p0) circle (1.25pt);
	    \node (p1)[point,draw=gray,label={ 90:$A$}] at (0,1) {};
	    	\filldraw [gray] (p1) circle (1.25pt);
	    \node (p2)[point,draw=gray,label={270:$D$}] at (1,0) {};
	    \node (p3)[point,draw=red,label={ 90:$C$}] at (1,1) {};
	    	\filldraw [red] (p3) circle (1.25pt);
	    \node (p4)[point,draw=gray,label={270:$F$}] at (2,0) {};
	    \node (p5)[point,draw=gray,label={ 90:$E$}] at (2,1) {};

	    \draw [red   ,thick](p0) -- (p1);
	    \draw [red   ,thick](p0) -- (p2);
	    \draw [red   ,thick](p0) -- (p3);
	    \draw [red   ,thick](p1) -- (p3);
	    \draw [red   ,thick](p2) -- (p3);	    
    \draw [dashed      ](p2) -- (p4);
    \draw [dashed      ](p2) -- (p5);
    \draw [dashed      ](p3) -- (p5);
    \draw [dashed      ](p4) -- (p5);
    \draw [gray,thick](p2) -- (p4);
    \draw [gray,thick](p4) -- (p5);
    \draw [gray,thick](p2) -- (p5);
    \draw [gray,thick](p3) -- (p5);
	        
	    \begin{scope}[on background layer]
	    \fill [fill=red!50  ](p0.center) -- (p1.center) -- (p3.center);
	    \fill [fill=red!50  ](p0.center) -- (p3.center) -- (p2.center);
	    \fill [fill=green!50](p2.center) -- (p3.center) -- (p5.center);
            \fill [fill=gray!50](p2.center) -- (p4.center) -- (p5.center);    	    
            \end{scope}

    \filldraw [gray] (p2) circle (1.25pt);
    \filldraw [gray] (p4) circle (1.25pt);
    \filldraw [gray] (p5) circle (1.25pt); 
	\end{tikzpicture}
	}
}
\subfloat[]{\label{subfig:PolyhedronNoPathCellsCompressed}
\resizebox{1.7in}{!}
{
\begin{tikzpicture}[scale=1.3,label distance=-2pt]
	    \tikzstyle{point}=[circle,fill=white,inner sep=0pt,minimum width=4pt,minimum height=4pt]
	    \node (p0S0d)[point,draw=red,fill=red,label={270:$B$}] at (0,0) {};
	    \node (p0S1d)[point] at (0.33,0) {};
	    \node (p0S2d)[point] at (0.66,0) {};
	    \node (p0S3d)[point] at (0.99,0) {};
	    \node (p0S0u)[point] at (0,0.33) {};
	    \node (p0S1u)[point] at (0.33,0.33) {};
	    \node (p0S2u)[point] at (0.66,0.33) {};
	    \node (p0S3u)[point] at (0.99,0.33) {};
	    
	    \node (p1S0d)[point] at (0,1.33) {};
	    \node (p1S1d)[point] at (0.33,1.33) {};
	    \node (p1S0u)[point,fill=gray,label={90:$A$}] at (0,1.66) {};
	    \node (p1S1u)[point] at (0.33,1.66) {};
	    
	    \node (p2S0u)[point] at (1.99,0.33) {};
	    \node (p2S0d)[point] at (1.99,0) {};
	    \node (p2S1u)[point] at (2.32,0.33) {};
	    \node (p2S1d)[point,fill=gray,label={270:$D$}] at (2.32,0.0) {};
	    \node (p2S2u)[point] at (2.65,0.33) {};
	    \node (p2S3u)[point] at (2.98,0.33) {};
	    \node (p2S4u)[point] at (3.31,0.33) {};
	    \node (p2S4d)[point] at (3.31,0.0) {};
	    
	    \node (p3S0u)[point] at (1.33,1.66) {};
	    \node (p3S1u)[point,fill=red,label={90:$C$}] at (2.32,1.66) {};
	    \node (p3S0d)[point] at (1.33,1.33) {};
	    \node (p3S1d)[point] at (1.66,1.33) {};
	    \node (p3S2u)[point] at (2.65,1.66) {};
	    \node (p3S3u)[point] at (2.98,1.66) {};
	    \node (p3S2d)[point] at (1.99,1.33) {};
	    \node (p3S3d)[point] at (2.32,1.33) {};
	    \node (p3S4d)[point] at (2.65,1.33) {};
	    
	    \node (p4S0u)[point] at (4.31,0.33) {};
	    \node (p4S0d)[point] at (4.31,0.0) {};
	    \node (p4S1u)[point] at (4.64,0.33) {};
	    \node (p4S1d)[point,fill=gray,label={270:$F$}] at (4.64,0.0) {};	    
	    
	    \node (p5S0u)[point] at (3.65,1.66) {};
	    \node (p5S1u)[point,fill=gray,label={90:$E$}] at (4.64,1.66) {};
	    \node (p5S0d)[point] at (3.65,1.33) {};
	    \node (p5S1d)[point] at (3.98,1.33) {};
	    \node (p5S2d)[point] at (4.31,1.33) {};
	    \node (p5S3d)[point] at (4.64,1.33) {};
	    
	    \draw [red,thick](p0S0u) -- (p1S0d);
	    \draw [red,thick](p1S1u) -- (p3S0u);
	    \draw [red,thick](p0S2u) -- (p3S1d);
	    \draw [red,thick](p0S3d) -- (p2S0d);
	    \draw [red,thick](p3S3d) -- (p2S1u);
	    \draw [gray,thick](p2S3u) -- (p5S1d);
	    \draw [gray,thick](p3S2u) -- (p5S0u);
	    \draw [gray,thick](p2S4d) -- (p4S0d);
	    \draw [gray,thick](p5S3d) -- (p4S1u);
	    	    
	    \begin{scope}[on background layer]
	    \fill [fill=red!50  ](p0S1u.center) -- (p1S1d.center) -- (p3S0d.center);
 	    \fill [fill=red!50  ](p0S3u.center) -- (p3S2d.center) -- (p2S0u.center);
	    \fill [fill=green!50  ](p2S2u.center) -- (p3S4d.center) -- (p5S0d.center);
	    \fill [fill=gray!50  ](p2S4u.center) -- (p5S2d.center) -- (p4S0u.center);    
            \end{scope}
	\end{tikzpicture}
	}
}
\subfloat[]{\label{subfig:PolyhedronNoPathPosetCompressed}
\resizebox{2.5in}{!}
{
\begin{tikzpicture}[scale=20, every node/.style={transform shape}]
    \tikzstyle{kstate}=[rectangle,draw=black,fill=white]
    \tikzset{->-/.style={decoration={
		markings,
		mark=at position #1 with {\arrow{>}}},postaction={decorate}}}
    
    \node[kstate,fill=red!50  ] (P0) at (  1,0) {$\relint{B}$};
    \node[kstate,fill=lightgray!50  ] (P1) at (  0,0) {$\relint{A}$};
    \node[kstate,fill=lightgray!50] (P2) at (3.5,0) {$\relint{D}$};
    \node[kstate,fill=red!50  ] (P3) at (2.5,0) {$\relint{C}$};
    \node[kstate,fill=lightgray!50] (P4) at (  6,0) {$\relint{F}$};
    \node[kstate,fill=lightgray!50] (P5) at (  5,0) {$\relint{E}$};

    \node[kstate,fill=red!50] (E0) at (-1,1) {$\relint{AB}$};
    \node[kstate,fill=red!50] (E1) at ( 2,1) {$\relint{BD}$};
    \node[kstate,fill=red!50] (E2) at ( 1,1) {$\relint{BC}$};
    \node[kstate,fill=red!50  ] (E3) at ( 0,1) {$\relint{AC}$};
    \node[kstate,fill=red!50  ] (E4) at ( 3,1) {$\relint{CD}$};
	\node[kstate,fill=lightgray!50] (E5) at ( 6,1) {$\relint{DF}$};
	\node[kstate,fill=lightgray!50] (E6) at ( 5,1) {$\relint{DE}$};
	\node[kstate,fill=lightgray!50] (E7) at ( 4,1) {$\relint{CE}$};
	\node[kstate,fill=lightgray!50] (E8) at ( 7,1) {$\relint{EF}$};

    \node[kstate,fill=red!50] (T0) at ( 2,2) {$\relint{BCD}$};
    \node[kstate,fill=red!50] (T1) at ( 0,2) {$\relint{ABC}$};
    \node[kstate,fill=lightgray!50] (T2) at ( 6,2) {$\relint{DEF}$};
    \node[kstate,fill=green!50] (T3) at ( 4,2) {$\relint{CDE}$};

    \draw (P0) to (E0);
    \draw (P0) to (E1);
    \draw (P0) to (E2);

    \draw (P1) to (E0);
    \draw (P1) to (E3);

    \draw (P2) to (E1);
    \draw (P2) to (E4);
    \draw (P2) to (E5);
    \draw (P2) to (E6);

    \draw (P3) to (E2);
    \draw (P3) to (E3);
    \draw (P3) to (E4);
    \draw (P3) to (E7);

    \draw (P4) to (E5);
    \draw (P4) to (E8);

    \draw (P5) to (E6);
    \draw (P5) to (E7);
    \draw (P5) to (E8);
    \draw (E0) to (T1);
    \draw (E2) to (T0);
    \draw (T1) to (E2);
    
    \draw (E1) to (T0);   
    \draw (E3) to (T1);    
    \draw (E4) to (T0);
    \draw (E4) to (T3);
    \draw (E5) to (T2);
    \draw (E6) to (T2);
	\draw (E6) to (T3);
	\draw (E7) to (T3);
	\draw (E8) to (T2);

\end{tikzpicture}
}
}
\caption{A polyhedral model $\calP_{\ref{fig:PolyhedronNoPathCompressed}}$
  (\ref{subfig:PolyhedronNoPathCompressed}) with its cells
  (\ref{subfig:PolyhedronNoPathCellsCompressed}) and the Hasse diagram
  of the related cell poset
  (\ref{subfig:PolyhedronNoPathPosetCompressed}).} 
\label{fig:PolyhedronNoPathCompressed}
\end{figure}

\subsection{Paths}
In a topological space $(X,\tau)$, a \emph{topological path} from
$x\in X$ is a total, continuous function $\pi : [0,1] \to X$ such that
$\pi(0)=x$.  We call $\pi(0)$ and~$\pi(1)$ the \emph{starting point}
and \emph{ending point} of~$\pi$, respectively, while $\pi(r)$~is an
\emph{intermediate point} of~$\pi$, for all $r \in
(0,1)$. Figure~\ref{subfig:PolyhedronWithPath} shows a path from a
point~$x$ in the open segment~$\relint{AB}$ to point $D$ in the polyhedral model of
Figure~\ref{subfig:PolyhedronNoPathCompressed}.

Topological paths relevant for our work 
are represented in cell posets by so-called \plm-paths, a subclass of
undirected paths~\cite{Be+22}. For technical reasons\footnote{We are
  interested in model checking structures resulting from the
  minimisation, via bisimilarity, of cell poset models, and such
  structures are often just (reflexive) Kripke models rather than
  poset models.}
in this paper we extend the definition given in~\cite{Be+22} to
general Kripke frames.

Given a Kripke frame $(W,R)$, an \emph{undirected path} of length
$\ell \in \nats$ from~$w$ is a total function $\pi : [0;\ell] \to W$
such that $\pi(0) = w$
and, for all $i \in [0;\ell)$,
$\dircnv{R}(\pi(i),\pi(i+1))$. The \emph{starting point} and
\emph{ending point} are $\pi(0)$ and~$\pi(\ell)$, respectively, while
$\pi(i)$ is an intermediate point, for all $i \in (0;\ell)$. For an
undirected path~$\pi$ of length~$\ell$ we often use the sequence
notation $(w_i)_{i=0}^{\ell}$ where $w_i = \pi(i)$ for $i \in [0;\ell]$.  

Given paths $\pi' = (w'_i)_{i=0}^{\ell'}$ and $\pi'' = (w''_i)_{i=0}^{\ell''}$,
with $w'_{\ell'} = w''_0$, the \emph{sequentialisation} $\pi' \cdot \pi'' : [0;\ell' + \ell''] \to W$
of~$\pi'$ with~$\pi''$ is the path 
from~$w'_0$ defined as follows:
$$
(\pi' \cdot \pi'')(i) = 
\left\{
\begin{array}{l}
\pi'(i), \text{ if } i \in [0;\ell'],\\
\pi''(i-\ell'), \text{ if } i \in [\ell'; \ell'+ \ell''].
\end{array}
\right.
$$

For a path $\pi = (w_i)_{i=0}^{\ell}$ and $k \in [0;\ell]$ we define the
$k$-shift of~$\pi$, denoted by $\pi{\uparrow} k$, as follows:
$\pi{\uparrow} k = (w_{j+k})_{j=0}^{\ell-k}$ and, for
$0 < m \leq \ell$, we let $\pi {\leftarrow} m$ denote the path
obtained from~$\pi$ by inserting a copy of $\pi(m)$ immediately
before~$\pi(m)$ itself. In other words, we have:
$\pi {\leftarrow} m = (\pi |[0;m]) \cdot ((\pi(m),\pi(m)) \cdot (\pi
{\uparrow} m))$.  Finally, any path $\pi |[0;k]$, for some
$k \in [0;\ell]$, is a \emph{(non-empty) prefix} of~$\pi$.

An undirected path $\pi : [0;\ell] \to W$ is a \emph{\plm-path} if and
only if $\ell\geq 2$, $R(\pi(0),\pi(1))$ and
$\cnv{R}(\pi(\ell-1),\pi(\ell))$.

\begin{figure}
\begin{center}
\subfloat[]{\label{subfig:PolyhedronWithPath}
\resizebox{1.2in}{!}{
\begin{tikzpicture}[scale=1.3,label distance=-2pt]
	    \tikzstyle{point}=[circle,draw=black,fill=white,inner sep=0pt,minimum width=4pt,minimum height=4pt]
	    \node (p0)[point,draw=red,label={270:$B$}] at (0,0) {};
	    	\filldraw [red] (p0) circle (1.25pt);
	    \node (p1)[point,draw=gray,label={ 90:$A$}] at (0,1) {};
	    	\filldraw [gray] (p1) circle (1.25pt);
	    \node (p2)[point,draw=gray,label={270:$D$}] at (1,0) {};
	    \node (p3)[point,draw=red,label={ 90:$C$}] at (1,1) {};
	    	\filldraw [red] (p3) circle (1.25pt);
	    \node (p4)[point,draw=gray,label={270:$F$}] at (2,0) {};
	    \node (p5)[point,draw=gray,label={ 90:$E$}] at (2,1) {};

	    \draw [red   ,thick](p0) -- (p1);
	    \draw [red   ,thick](p0) -- (p2);
	    \draw [red   ,thick](p0) -- (p3);
	    \draw [red   ,thick](p1) -- (p3);
	    \draw [red   ,thick](p2) -- (p3);	    
    \draw [dashed      ](p2) -- (p4);
    \draw [dashed      ](p2) -- (p5);
    \draw [dashed      ](p3) -- (p5);
    \draw [dashed      ](p4) -- (p5);
    \draw [gray,thick](p2) -- (p4);
    \draw [gray,thick](p4) -- (p5);
    \draw [gray,thick](p2) -- (p5);
    \draw [gray,thick](p3) -- (p5);
	        
	    \begin{scope}[on background layer]
	    \fill [fill=red!50  ](p0.center) -- (p1.center) -- (p3.center);
	    \fill [fill=red!50  ](p0.center) -- (p3.center) -- (p2.center);
	    \fill [fill=green!50](p2.center) -- (p3.center) -- (p5.center);
            \fill [fill=gray!50](p2.center) -- (p4.center) -- (p5.center);    	    
            \end{scope}

	    \node at (-.15,.5) {$x$};
	    \fill[blue] (0,.5) circle (.7pt);
	    \draw [line width= 0.5mm, blue] plot [smooth,tension=1] coordinates { (0,.5) (.6,.3) (1,0)};
	    \fill[blue] (1,0) circle (.7pt);
    \filldraw [gray] (p2) circle (1.25pt);
    \filldraw [gray] (p4) circle (1.25pt);
    \filldraw [gray] (p5) circle (1.25pt); 
	\end{tikzpicture}
}
	}\quad\quad\quad
\subfloat[]{\label{subfig:PosetWithPath}
\resizebox{2.5in}{!}{
\begin{tikzpicture}[scale=0.8, every node/.style={transform shape}]
    \tikzstyle{kstate}=[rectangle,draw=black,fill=white]
    \tikzset{->-/.style={decoration={
		markings,
		mark=at position #1 with {\arrow{>}}},postaction={decorate}}}
    
    \node[kstate,fill=red!50  ] (P0) at (  1,0) {$\relint{B}$};
    \node[kstate,fill=lightgray!50  ] (P1) at (  0,0) {$\relint{A}$};
    \node[kstate,fill=lightgray!50,draw=blue,thick] (P2) at (3.5,0) {$\relint{D}$};
    \node[kstate,fill=red!50  ] (P3) at (2.5,0) {$\relint{C}$};
    \node[kstate,fill=lightgray!50] (P4) at (  6,0) {$\relint{F}$};
    \node[kstate,fill=lightgray!50] (P5) at (  5,0) {$\relint{E}$};

    \node[kstate,fill=red!50  ,draw=blue,thick] (E0) at (-1,1) {$\relint{AB}$};
    \node[kstate,fill=red!50  ] (E1) at ( 2,1) {$\relint{BD}$};
    \node[kstate,fill=red!50  ,draw=blue,thick] (E2) at ( 1,1) {$\relint{BC}$};
    \node[kstate,fill=red!50  ] (E3) at ( 0,1) {$\relint{AC}$};
    \node[kstate,fill=red!50  ] (E4) at ( 3,1) {$\relint{CD}$};
	\node[kstate,fill=lightgray!50] (E5) at ( 6,1) {$\relint{DF}$};
	\node[kstate,fill=lightgray!50] (E6) at ( 5,1) {$\relint{DE}$};
	\node[kstate,fill=lightgray!50] (E7) at ( 4,1) {$\relint{CE}$};
	\node[kstate,fill=lightgray!50] (E8) at ( 7,1) {$\relint{EF}$};

    \node[kstate,fill=red!50  ,draw=blue,thick] (T0) at ( 2,2) {$\relint{BCD}$};
    \node[kstate,fill=red!50  ,draw=blue,thick] (T1) at ( 0,2) {$\relint{ABC}$};
    \node[kstate,fill=lightgray!50] (T2) at ( 6,2) {$\relint{DEF}$};
    \node[kstate,fill=green!50] (T3) at ( 4,2) {$\relint{CDE}$};

    \draw (P0) to (E0);
    \draw (P0) to (E1);
    \draw (P0) to (E2);

    \draw (P1) to (E0);
    \draw (P1) to (E3);

    \draw (P2) to (E1);
    \draw (P2) to (E4);
    \draw (P2) to (E5);
    \draw (P2) to (E6);

    \draw (P3) to (E2);
    \draw (P3) to (E3);
    \draw (P3) to (E4);
    \draw (P3) to (E7);

    \draw (P4) to (E5);
    \draw (P4) to (E8);

    \draw (P5) to (E6);
    \draw (P5) to (E7);
    \draw (P5) to (E8);

    \draw[blue,line width= 0.6mm,->=.5] (E0) to (T1);

    \draw (E1) to (T0);
    
    \draw[blue,line width= 0.6mm,->=.5] (E2) to (T0);
    
    \draw[blue,line width= 0.6mm,->=.5] (T1) to (E2);

    \draw (E3) to (T1);
    
    \draw (E4) to (T0);
    \draw (E4) to (T3);

	\draw (E5) to (T2);

	\draw (E6) to (T2);
	\draw (E6) to (T3);

	\draw (E7) to (T3);

	\draw (E8) to (T2);

	\begin{scope}[on background layer]
		\draw[blue,line width= 0.6mm,->=.2] (T0) to (P2);
		\draw[blue,line width= 0.6mm,->=.8] (T0) to (P2);
	\end{scope}
\end{tikzpicture}
}
}
\end{center}
\caption{(\ref{subfig:PolyhedronWithPath}) A topological path $\pi$ from a
  point $x$ to vertex $D$ in the polyhedral model $\calP_{\ref{fig:PolyhedronNoPathCompressed}}$ of
  Figure~\ref{subfig:PolyhedronNoPathCompressed}. (\ref{subfig:PosetWithPath})
  The corresponding \plm-path
  $(\relint{A},\relint{ABC},\relint{BC},\relint{BCD},\relint{D})$, in blue, in the Hasse diagram of the
  cell poset model $\map(\calP)$. Note that the \plm-path does not pass through $\relint{CD}$ but it goes directly from $\relint{BCD}$ to $\relint{D}$. This reflects the fact that, for small $\epsilon>0$  we have $\pi(1-\epsilon) \in \relint{BCD}$ while $\pi(1) = D$ and $\pi([0,1]) \cap \relint{CD} = \emptyset$.
  }
\label{fig:poset}
\end{figure}

\begin{exa}\label{ex:plmPath}
The \plm-path
$(\relint{AB},\relint{ABC},\relint{BC},\relint{BCD},\relint{D})$,
drawn in blue in Figure~\ref{subfig:PosetWithPath}, 
passes through the same cells, and in the same order, as the topological path from $x$  in the polyhedral model
$\calP_{\ref{fig:PolyhedronNoPathCompressed}}$ of Figure~\ref{fig:PolyhedronNoPathCompressed}
shown in
Figure~\ref{subfig:PolyhedronWithPath}  (source~\cite{Ci+23c}).\closeex
\end{exa}

Note that a topological path could, in principle, pass through some cells infinitely often. Such paths are not relevant for our theory since they play no role in the semantics of the logic and have no impact on weak simplicial bisimilarity, neither on the proofs of related results and, consequently, we are not interested in representing them. We will come back to this issue
in Section~\ref{sec:WeakBis}.

In the context of this paper it is often convenient to use a
generalisation of \plm-paths, so-called ``down paths'', \dwn-paths for
short: a \dwn-path from~$w$, of length $\ell \geq 1$, is an undirected
path~$\pi$ from~$w$ of length~$\ell$ such that
$\cnv{R}(\pi(\ell-1),\pi(\ell))$. Finally, it is also convenient to
use a subclass of \plm-paths, namely \upd-paths (to be read ``up-down
paths''): an \emph{\upd-path} from~$w$, of length~$2 \ell$, for
$\ell \geq 1$, is a \plm-path~$\pi$ of length~$2 \ell$ such that
$R(\pi(2i),\pi(2i+1))$ and $\cnv{R}(\pi(2i+1),\pi(2i+2))$, for all
$i \in [0;\ell)$.

Clearly, every \upd-path is also a \plm-path and every \plm-path is
also a \dwn-path. The following lemmas ensure that in \emph{reflexive}
Kripke frames \upd-, \plm-, and \dwn-paths can be safely used
interchangeably since for every \plm-path there is an \upd-path with
the same starting and ending points and with the same set of
intermediate points, occurring in the same order
(Lemma~\ref{lem:pm2ud} below, 
proven in Appendix~\ref{apx:prf:lem:pm2ud}).
Furthermore, for every \dwn-path there
is a \upd-path with the same starting and ending points and with the
same set of intermediate points, occurring in the same order
(Lemma~\ref{lem:d2ud} below, 
proven in Appendix~\ref{apx:prf:lem:d2ud}). 
Finally, for every \dwn-path
there is a \plm-path with the same starting and ending points and with
the same set of intermediate points, occurring in the same order
(Lemma~\ref{lem:d2plm} below, 
proven in Appendix~\ref{apx:prf:lem:d2plm}).

\begin{lem}\label{lem:pm2ud}
  Given a reflexive Kripke frame $(W,R)$ and a \plm-path
  $\pi : [0;\ell] \to W$, there is a \upd-path
  $\pi' : [0;\ell']\to W$, for some~$\ell'$, and a total, surjective,
  monotonic non-decreasing function $f : [0;\ell'] \to [0;\ell]$ such
  that $\pi'(j) = \pi(f(j))$ for all $j \in [0;\ell']$. \qed
\end{lem}

\begin{lem}\label{lem:d2ud}
  Given a reflexive Kripke frame $(W,R)$ and a \dwn-path
  $\pi : [0;\ell] \to W$, there is a \upd-path
  $\pi' : [0;\ell''] \to W$, for some~$\ell'$, and a total,
  surjective, monotonic non-decreasing function
  $f : [0;\ell'] \to [0;\ell]$ such that $\pi'(j) = \pi(f(j))$ for all
  $j \in [0;\ell']$. \qed
\end{lem}

\begin{lem}\label{lem:d2plm}
  Given a reflexive Kripke frame $(W,R)$ and a \dwn-path
  $\pi : [0;\ell] \to W$, there is a \plm-path
  $\pi' : [0;\ell''] \to W$, for some~$\ell'$, and a total,
  surjective, monotonic, non-decreasing function
  $f : [0;\ell'] \to [0;\ell]$ with $\pi'(j) = \pi(f(j))$ for all
  $j \in [0;\ell']$. \qed
\end{lem}

\subsection{The Logic \slcsG{} and Related Bisimilarities}\label{subsec:slcsG}
In~\cite{Be+22}, \slcsG, a version of \slcs{} for polyhedral models,
has been presented that consists of predicate letters, negation,
conjunction, and the single modal operator~$\gamma$, expressing
conditional reachability. The satisfaction relation for
$\gamma(\form_1, \form_2)$, for a polyhedral model
$\calP = (P,\peval{\calP})$, with $P=|K|$ for some simplicial complex $K$,
and $x\in P$,
 as defined in~\cite{Be+22}, is recalled below:\\[0.5em]
$
\begin{array}{l c l}
\calP, x \models \gamma(\form_1,\form_2) & \Leftrightarrow &
\mbox{a } \mbox{topological path } \pi: [0,1] \to |K| \mbox{ exists such that } \pi(0)=x,\\&&
\calP, \pi(1) \models \form_2, \mbox{and }
\calP, \pi(r) \models \form_1 \mbox{ for all } r\in \!\!(0,\!1).
\end{array}
$\\[0.5em]
We also recall the interpretation of \slcsG{} on poset models. The
satisfaction relation for $\gamma(\form_1, \form_2)$, for a poset model
$\calF=(W,\preccurlyeq,\peval{\calF})$ and
$w\in W$, is as follows:\\[0.5em]
$
\begin{array}{l c l}
\calF, w \models \gamma(\form_1,\form_2) & \Leftrightarrow &
\mbox{a } \mbox{\plm-path } \pi: [0;\ell] \to W \mbox{ exists such that } \pi(0)=w,\\&&
\calF, \pi(\ell) \models \form_2, \mbox{and }
\calF, \pi(i) \models \form_1 \mbox{ for all } i\in \!\!(0;\!\ell).
\end{array}
$\\[0.5em]
In~\cite{Be+22} it has also been shown that, for all 
$x \in P$
and \slcsG{} formulas~$\form$, we have: $\calP,x \models \form$ if and
only if $\map(\calP),\map(x) \models \form$.  In addition,
\emph{simplicial bisimilarity}, a novel notion of bisimilarity for
polyhedral models, has been  defined. It is based on the notion of {\em simplicial path}: 
given a polyhedral model $\calP=(P,\peval{\calP})$, with $P=|K|$ for some simplicial complex $K$, a topological path $\pi$ in $P$ is {\em simplicial} if and only if there is a finite sequence
$r_0 = 0 < \ldots < r_k = 1$ of values in [0,1] and cells 
$\relint{\sigma}_1, \ldots, \relint{\sigma}_k \in \relint{K}$ such that, for all $i\in [1;k]$ we have
that $\pi((r_{i-1},r_i)) \subseteq \relint{\sigma}_i$.\footnote{Essentially, simplicial paths have been
introduced for  avoiding to have to deal with ``bad'' paths, e.g. paths that 
can oscillate infinitely often between a set of cells.}

\begin{defi}\label{def:sibis}
Given a polyhedral model $\calP=(P,\peval{\calP})$, with $P=|K|$ for some simplicial complex $K$, a symmetric binary relation  $Z \subseteq P \times P$ is a {\em simplicial bisimulation}
if, for all $x_1, x_2 \in P$, whenever $Z(x_1,x_2)$ holds, we have that:
\begin{enumerate}
\item
$\invpeval{\calP}(x_1)=\invpeval{\calP}(x_2)$ and
\item 
for each simplicial path
$\pi_1$ from $x_1$ there is a simplicial path $\pi_2$ from $x_2$, such that
$Z(\pi_1(r),\pi_2(r))$ for all $r\in[0,1]$.
\end{enumerate}
Two points $x_1,x_2 \in P$
are {\em simplicial  bisimilar}, written $x_1 \sibis^{\calP} x_2$, if there exists a simplicial bisimulation $Z$ such
that $Z(x_1,x_2)$.\closedefi
\end{defi}

It has been shown that simplicial bisimilarity enjoys the classical
Hennessy-Milner property:
two points 
$x_1,x_2 \in P$
are simplicial
bisimilar if and only if they
satisfy the same \slcsG{} formulas, i.e. they are equivalent with
respect to the logic \slcsG{}, written $x_1 \slcsGeq^{\calP} x_2$.

The result has been extended to \emph{\plm-bisimilarity} on finite
poset models, a notion of bisimilarity based on \plm-paths: given
finite poset model $\calF=(W,\preccurlyeq,\peval{\calF})$, 
$w_1,w_2 \in W$ are \plm-bisimilar, written $x_1 \plmbis^{\calF} x_2$,
if and only if they satisfy the same \slcsG{} formulas, i.e.
$x_1 \slcsGeq^{\calF} x_2$ (see~\cite{Ci+23c} for details).  In
summary, we have:
\[
x_1 \sibis^{\calP} x_2 \mbox{ iff }
x_1 \slcsGeq^{\calP} x_2 \mbox{ iff }
\map(x_1) \slcsGeq^{\map(\calP)} \map(x_2) \mbox{ iff }
\map(x_1)\plmbis^{\map(\calP)} \map(x_2).
\]

 In Section~\ref{sec:WeakBis} we show a similar result for a {\em weaker}
  logic introduced in the next section, and originally presented
  in~\cite{Be+24}.
 Finally, in~\cite{Be+22} it has been shown that the classical modality $\Diamond$
 can be expressed using $\gamma$.  
 We recall  that
 for polyhedral model $\calP=(P,\peval{\calP})$
  and  for  poset model $\calF=(W,\preccurlyeq, \peval{\calF})$, the semantics of $\Diamond \form$ is defined as follows:
  \begin{displaymath}
    \begin{array}{l c l}
      \calP, x \models \Diamond \,\form
      & \Leftrightarrow
      & x \in \closure_T(\ZET{\,x'\in P}{\calP,x' \models \form\,})
    \end{array}
  \end{displaymath}
  \begin{displaymath}
    \begin{array}{l c l}
      \calF, w\models \Diamond \form
      & \Leftrightarrow
      & \text{$w' \in W$ exists such that $w \preccurlyeq w'$ and
        $\calF, w' \models \form$}.
    \end{array}
  \end{displaymath}
It turns out that $\Diamond \form$ is equivalent to $\gamma(\form,\ltrue)$, 
 for all \slcsG{} formulas $\form$.

We close this section with a small example. 
\begin{exa}
With reference to
Figure~\ref{subfig:PolyhedronNoPathCompressed}, we have that no red
point, call it~$y$, in the open segment~$CD$ is simplicial bisimilar
to the red point~$C$. In fact, although both $y$ and~$C$ satisfy
$\gamma(\mathbf{green}, \ltrue)$, we have that $C$~satisfies also
$\gamma(\mathbf{grey}, \ltrue)$, which is not the case for~$y$.
Similarly, with reference to
Figure~\ref{subfig:PolyhedronNoPathPosetCompressed}, cell~$\relint{C}$
satisfies $\gamma(\mathbf{grey}, \ltrue)$, which is not satisfied
by~$\relint{CD}$. 
\closeex
\end{exa}

\subsection{Labelled Transition Systems and Related Bisimilarities}
\begin{defi}\label{defi:lts}
A {\em labelled transition system}, \lts{} for short,
is a tuple $(S,L,\longrightarrow~\!\!)$ where $S$ is a non-empty set of {\em states},
$L$ is a non-empty set of {\em transition labels} and $\longrightarrow \subseteq S \times L \times S$ is the transition relation.\closedefi
\end{defi}

For $\tau \in L$ denoting the ``silent'' action we let $t \trans{\tau^*} t'$ whenever
$t=t'$ or there are $t_0, \ldots, t_n$, for $n>0$ such that 
$t_0 = t$,  $t_n = t'$ and $t_i \trans{\tau} t_{i+1}$ for $i\in [0;n)$.

\begin{defi}[Strong Bisimulation and Strong Equivalence]\label{def:StrongBisEq}
Given an \lts{} $\LTS= (S,L,\longrightarrow)$ a binary relation $B\subseteq S \times S$ is a {\em strong bisimulation} if, for all $s_1, s_2 \in S$, if $B(s_1,s_2)$ then the following holds:
\begin{enumerate}
\item if $s_1 \trans{\lambda} s'_1$ for some $\lambda$ and $s_1$, then $s'_2$ exists such that 
$s_2 \trans{\lambda} s'_2$  and $B(s'_1,s'_2)$, and
\item if $s_2 \trans{\lambda} s'_2$ for some $\lambda$ and $s_2$, then $s'_1$ exists such that 
$s_1 \trans{\lambda} s'_1$  and $B(s'_1,s'_2)$.
\end{enumerate}
We say that $s_1$ and $s_2$ are {\em strongly equivalent} in $\LTS$, written $s_1 \, \seq^{\LTS} s_2$ if a strong bisimulation $B$ exists such that $B(s_1,s_2)$.\closedefi
\end{defi}

It has been shown that $\sim^{\LTS}$ is the union of all strong bisimulations in $\LTS$, it is the largest strong bisimulation and it is an equivalence relation~\cite{Mil89}.

\begin{defi}[Branching Bisimulation and Equivalence]\label{def:BranBisEq}
Given an \lts{} $\LTS= (S,L,\longrightarrow)$ such that $\tau \in L$ a binary relation $B\subseteq S \times S$ is a {\em branching bisimulation} iff, for all $s, t,s' \in S$, and $\lambda\in L$, 
whenever $B(s,t)$ and  $s \trans{\lambda} s'$, it holds that:
(i) $B(s',t)$ and $\lambda=\tau$, or
(ii) $B(s,\bar{t}), B(s',t')$ and $t \trans{\tau^*} \bar{t}$, $\bar{t} \trans{\lambda} t'$, for some
$\bar{t},t' \in S$.\\[0.5em]
Two states $s,t \in S$ are called {\em branching bisimilar} in $\LTS$,  written 
$s \beq^{\LTS} t$ if $B(s,t)$ for some branching bisimulation $B$ for $S$. 
\closedefi
\end{defi}

It has been shown that $\beq^{\LTS}$ is the union of all branching bisimulations in $\LTS$, it is the largest branching bisimulation and it is an equivalence relation~\cite{GlW96}. 

We will omit the superscript $\LTS$ in $\seq^{\LTS}$ and $\beq^{\LTS}$ when this will not cause confusion.

 \section{Weak \slcs{} on Polyhedral Models}\label{sec:slcsE}
 
 In this section we introduce \slcsE, a logic for polyhedral models
 that is weaker than \slcsG, yet is still capable of expressing
 interesting conditional reachability properties. We present also an
 interpretation of the logic on finite poset models.

 \begin{defi}[Weak \slcs{} on polyhedral models - \slcsE{}]
   \label{def:SlcsEPolMod}
   The abstract language of \slcsE{} is the following:
   \begin{displaymath}
     \form ::= p \; \sep \; \lneg \form \; \sep \;
     \form_1 \land \form_2 \; \sep \; \eta(\form_1,\form_2). 
   \end{displaymath}
   The satisfaction relation of \slcsE{} with respect to a given
   polyhedral model 
$\calP = (P,\peval{\calP})$, with $P=|K|$ for some simplicial complex $K$,
\slcsE{}
   formula~$\form$, and point $x \in P$ is defined recursively on
   the structure of~$\form$ as follows:
   \begin{displaymath}
     \begin{array}{l c l}
       \calP, x \models p
       & \Leftrightarrow
       & x \in \peval{\calP}(p);
       \\
       \calP, x \models \lneg \form
       & \Leftrightarrow
       & \text{$\calP , x \models \form$ does not hold};
       \\
       \calP , x \models \form_1 \land \form_2
       & \Leftrightarrow
       & \text{$\calP, x \models \form_1$ and $\calP, x \models
         \form_2$};
       \\
       \calP, x \models \eta(\form_1,\form_2)
       & \Leftrightarrow
       & \text{a topological path $\pi : [0,1] \to P$ exists such
         that}
       \\
       &
       & \text{$\pi(0) = x$, $\calP , \pi(1) \models \form_2$, and
         $\calP , \pi(r) \models \form_1$ for all $r \in \!\![0,\!1)$}.
     \end{array}
   \end{displaymath}
   \vspace{-0.3in}\\\mbox{ }\hfill\closedefi
 \end{defi}

\begin{rem}
It is worth pointing out that the definition of the satisfaction relation of \slcsE{} does {\em not}
depend on the specific simplicial complex $K$ that generates the polyhedron $P=|K|$.
In other words: given polyhedral models $\calP'=(P,\peval{\calP'})$ with $P=|K'|$ and
$\calP''=(P,\peval{\calP''})$ with $P=|K''|=|K'|$ and
$\peval{\calP'} = \peval{\calP''}$, for all \slcsE{} formulas $\form$ and $x \in P$
 the following holds: $\calP',x \models \form$ iff $\calP'',x \models \form$.\closerem
\end{rem}

 As usual, disjunction~($\lor$) is derived as the dual of~$\land$.
 Note that the only difference between $\eta(\form_1,\form_2)$ and
 $\gamma(\form_1,\form_2)$ is that the former requires that \emph{also
   the first element} of a path witnessing the formula
 satisfies~$\form_1$, hence the use of the left closed interval
 $[0,1)$ here. 
 Although this might seem at first sight only a very
 minor difference,  it has considerable consequences:  
 $\eta$ cannot express $\Diamond$, which, instead, can be expressed in terms of $\gamma$ (see Remark~\ref{rem:NoDiamondInEtaCont} and Remark~\ref{rem:noDiamondInEta} below).
  
 \begin{defi}[\slcsE{} Logical Equivalence]
   \label{def:SlcsEeq}
   Given a polyhedral model 
$\calP = (P,\peval{\calP})$, with $P=|K|$ for some simplicial complex $K$,
   and $x_1, x_2 \in P$, we say that $x_1$ and~$x_2$ are \emph{logically
     equivalent}  with respect to \slcsE, written
   $x_1 \slcsEeq^{\calP} x_2$, if and only if, for all \slcsE{}
   formulas~$\form$, it holds that $\calP,x_1 \models \form$ if and
   only if $\calP,x_2 \models \form$.\closedefi
\end{defi}

In the following, we will refrain from indicating the model~$\calP$
explicitly as a superscript of $\slcsEeq^{\calP}$ when it is clear
from the context.
Below, we show that \slcsE{} can be encoded into \slcsG{} so that the
latter is at least as expressive as the former.

\begin{defi}
  \label{def:etga}
  We define the encoding~$\etga$ of \slcsE{} into \slcsG{} as follows:
  \begin{displaymath}
    \begin{array}{l c l}
      \etga(p) & = & p \\
      \etga(\lneg \form) & = & \lneg\etga(\form)
    \end{array}
    \quad\quad\quad\quad
    \begin{array}{l c l}
      \etga(\form_1 \land \form_2) & = &\etga(\form_1) \land \etga(\form_2) \\
      \etga(\eta(\form_1,\form_2)) & = & \etga(\form_1) \land \gamma(\etga(\form_1),\etga(\form_2)) 
    \end{array}
  \end{displaymath}
  \vspace{-0.3in}\\\mbox{ }\hfill\closedefi
\end{defi}

The following lemma is easily proven by structural induction on~$\form$ 
(see Appendix~\ref{apx:Prf:lem:etgaCorrectG}).

\begin{lem}
  \label{lem:etgaCorrectG}
  Let 
  $\calP = (P,\peval{\calP})$, with $P=|K|$ for some simplicial complex $K$,
  be a polyhedral model, $x \in P$,
  and $\form$ a \slcsE{} formula.  Then $\calP,x \models \form$ if and
  only if $\calP,x \models \etga(\form)$. \qed
\end{lem}

A direct consequence of Lemma~\ref{lem:etgaCorrectG} is that \slcsE{}
is weaker than \slcsG{.} 

\begin{prop}
  \label{prop:SLCSGimplWSLCSG}
  Let 
  $\calP = (P,\peval{\calP})$, with $P=|K|$ for some simplicial complex $K$,
   be a polyhedral model.
  For all $x_1,x_2 \in P$  the following holds:
  if $x_1 \slcsGeq{} x_2$ then $x_1 \slcsEeq x_2$. \qed
\end{prop}

\begin{rem}
  \label{rem:weakGweakerG}
  The converse of Proposition~\ref{prop:SLCSGimplWSLCSG} does
  \emph{not} hold, as shown by the  polyhedral model
  $\calP_{\ref{fig:AltTriAndPoset}}=(P_{\ref{fig:AltTriAndPoset}},\peval{\calP_{\ref{fig:AltTriAndPoset}}})$
  in Figure~\ref{fig:AlternatingTriangle}, 
  where $P_{\ref{fig:AltTriAndPoset}}$ is the simplex  $K_{\ref{fig:AltTriAndPoset}}$ generated by points $A$, $B$, and $C$, i.e. the triangle $ABC$,
  and $\peval{\calP_{\ref{fig:AltTriAndPoset}}}$ is specified by the colours in the figure.
  It is easy to see that, for all $x \in \relint{ABC}$, we have
  $A \not\slcsGeq{} x$ and $A \slcsEeq{} x$. Let, in fact,
  $x \in \relint{ABC}$. Clearly, $A \not\slcsGeq{} x$ since
  $\calP_{\ref{fig:AltTriAndPoset}}, A \models \gamma(\mathbf{red},\ltrue)$ whereas
  $\calP_{\ref{fig:AltTriAndPoset}}, x \not\models \gamma(\mathbf{red},\ltrue)$. It can 
  easily be shown, by induction on the structure of formulas, that
  $A \slcsEeq x$ for all $x \in \relint{ABC}$ 
  (see Appendix~\ref{apx:prf:rem:weakGweakerG}).
  As an additional, a bit more complex, example, let us consider the polyhedral model 
  $\calP_{\ref{fig:PolyhedronNoPathCompressed}}$ of Figure~\ref{fig:PolyhedronNoPathCompressed}. It is easy to see that
  every $x \in \relint{CE}$ satisfies $\gamma(\mathbf{green},\ltrue)$, 
  while for no $y \in \relint{DEF}$ we have
  $\calP_{\ref{fig:PolyhedronNoPathCompressed}}, y \models \gamma(\mathbf{green},\ltrue)$. 
  So, for  all such $x$ and $y$, we have $x \, \not\slcsGeq \, y$.
  On the other hand,
  as we will see in Example~\ref{ex:min} of Section~\ref{sec:WeakBis} (on page~\pageref{ex:min}), 
  cells $\relint{CE}$ and $\relint{DEF}$ will fall in the same equivalence class of
  $\slcsEeq$ on $\map(\calP_{\ref{fig:PolyhedronNoPathCompressed}})$ and so, by Theorem~\ref{theo:calMPresForm} below --- guaranteeing that 
  \slcsE{} is preserved and reflected by mapping $\map$ --- and Theorem~\ref{theo:MinE} of Section~\ref{sec:EtaMinimisation} --- stating correctness of \plm-minimisation --- we get that 
  $x \, \slcsEeq \, y$.
  The above reasoning can be generalised to any pair of points 
  $x \in \relint{D} \, \cup \, \relint{E} \, \cup \, \relint{CE} \, \cup \relint{DE}$ and 
  $y \in \relint{F} \, \cup \, \relint{DF} \, \cup \, \relint{EF} \, \cup \, \relint{DEF}$: we have
  $x \, \slcsEeq \, y$ but $x \, \not\slcsGeq \, y$.
  \closerem
\end{rem}

\begin{rem}\label{rem:NoDiamondInEtaCont}
  The example of Figure~\ref{fig:AlternatingTriangle} is useful also
  for showing that the classical topological interpretation of the
  modal logic operator $\Diamond$ cannot be expressed in \slcsE.
 Clearly, in the model of the figure, we have
  $\calP_{\ref{fig:AltTriAndPoset}},A \models \Diamond \mkern1mu \mathbf{red}$ while
  $\calP_{\ref{fig:AltTriAndPoset}}, x \models \Diamond \mkern1mu \mathbf{red}$ for
  no~$x \in \relint{ABC}$. On the other hand, $A \slcsEeq x$ holds for
  all $x \in \relint{ABC}$, as we have just seen in
  Remark~\ref{rem:weakGweakerG}. So, if $\Diamond$ were expressible in
  \slcsE, then $A$ and~$x$ should have agreed on
  $\Diamond \mkern1mu \mathbf{red}$ for each $x \in
  \relint{ABC}$. 
  \closerem
\end{rem}

\begin{figure}
\begin{center}
\subfloat[]{\label{fig:AlternatingTriangle}
\resizebox{0.9in}{!}
{
\begin{tikzpicture}[scale=2]%
    \tikzset{node distance=1.5cm}
    \tikzstyle{point}=[circle,fill=white,inner sep=0pt,minimum width=4pt,minimum height=4pt]	  

	    \node (tC)[point, draw=blue!35, fill=blue!35, label={90:$C$}] at (0,0){};
	    \node (tA)[point, draw=blue!35, fill=blue!35, below of = tC, xshift=-1cm,label={270:$A$}]{};
	    \node (tB)[point, draw=blue!35, fill=blue!35,below of = tC, xshift=1cm, label={270:$B$}]{};
	   
	    \fill [fill=blue!35](tA.center) -- (tB.center) -- (tC.center);
	    \draw [line width=0.6mm, red!60] (tA) -- (tB);	
	    \draw [line width=0.6mm, red!60] (tA) -- (tC);	
	    \draw [line width=0.6mm, red!60] (tB) -- (tC);
 \end{tikzpicture}
 }
 }\quad\quad\quad\quad\quad\quad
 \subfloat[]{\label{fig:PosetAlternatingTriangle}
 \resizebox{1.5in}{!}
 {
 \begin{tikzpicture}[scale=2, every node/.style={transform shape}]
    \tikzstyle{kstate}=[rectangle,draw=black,fill=white]
    \tikzset{->-/.style={decoration={
        markings,
        mark=at position #1 with {\arrow{>}}},postaction={decorate}}}
    \node[kstate,fill=blue!35] (A) at (  0,0) {$\relint{A}$};
    \node[kstate,fill=blue!35, right of= A, xshift=1cm] (C) {$\relint{C}$};
    \node[kstate,fill=blue!35, right of= C, xshift=1cm] (B) {$\relint{B}$};
    \node[kstate,fill=red!60, above of= A] (AC) {$\relint{AC}$};
    \node[kstate,fill=red!60, right of= AC, xshift=1cm] (AB) {$\relint{AB}$};
    \node[kstate,fill=red!60, right of= AB, xshift=1cm] (BC) {$\relint{BC}$};
    \node[kstate,fill=blue!35, above of= AB] (ABC) {$\relint{ABC}$};

    \draw(A) edge[thick](AC);
    \draw(A) edge[thick](AB);
    \draw(C) edge[thick](AC);
    \draw(C) edge[thick](BC);
    \draw(B) edge[thick](AB);
    \draw(B) edge[thick](BC);
    \draw(AC) edge[thick](ABC);
    \draw(AB) edge[thick](ABC);
    \draw(BC) edge[thick](ABC);
\end{tikzpicture}
}
}
\end{center}
\caption{A polyhedral model (\ref{fig:AlternatingTriangle}) $\calP_{\ref{fig:AltTriAndPoset}}$,
and the Hasse diagram of its cell poset model (\ref{fig:PosetAlternatingTriangle}).\label{fig:AltTriAndPoset}} 
\end{figure}

Below, we re-interpret \slcsE{} on finite Kripke models instead of
polyhedral models. The only difference from
Definition~\ref{def:SlcsEPolMod} is, of course, the fact that
$\eta$-formulas are defined using \plm-paths instead of topological
ones.

\begin{defi}[\slcsE{} on finite Kripke models]
  \label{def:WSlcsFinPos}
  The satisfaction relation of \slcsE{} with respect to a given finite
  Kripke model $\calK = (W,R,\peval{\calK})$,
  an \slcsE{} formula~$\form$, and an element $w \in W$, is defined
  recursively on the structure of $\form$: \\
  \begin{displaymath}
    \begin{array}{l c l}
      \calK, w \models p
      & \Leftrightarrow
      & w \in \peval{\calK}(p);
      \\
      \calK, w \models \lneg \form
      & \Leftrightarrow
      & \calK, w \not\models \form;
      \\
      \calK, w \models \form_1 \land \form_2
      & \Leftrightarrow
      & \text{$\calK, w \models \form_1$ and $\calK, w \models
        \form_2$};
      \\
      \calK, w \models \eta(\form_1,\form_2)
      & \Leftrightarrow
      & \text{a \plm-path $\pi : [0;\ell] \to W$ exists such that}
      \\
      &
      & \text{$\pi(0) = w$, $\calK, \pi(\ell) \models \form_2$,
        and $\calK, \pi(i) \models \form_1$ for all $i \in [0;\ell)$}.
    \end{array}
  \end{displaymath}
\vspace{-0.3in}\\\mbox{ }\hfill\closedefi		
\end{defi}

\begin{rem}\label{rem:KripkeNotPoset}
We recall here that \plm-paths are defined on general Kripke frames, of which 
finite posets are a subclass. 
The reason why in Definition~\ref{def:WSlcsFinPos} we use finite Kripke models, instead 
of restricting it to finite poset models, stems from the fact that the result of minimisation 
of a finite poset model, modulo weak \plm-bisimilarity,  
is, in general, not guaranteed to be again a poset model, whereas it is 
guaranteed to be a (reflexive) finite Kripke model.
As we will see in Section~\ref{sec:EtaMinimisation}, the fact that the minimal model is not necessarily a poset model does not affect correctness of the minimisation procedure, 
and so it does not constitute a problem for the optimised model checking method
presented in this paper. In the rest of this section, as well as in Section~\ref{sec:WeakBis}, we will 
anyway be interested in poset models, so that we will restrict the relevant results to the latter.
\closerem
\end{rem}

The following result,
proven in Appendix~\ref{apx:prf:prop:interchangeableE},
states that to evaluate an \slcsE{} formula
$\eta(\form_1,\form_2)$ in a poset model, it does not matter whether
one considers \plm-paths or \dwn-paths.

\begin{prop}
  \label{prop:interchangeableE}
  Given a finite poset model $\calF = (W,\preccurlyeq,\peval{\calF})$,
  $w \in W$, and  \slcsE{} formulas $\form_1$ and $\form_2$, the
  following statements are equivalent:
  \begin{enumerate}
  \item\label{enu:plmE}
    There exists a \plm-path
    $\pi : [0;\ell] \to W$ for some~$\ell$ with $\pi(0) = w$,
    $\calF, \pi(\ell)\models \form_2$, and
    $\calF, \pi(i)\models \form_1$ for all $i \in [0;\ell)$.
  \item\label{enu:dwnE}
    There exists a \dwn-path $\pi : [0;\ell'] \to W$ for some~$\ell'$ with
    $\pi(0) = w$, 
    $\calF, \pi(\ell')\models \form_2$, and 
    $\calF, \pi(i)\models \form_1$ for all $i \in [0;\ell')$.
    \qed
  \end{enumerate}
\end{prop}

\begin{defi}[Logical Equivalence]
  \label{def:WFPSLCSeq}
  Given a finite poset model
  $\calF = (W,{\preccurlyeq},\peval{\calF})$ and elements
  $w_1, w_2 \in W$ we say that $w_1$ and~$w_2$ are \emph{logically
    equivalent} with respect to \slcsE, written
  $w_1 \slcsEeq^{\calF} w_2$, if and only if, for all \slcsE{}
  formulas~$\form$, it holds that $\calF,w_1 \models \form$ if
  and only if $\calF,w_2 \models \form$.
  \closedefi
\end{defi}

Again, in the following, we will refrain from indicating the model
$\calF$ explicitly in~$ \slcsEeq^{\calF}$ when it is clear from the
context.
It is useful to define a ``characteristic'' \slcsE{} formula~$\chi(w)$
that is satisfied by all and only those elements~$w'$ with
$w' \slcsEeq w$,
as shown in Appendix~\ref{apx:prf:prop:chiE}.

\begin{defi}
  \label{def:chiE}
  Given a finite poset model $\calF = (W,\preccurlyeq,\peval{\calF})$,
  $w_1,w_2 \in W$, define \slcsE{} formula~$\delta_{w_1,w_2}$ as
  follows: if $w_1 \slcsEeq w_2$, then set $\delta_{w_1,w_2}= \ltrue$,
  otherwise pick some \slcsE{} formula~$\psi$ such that
  $\calF,w_1 \models \psi$ and $\calF,w_2 \models \lneg \psi$, and set
  $\delta_{w_1,w_2}= \psi$.  For $w \in W$ define
  $\chi(w) = \bigwedge_{w' \in W} \: \delta_{w,w'}$.  \closedefi
\end{defi}

\begin{prop}\label{prop:chiE} 
  Given a finite poset model $\calF = (W,\preccurlyeq,\peval{\calF})$,
  for $w_1,w_2 \in W$, it holds that $ \calF,w_2 \models \chi(w_1)$ if
  and only if $w_1 \slcsEeq w_2$. \qed
\end{prop}

The following lemma is the poset model counterpart of
Lemma~\ref{lem:etgaCorrectG} 
(see Appendix~\ref{apx:prf:lem:etgaCorrectE}):

\begin{lem}
  \label{lem:etgaCorrectE}
  Let $\calF = (W,\preccurlyeq,\peval{\calF})$ be a finite poset
  model, $w \in W$, and $\form$ an \slcsE{} formula. Then
  $\calF,w \models \form$ if and only if
  $\calF,w \models \etga(\form)$.
  \qed
\end{lem}

Thus we get, as for the interpretation on polyhedral models, that 
\slcsE{} on finite poset models is weaker than \slcsG:

\begin{prop}
  \label{prop:SLCSPMimplWSLCSPM}
  Let $\calF = (W,\preccurlyeq, \peval{\calF})$ be a finite poset model. 
  For all $w_1,w_2 \in W$ the following holds:
  if  $w_1 \slcsGeq{} w_2$ then $w_1 \slcsEeq w_2$.
  \qed
\end{prop}

\begin{rem}\label{rem:weakPMeakerPM}
  As expected, the converse of
  Proposition~\ref{prop:SLCSPMimplWSLCSPM} does not hold, as shown by
  the poset model $\map(\calP_{\ref{fig:AltTriAndPoset}})$ of
  Figure~\ref{fig:PosetAlternatingTriangle}. Clearly,
  $\relint{A} \not\slcsGeq{} \relint{ABC}$. In fact
  $\map(\calP_{\ref{fig:AltTriAndPoset}}),\relint{A} \models \gamma(\mathbf{red},\ltrue)$ whereas
  $\map(\calP_{\ref{fig:AltTriAndPoset}}),\relint{ABC} \not\models \gamma(\mathbf{red},\ltrue)$. On the
  other hand, it can be easily shown, by induction on the structure of
  formulas, that $\relint{A} \slcsEeq \relint{ABC}$
  (see Appendix~\ref{apx:prf:rem:weakPMeakerPM}).
  With reference to the polyhedral model 
  $\calP_{\ref{fig:PolyhedronNoPathCompressed}}$ of Figure~\ref{fig:PolyhedronNoPathCompressed}, its poset model 
  $\map(\calP_{\ref{fig:PolyhedronNoPathCompressed}}) = 
  (W_{\ref{fig:PolyhedronNoPathCompressed}}, \preccurlyeq, \peval{\map(\calP_{\ref{fig:PolyhedronNoPathCompressed}})})$, and 
  Example~\ref{ex:min} of Section~\ref{sec:WeakBis}, 
  we have that 
  $\relint{D},\relint{E},\relint{F},\relint{CE},\relint{DE},\relint{DF},\relint{EF}$, and
$\relint{DEF}$ are all equivalent according to weak \plm-bisimilarity. We invite the reader to check that, letting $\phi_0, \phi_1, \phi_2, \psi_1, \psi_2, \psi_3$, and $\psi_4$ be defined as\\[0.5em]
\begin{minipage}{2in}
$\phi_0=\gamma(\mathbf{green},\ltrue)$\\
$\phi_1=\gamma(\lneg\phi_0,\ltrue)$\\
$\phi_2=\gamma(\phi_0 \land \lneg\phi_1,\ltrue)$\\
$ $ 
\end{minipage}
\begin{minipage}{2in}
$\psi_1= \lneg\phi_0$\\
$\psi_2= \phi_0 \land \lneg \phi_1$\\
$\psi_3=\phi_1 \land \phi_1 \land \lneg\phi_2$\\ 
$\psi_4=\phi_0 \land \phi_1 \land \phi_2$
\end{minipage}\\
we have\\[0.5em]
$\calP_{\ref{fig:PolyhedronNoPathCompressed}},\relint{DF} \models \lneg\phi_0$, and the same holds for $\relint{DEF},\relint{EF}$ and $\relint{F}$,\\
$\calP_{\ref{fig:PolyhedronNoPathCompressed}},\relint{CE}\models \lneg\phi_0 \land \lneg\phi_1$,\\
$\calP_{\ref{fig:PolyhedronNoPathCompressed}},\relint{DE}\models \lneg\phi_0 \land\phi_1 \land \lneg\phi_2$, and\\
$\calP_{\ref{fig:PolyhedronNoPathCompressed}},\relint{E}\models \lneg\phi_0 \land \phi_1 \land \phi_2$.\\[0.5em]
As a consequence,
each of $\psi_1, \psi_2, \psi_3$, and $\psi_4$ cannot be true in conjunction with any of the others
and so, the classes $\SET{\relint{DF}, \relint{DEF},\relint{EF},\relint{F}}$, 
$\SET{\relint{CE}}$, $\SET{\relint{DE}}$, and $\SET{\relint{E}}$ must definitely be distinct in 
the quotient of $W_{\ref{fig:PolyhedronNoPathCompressed}}$ modulo $\slcsGeq$.
\closerem
\end{rem}

\begin{rem}
  \label{rem:noDiamondInEta}
  As for the case of the continuous interpretation of \slcsE, the example
  of Figure~\ref{fig:PosetAlternatingTriangle} is useful also for
  showing that the classical modal logic operator $\Diamond$ cannot be
  expressed in \slcsE{.} 
  Clearly, in the model of the figure, we have
  $\map(\calP_{\ref{fig:AltTriAndPoset}}), \relint{A} \models \Diamond \mkern1mu \mathbf{red}$ while
  $\map(\calP_{\ref{fig:AltTriAndPoset}}), \relint{ABC} \not\models \Diamond \mkern1mu \mathbf{red}$. On the
  other hand $\relint{A} \slcsEeq \relint{ABC}$ holds, as we have just
  seen in Remark~\ref{rem:weakPMeakerPM}. So, if $\Diamond$ were
  expressible in \slcsE, then $\relint{A}$ and $\relint{ABC}$ should
  have agreed on $\Diamond \mkern1mu \mathbf{red}$.   
  \closerem
\end{rem}

The following result, 
proven in Appendix~\ref{apx:prf:lem:VB},
is useful to set up a bridge between the
continuous and the discrete interpretations of \slcsE{.}

\begin{lem}
  \label{lem:VB}
  Given a polyhedral model 
  $\calP = (P,\peval{\calP})$, with $P=|K|$ for some simplicial complex $K$,
  for all~$x \in P$
  and formulas~$\form$ of \slcsE{}
  the following holds: $\calP,x \models \form$ if and only if
  $\map(\calP),\map(x) \models \etga(\form)$.
  \qed
\end{lem}

As a direct consequence of Lemma~\ref{lem:etgaCorrectE} and
Lemma~\ref{lem:VB} 
we get, 
by Theorem~\ref{theo:calMPresForm} below, 
proven in~\ref{apx:prf:theo:calMPresForm},
the bridge between the continuous and the
discrete interpretations of \slcsE:

\begin{thm}
  \label{theo:calMPresForm}
  Given a polyhedral model 
  $\calP = (P,\peval{\calP})$, with $P=|K|$ for some simplicial complex $K$,
for all~$x \in P$
 and formulas $\form$ of \slcsE{} it holds that:
  $\calP,x \models \form$ if and only if
  $\map(\calP),\map(x) \models \form$.
  \qed
\end{thm}

This theorem allows one to go back and forth between the polyhedral
model and the corresponding poset model without losing anything
expressible in \slcsE.

\section{Weak Simplicial Bisimilarity}
\label{sec:WeakBis}

In this section, we introduce weak versions of simplicial bisimilarity
and \plm-bisimilarity and we show that they coincide with logical
equivalence induced by \slcsE{} in polyhedral and poset models,
respectively.
We are looking for a notion of bisimilarity that enjoys the HMP with respect to 
\slcsE, i.e. that coincides with  $\slcsEeq$. 
We already know that simplicial bisimilarity $\sibis$ enjoys the HMP
with respect \slcsG, i.e. $\sibis \, = \, \slcsGeq$ and, moreover, that 
$\slcsEeq$ is weaker than $\slcsGeq$. 
Here, by ``weaker'' we mean coarser, i.e. one that includes simplicial bisimilarity, in the sense of set inclusion, $\slcsGeq \, \subset \, \slcsEeq$.

A natural step in the search for such a notion of bisimilarity is to reconsider the 
definition of simplicial bisimilarity, recalled in Section~\ref{subsec:slcsG} 
(see Definition~\ref{def:sibis}), and 
seek to weaken its conditions.
Of course, the first condition cannot be relaxed in any meaningful way:
equivalent points must at least satisfy the same predicate letters. 
Let us thus focus on the second
condition, namely the one concerning topological paths. 
The condition requires 
that as ``one moves on'' $\pi_2$ using cursor $r$,
the corresponding point on $\pi_1$, i.e. $\pi_1(r)$, must be related by $Z$ to the current point in $\pi_2$, namely $\pi_2(r)$. The points in $\pi_2$ and $\pi_1$, while one moves the cursor $r$, must go ``hand in hand'' in $Z$. 

One way of relaxing the above condition is to require only that (2.a) the ending points of 
$\pi_1$ and $\pi_2$ are related --- i.e. $Z(\pi_1(1),\pi_2(1))$ --- and (2.b) for each
other point $y_2$ of $\pi_2$, there is a  point $y_1$ of $\pi_1$, different from $\pi_1(1)$, such that 
$y_1$ and $y_2$ are related --- i.e. for each $r_2 \in [0,1)$ there is $r_1 \in [0,1)$ such that
$Z(\pi_1(r_1), \pi_2(r_2))$.

Interestingly, it turns out that the bisimilarity induced by a definition of bisimulation relation
where condition (2) is relaxed as above, coincides exactly with $\slcsEeq$, the
logical equivalence induced by \slcsE! 
In practice, we do not even need the notion of simplicial path, in the sense that the actual
definition, given below, is based on general topological paths and characterises an 
equivalence relation --- which we call  {\em weak simplicial bisimilarity},   
$\wsibis$ --- that  coincides with  $\slcsEeq$, as guaranteed by Theorem~\ref{thm:HMPpoly}. The proof of this theorem, as well as those of all  results related to $\wsibis$, does not require the use of simplicial paths.

\begin{defi}[Weak Simplicial Bisimulation]
  \label{def:WSimBis}
  Given a polyhedral model 
  $\calP = (P,\peval{\calP})$, with $P=|K|$ for some simplicial complex $K$,
  a symmetric relation $Z \subseteq {|K| \mathord\times |K|}$ is a \emph{weak simplicial
    bisimulation} if, for all $x_1,x_2 \in |K|$, whenever
  $Z(x_1,x_2)$, it holds that:
  \begin{enumerate}
  \item $\invpeval{\calP}(\SET{x_1}) = \invpeval{\calP}(\SET{x_2})$;
  \item for each topological path~$\pi_1$ from~$x_1$, there is a
    topological path~$\pi_2$ from~$x_2$ such that
    $Z(\pi_1(1),\pi_2(1))$ and for all $r_2 \in [0,1)$ there is
    $r_1 \in [0,1)$ such that $Z(\pi_1(r_1),\pi_2(r_2))$.
  \end{enumerate}
  Two points 
  $x_1,x_2 \in P$
  are weakly simplicial bisimilar,
  written $x_1 \wsibis^{\calP} x_2$, if there is a weak simplicial
  bisimulation $Z$ such that $B(x_1,x_2)$.  \closedefi
\end{defi}

\begin{exa}\label{ex:wbis}
With reference to Figure~\ref{fig:AlternatingTriangle},
the binary relation $Z$ composed of all those pairs of points that have the same colour, i.e.
$$
Z= \left(\relint{AB} \, \cup \, \relint{BC}  \, \cup \, \relint{AC}\right)^2
\cup
\left(ABC^2 \, \setminus \, \left(\relint{AB} \, \cup \, \relint{BC}  \, \cup \, \relint{AC}\right)^2\right)
$$
is a weak simplicial bisimulation.
Take, for example, any pair $(x,y) \in \relint{AB} \times \relint{BC}$: both $x$ and~$y$ satisfy  only one predicate letter, namely $\mathbf{red}$. In addition, let $\pi_x$ be any topological path starting from $x$ and such that $\pi_x(1)$ is red. Then it is easy to see, just by visual inspection, that one can find a path $\pi_y$ from $y$ such that $\pi_y(1)$ is red and, for each
intermediate point of $\pi_y$ there is in $\pi_x$ an intermediate point of the same colour.
The reasoning for the case in which $\pi_x(1)$ is blue is similar.
Thus $x \, \wsibis y$. The reasoning can be extended to all pairs in~$Z$: actually $\wsibis$ coincides with $Z$ for the polyhedral model of Figure~\ref{fig:AlternatingTriangle}.

As an additional example, let us consider the polyhedral model $\calP_{\ref{fig:PolyhedronNoPathCompressed}}$ of Figure~\ref{subfig:PolyhedronNoPathCompressed} and points $A$ and $D$ therein.
It is easy to see that there is no weak simplicial bisimulation $Z$ such that $Z(A,D)$.
Suppose such a $Z$ exists. Take $\pi_1$ from $D$ such that, 
$\pi_1(r)=D$ for all $r\in [0,\bar{r}]$, and 
$\emptyset \subset \pi_1((\bar{r},1]) \subset \relint{CDE}$, for some $\bar{r} \in (0,1)$.
Clearly, any $\pi_2$ from $A$ should be such that $\pi_2(1) \in \relint{CDE}$,
otherwise $Z(\pi_1(1),\pi_2(1))$ would not hold. But any topological path starting from $A$ and ending in
$\relint{CDE}$ would necessarily pass by red points, and for any such red point 
$\pi_2(r_2)$ for some $r_2 \in (0,1)$ there would be no $r_1 \in (0,1)$ such that 
$Z(\pi_1(r_1),\pi_2(r_2))$, since no point of $\pi_1$ is red.
As one would expect, we have also that
$\calP_{\ref{fig:PolyhedronNoPathCompressed}}, D \models \eta(\mathbf{green}\, \lor \, \mathbf{grey}, \mathbf{green})$
whereas 
$\calP_{\ref{fig:PolyhedronNoPathCompressed}}, A \not\models \eta(\mathbf{green}\, \lor \, \mathbf{grey}, \mathbf{green})$. \closeex
\end{exa}

Definition~\ref{def:WPLMBis} below rephrases Definition~\ref{def:WSimBis} for finite posets and discrete paths and it settles the finite poset counterpart of weak simplicial bisimilarity, 
namely weak \plm-bisimilarity, a weaker version of \plm-bisimilarity introduced in~\cite{Ci+23c}.
The second condition in the definition deals with \dwn-paths. In particular, 
for a weak \plm-bisimulation $Z$ on a poset model, it is required that,
for all nodes $w_1,w_2$ of the  poset, whenever $Z(w_1,w_2)$, 
for each \dwn-path $\pi_1 = (w_1,u_1,d_1)$, 
there is a \plm-path\footnote{Recall that \plm-paths are a subclass of \dwn-paths.} $\pi_2$ from $w_2$ of some length $\ell_2\geq 2$ such that
(ii.a) the ending elements of $\pi_1$ and $\pi_2$ are related --- i.e. $Z(d_1,\pi_2(\ell_2))$ ---
and (ii.b) for each other element $v_2$ of $\pi_2$ there is an element $v_1$ of $\pi_1$, different from
$\pi_1(2)$, such that  $v_1$ and $v_2$ are related --- i.e. for all $j\in [0;\ell_2)$, there is
$i\in [0;2)$ such that $Z(\pi_1(i), \pi_2(j))$. In other words, since 
$\pi_1(0)=w_1$ and $\pi_1(1)=u_1$, it is required that 
$Z(w_1,\pi_2(j))$ or $Z(u_1,\pi_2(j))$ holds for all $j\in [0;\ell_2)$. Note that it is sufficient to consider \dwn-paths of length $2$ starting from $w_1$. As shown by Theorem~\ref{theo:PMbisEqSLCSEeq}, the resulting relation $\wsibis$ coincides with
$\slcsEeq$.

\begin{defi}[Weak \plm-bisimulation]
  \label{def:WPLMBis}
  Given a finite poset model $\calF = (W,\preccurlyeq,\peval{\calF})$,
  a symmetric binary relation $Z \subseteq {W \mathord\times W}$ is a
  weak \plm-bisimulation if, for all $w_1,w_2 \in W$, whenever
  $Z(w_1,w_2)$, it holds that:
  \begin{enumerate}
  \item $\invpeval{\calF}(\SET{w_1}) = \invpeval{\calF}(\SET{w_2})$;
  \item for each $u_1,d_1 \in W$ such that
    $w_1 \dircnv{\preccurlyeq} u_1 \succcurlyeq d_1$ there is a
    \plm-path $\pi_2:[0;\ell_2] \to W$ from~$w_2$ such that
    $Z(d_1,\pi_2(\ell_2))$ and, for all $j \in [0;\ell_2)$, it holds
    that $Z(w_1,\pi_2(j))$ or $Z(u_1,\pi_2(j))$.
  \end{enumerate}
  We say that $w_1$ is weakly \plm-bisimilar to $w_2$, written
  $w_1 \wplmbis^{\calF} w_2$ if there is a weak \plm-bisimulation $Z$
  such that $Z(w_1,w_2)$. \closedefi
\end{defi}

For example, all red cells in the Hasse diagram of
Figure~\ref{fig:PosetAlternatingTriangle} are weakly \plm-bisimilar
and all blue cells are weakly \plm-bisimilar.

The following lemma shows that, in a polyhedral model $\calP$, weak
simplicial bisimilarity~$\wsibis^{\calP}$, as given by
Definition~\ref{def:WSimBis}, is stronger than $\slcsEeq$ -- logical
equivalence with respect to\ 
\slcsE: 

\begin{lem}
  \label{lem:WSBisIMPLEtaLog}
  Given a polyhedral model
  $\calP = (P,\peval{\calP})$, with $P=|K|$ for some simplicial complex $K$,  for all
  $x_1,x_2 \in P$,
  the following holds: if $x_1 \wsibis^{\calP} x_2$
  then $x_1 \slcsEeq x_2$. \qed
\end{lem}

\begin{proof}
  By induction on the structure of the formulas. We consider only the
  case $\eta(\form_1, \form_2)$. Suppose $x_1 \wsibis x_2$ and
  $\calP,x_1 \models \eta(\form_1, \form_2)$. Then there is a
  topological path~$\pi_1$ from~$x_1$ such that
  $\calP,\pi_1(1)\models \form_2$ and
  $\calP,\pi_1(r_1)\models \form_1$ for all $r_1\in [0,1)$. Since
  $x_1 \wsibis x_2$, then there is a topological path~$\pi_2$
  from~$x_2$ such that $\pi_1(1) \wsibis \pi_2(1)$ and for each
  $r_2 \in [0,1)$ there is $r'_1 \in [0,1)$ such that
  $\pi_1(r'_1) \wsibis \pi_2(r_2)$. By the Induction Hypothesis, we
  get $\calP, \pi_2(1)\models \form_2$
  and, for each $r_2 \in [0,1)$ $\calP, \pi_2(r_2) \models
  \form_1$. Thus $\calP,x_2 \models \eta(\form_1, \form_2)$.
\end{proof}

Furthermore, logical equivalence induced by \slcsE{} is stronger than
weak simplicial-bisimilarity, 
as implied by Lemma~\ref{lem:EtaLogIsWSB} below, which uses the
following auxiliary lemmas,
proven in Appendix~\ref{apx:prf:lem:dPExists}, Appendix~\ref{apx:prf:lem:dPtoTP},
 and Appendix~\ref{apx:prf:lem:ExisDwnPath} respectively.

\begin{lem}
  \label{lem:dPExists}
  Given a finite poset model $\calF = (W,\preccurlyeq,
  \peval{\calF})$ and weak \plm-bisimulation $Z \subseteq {W
    \mathord\times W}$, for all $w_1,w_2$ such that
  $Z(w_1,w_2)$, the following holds: for each \dwn-path $\pi_1 :
  [0;k_1] \to W$ from~$w_1$ there is a \dwn-path $\pi_2:[0;k_2] \to
  W$ from~$w_2$ such that $Z(\pi_1(k_1),\pi_2(k_2))$ and, for each $j
  \in [0;k_2)$, exists $i \in
  [0;k_1)$ such that $Z(\pi_1(i),\pi_2(j))$. \qed
\end{lem}

\begin{lem}
  \label{lem:dPtoTP}
  Given a polyhedral model 
  $\calP = (P,\peval{\calP})$, with $P=|K|$ for some simplicial complex $K$,
  and associated cell poset model
  $\map(\calP) = (W,\preccurlyeq,\peval{\map(\calP)})$, for
  any 
  \dwn-path $\pi : [0;\ell] \to W$, 
  there is a topological path $\pi' : [0,1] \to |K|$ such that:
  (i)~$\map(\pi'(0)) = \pi(0)$, (ii)~$\map(\pi'(1)) = \pi(\ell)$, and
  (iii)~for all $r \in (0,1)$ exists $i < \ell$ such that
  $\map(\pi'(r)) = \pi(i)$.\qed
\end{lem}

  \begin{lem}
    \label{lem:ExisDwnPath}
    Given a polyhedral model 
    $\calP = (P,\peval{\calP})$, with $P=|K|$ for some simplicial complex $K$,
    and associated cell poset model
    $\map(\calP) = (W,\preccurlyeq,\peval{\map(\calP)})$, for any
    topological path $\pi : [0,1] \to |K|$ the following holds:
    $\map(\pi([0,1]))$ is a connected subposet of~$W$ and there are
    $k>0$ and a \dwn-path $\hat\pi: [0;k] \to W$ from~$\map(\pi(0))$
    to $\map(\pi(1))$ such that, for all $i \in [0;k)$, $r \in [0,1)$
    exists with $\hat\pi(i) = \map(\pi(r))$. \qed
\end{lem}

\begin{lem}\label{lem:EtaLogIsWSB}
  In a given polyhedral model 
  $\calP = (P,\peval{\calP})$, with $P=|K|$ for some simplicial complex $K$,
  it holds that $\slcsEeq$ is a weak simplicial bisimulation.
\end{lem}

\begin{proof}
  Let $x_1,x_2 \in |K|$ be such that $x_1 \slcsEeq x_2$. The first
  condition of Definition~\ref{def:WSimBis} is clearly satisfied since
  $x_1 \slcsEeq x_2$. Suppose $\pi_1$~is a topological path
  from~$x_1$. By Lemma~\ref{lem:ExisDwnPath},
  $\map(\pi_1([0,1]))$~is a connected subposet of~$\relint{K}$ and
  a \dwn-path $\hat\pi_1:[0;k_1] \to \relint{K}$ 
  from~$\map(\pi_1(0))$
  to~$\map(\pi_1(1))$ exists such that, for all $i \in [0;k_1)$, $r_1
  \in [0,1)$ exists with $\hat\pi_1(i) =
  \map(\pi_1(r_1))$. We also know that $\map(x_1) \slcsEeq
  \map(x_2)$, as a consequence of Theorem~\ref{theo:calMPresForm},
  since $x_1 \slcsEeq
  x_2$. In addition, due to Lemma~\ref{lem:EtaLogIsWFpBis} below, we
  also know that $\map(x_1) \wplmbis
  \map(x_2)$. By Lemma~\ref{lem:dPExists}, we get that there is a
  \dwn-path $\hat\pi_2:[0;k_2] \to \relint{K}$
  such that $\hat\pi_1(k_1) \slcsEeq \hat\pi_2(k_2)$ and, for each
  $j \in [0;k_2)$, $i \in [0;k_1)$ exists such that
  $\hat\pi_1(i) \slcsEeq \hat\pi_2(j)$. By Lemma~\ref{lem:dPtoTP}, it
  follows that there is topological path~$\pi_2$ from~$x_2$ satisfying
  the three conditions of the lemma and, again by
  Theorem~\ref{theo:calMPresForm}, we have that
  $\pi_2(1) \slcsEeq \pi_1(1)$. In addition, for any $r_2 \in [0,1)$,
  since $\map(\pi_2(r_2)) = \hat\pi_2(j)$ for $j \in [0;k_2)$
  (condition~(ii) of Lemma~\ref{lem:dPtoTP}) there is $i \in [0;k_1)$
  such that $\hat\pi_1(i) \slcsEeq \hat\pi_2(j)$.  Finally, by
  construction, there is $r_1 \in [0,1)$ such that
  $\map(\pi_1(r_1)) = \hat\pi_1(i)$. By
  Theorem~\ref{theo:calMPresForm}, we arrive at
  $\pi_1(r_1) \slcsEeq \pi_2(r_2)$.
\end{proof}

On the basis of Lemma~\ref{lem:WSBisIMPLEtaLog} and
Lemma~\ref{lem:EtaLogIsWSB}, we have that the largest weak simplicial
bisimulation exists, it is a weak simplicial bisimilarity, it is an
equivalence relation, and it coincides with logical equivalence in the
polyhedral model induced by \slcsE{,} thus establishing the HMP for
$\wsibis^{\calP}$ with respect to~\slcsE:

\begin{thm}
  \label{thm:HMPpoly}
  Given a polyhedral model
  $\calP = (P,\peval{\calP})$, with $P=|K|$ for some simplicial complex $K$, and  $x_1,x_2 \in P$,
  the following holds: $x_1 \slcsEeq^{\calP} x_2$ if and only if $x_1 \wsibis^{\calP}
  w_2$. \qed
\end{thm}

Similar results hold for poset models. The following lemma
shows that, in every finite poset model $\calF$, weak
\plm-bisimilarity (Definition~\ref{def:WPLMBis}) is stronger than
logical equivalence with respect to \slcsE{,} i.e.\
$\wplmbis^{\calF} \: \subseteq \: \slcsEeq^{\calF}$:

\begin{lem}
  \label{lem:WFpBisIMPLEtaLog}
  Given a finite poset model $\calF = (W,\preccurlyeq, \peval{\calF})$,
  for all $w_1,w_2 \in W$, if $w_1 \wplmbis^{\calF} w_2$ then
  $w_1 \slcsEeq^{\calF} w_2$.
\end{lem}

\begin{proof}
  By induction on formulas. We consider only the case
  $\eta(\form_1,\form_2)$.  Suppose $w_1 \wplmbis w_2$ and
  $\calF, w_1 \models \eta(\form_1,\form_2)$. Then, there is (a
  \plm-path and so) a \dwn-path~$\pi_1$ from~$w_1$ of some
  length~$k_1$ such that $\calF,\pi_1(k_1) \models \form_2$ and for
  all $i \in [0;k_1)$ it holds that $\calF,\pi_1(i) \models \form_1$.
  By Lemma~\ref{lem:dPExists}, we know that a \dwn-path~$\pi_2$
  from~$w_2$ exists of some length~$k_2$ such that
  $\pi_1(k_1) \wplmbis \pi_2(k_2)$ and for all $j \in [0;k_2)$ exists
  $i \in [0;k_1)$ such that $\pi_1(i) \wplmbis \pi_2(j)$. By the
  Induction Hypothesis, we then get that
  $\calF,\pi_2(k_2) \models \form_2$ and for all $j \in [0;k_2)$ we
  have $\calF,\pi_2(j) \models \form_1$. This implies that
  $\calF, w_2 \models \eta(\form_1,\form_2)$.
\end{proof}

Furthermore, logical equivalence induced by \slcsE{} is stronger than
weak \plm-bisimilarity, i.e.\
$\slcsEeq^{\calF}\: \subseteq\: \wplmbis^{\calF}$, as implied by the
following:

\begin{lem}
  \label{lem:EtaLogIsWFpBis}
  In a finite poset model $\calF = (W,\preccurlyeq, \peval{\calF})$,
  $\slcsEeq^{\calF}$ is a weak \plm-bisimulation.
\end{lem}

\begin{proof}
  If $w_1 \slcsEeq w_2$, then the first requirement of
  Definition~\ref{def:WPLMBis} is trivially satisfied. We prove that
  $\slcsEeq$ satisfies the second requirement of
  Definition~\ref{def:WPLMBis}. Suppose $w_1 \slcsEeq w_2$ and let
  $u_1, d_1$ be as in the above-mentioned requirement. This implies that
  $\calF, w_1 \models \eta(\chi(w_1) \lor \chi(u_1), \chi(d_1))$,
  where, we recall, $\chi(w)$ is the `characteristic formula' for~$w$
  as in Definition~\ref{def:chiE}. Since $w_1 \slcsEeq w_2$, we also
  have that
  $\calF, w_2 \models \eta(\chi(w_1) \lor \chi(u_1), \chi(d_1))$
  holds. This in turn means that a \dwn-path~$\pi_2$ of some
  length~$k_2$ from~$w_2$ exists such that
  $\calF, \pi_2(k_2) \models \chi(d_1)$ and for all $j \in [0;k_2)$ we
  have $\calF, \pi_2(j) \models \chi(w_1) \lor \chi(u_1)$,
  i.e.\ $\calF, \pi_2(j) \models \chi(w_1)$ or
  $\calF, \pi_2(j) \models \chi(u_1)$.  Consequently, by
  Proposition~\ref{prop:chiE}, we have: $\pi_2(k_2) \slcsEeq d_1$ and,
  for all $j\in [0;k_2)$, $\pi_2(j) \slcsEeq w_1$ or
  $\pi_2(j) \slcsEeq u_1$, so that the second condition of the
  definition is fulfilled.
\end{proof}

On the basis of Lemma~\ref{lem:WFpBisIMPLEtaLog} and
Lemma~\ref{lem:EtaLogIsWFpBis}, we have that the largest weak
\plm-bisimulation exists, it is a weak \plm-bisimilarity, it is an
equivalence relation, and it coincides with logical equivalence in the
finite poset induced by \slcsE:

\begin{thm}
  \label{theo:PMbisEqSLCSEeq}
  For every finite poset model
  $\calF = (W,{\preccurlyeq},\peval{\calF}), w_1,w_2 \in W$, the
  following holds: $w_1 \slcsEeq^{\calF} w_2$ if and only if
  $w_1 \wplmbis^{\calF} w_2$. \qed
\end{thm}

By this we have established the HMP for $\wplmbis$ with respect to~\slcsE.

Recalling that, by Theorem~\ref{theo:calMPresForm}, given
polyhedral model $\calP = (|K|,\peval{\calP})$ for all $x \in |K|$ and
\slcsE{} formula~$\form$, we have that $\calP,x \models \form$ if and
only if $\map(\calP),\map(x) \models \form$, we get the following
final result:

\begin{cor}
  \label{cor:coincidenceE}
  Given a polyhedral model
  $\calP = (P,\peval{\calP})$, with $P=|K|$ for some simplicial complex $K$,
  for all $x_1,x_2 \in P$ the following holds:
  \begin{displaymath}
    \phantom{\qedsymbol \qquad \quad}
    x_1 \wsibis^{\calP} x_2 \,\mbox{ iff } \,
    x_1  \slcsEeq^{\calP}  x_2 \,\mbox{ iff } \,
    \map(x_1)  \slcsEeq^{\map(\calP)} \map(x_2) \,\mbox{ iff } \,
    \map(x_1)  \wplmbis^{\map(\calP)}  \map(x_2).
    \qquad \quad \qed
  \end{displaymath}
\end{cor}

\begin{figure}[t!]

\begin{center}
\subfloat[]{\label{subfig:MinEmodel}
    \begin{tikzpicture}[scale=0.9, every node/.style={transform shape}]
    \tikzstyle{kstate}=[rectangle,draw=black,fill=white]
    \tikzset{->-/.style={decoration={
		markings,
		mark=at position #1 with {\arrow{>}}},postaction={decorate}}}
   
    \node[kstate,line width=1mm,draw=violet!50,fill=green!50,label={270:$$}   ] (P1) at (  1,2) {$C_4$};
    
    \node[kstate,line width=1mm,draw=cyan!50,fill=red!50,label={90:$$} ] (E0) at (0,1) {$C_2$};
    
    \node[kstate,line width=1mm,draw=brown!50,fill=gray!50,label={90:$$} ] (T0) at ( 1,1) {$C_3$};
    
    \node[kstate,line width=1mm,draw=orange!50, fill=gray!50,label={90:$$} ] (P0) at ( 0,0) {$C_1$};
    
    \path (T0) edge[->,thick] (P1);
    
     \path (P0) edge[->,thick] (E0);
     
     \path (E0) [->, bend left,thick] edge[->] (P1);
    
    \draw [->, bend right,thick] (E0) to (T0);
    \draw [->, bend right,thick] (T0) to (E0); 
    \path (P0) edge[->, loop left,thick] (P0); 
    \path (P1) edge[->, loop right,thick] (P1);
    \path (E0) edge[->, loop left,thick] (E0); 
    \path (T0) edge[->, loop right,thick] (T0);     
    \end{tikzpicture} 
}\quad\quad\quad
\subfloat[]{\label{subfig:MinGmodel}
\begin{tikzpicture}[scale=0.9, every node/.style={transform shape}]
    \tikzstyle{gstate}=[rectangle,draw=black,fill=white]
    \tikzset{->-/.style={decoration={
        markings,
        mark=at position #1 with {\arrow{>}}},postaction={decorate}}}
    \node[gstate,line width=1mm,draw=magenta!25,fill=red!50  ] (C4) at (  0,0) {$C'_3$};
    \node[gstate,line width=1mm,draw=blue!15,fill=gray!50,above of=C4,xshift=2cm] (C6){$C'_5$};
    \node[gstate,line width=1mm,draw=orange!50,fill=gray!50,above of=C4,xshift=-2cm] (C3){$C'_2$};
    \node[gstate,line width=1mm,draw=blue!75,fill=gray!50,above of=C6,xshift=2cm] (C7){$C'_6$};
    \node[gstate,line width=1mm,draw=red!50,fill=gray!50,above of=C6,xshift=-2cm] (C8){$C'_7$};
    \node[gstate,line width=1mm,draw=yellow!75,fill=red!50,above of=C3,xshift=-2cm] (C2){$C'_1$};
    \node[gstate,line width=1mm,draw=brown!25,fill=gray!50,above of=C7] (C9){$C'_8$};
    \node[gstate,line width=1mm,draw=cyan!50,fill=red!50,above of=C2,xshift=-0.5cm] (C1){$C'_0$};
    \node[gstate,line width=1mm,draw=violet!50,fill=green!50,above of=C8] (C5){$C'_4$};
    \node[gstate,line width=1mm,draw=magenta,fill=gray!50,left of=C2] (C10){$C'_9$};
    \draw(C4) edge[thick, bend right](C7);
    \draw(C4) edge[thick, bend left](C2);
    \draw(C6) edge[thick](C7);
    \draw(C6) edge[thick](C8);
    \draw(C3) edge[thick](C2);
    \draw(C3) edge[thick](C8);
    \draw(C7) edge[thick](C5);
    \draw(C8) edge[thick](C9);
    \draw(C8) edge[thick](C5);
    \draw(C2) edge[thick](C5);
    \draw(C2) edge[thick](C1);
    \draw(C10) edge[thick](C1);
\end{tikzpicture}
}
\end{center}
\caption{
The minimal model 
$\map(\calP_{\ref{fig:PolyhedronNoPathCompressed}})_{\min}$, modulo weak \plm-bisimilarity 
(\ref{subfig:MinEmodel}), and modulo \plm-bisimilarity (\ref{subfig:MinGmodel}), of the cell poset model 
$\map(\calP_{\ref{fig:PolyhedronNoPathCompressed}})$
of Figure~\ref{subfig:PolyhedronNoPathPosetCompressed}.
Note that the minimal model modulo \plm-bisimilarity is a poset model and so it is represented by its Hasse diagram.
}\label{fig:exa:MinRunExaE}
\end{figure}
\noindent This says that \slcsE-equivalence in a polyhedral model is the
same as weak simplicial bisimilarity, which maps by~$\map$ to the weak
\plm-bisimilarity in the corresponding poset model, where the latter
coincides with the \slcsE-equivalence.

In the example below, and in the sequel, whenever we show a graphical representation of a minimal model in a figure, 
we use the following convention: each node of the Kripke model is coloured according to the
predicate letter satisfied by the cells belonging to the equivalence class represented by the node ---
obviously, since all such cells are weakly \plm-bisimilar, they all satisfy the same predicate letters\footnote{In the examples, for the sake of readability, each cell satisfies a single predicate letter, namely its ``colour''.} --- whereas the colour of the {\em border} of the node identifies the equivalence class itself,
and is, therefore, unique within the model.
Note that the colour of the borders of the nodes have only an illustrative purpose. In particular, they are not related to the colours expressing the evaluation of proposition letters.

\begin{exa}\label{ex:min}
Figure~\ref{subfig:MinEmodel} shows the minimal model 
$\map(\calP_{\ref{fig:PolyhedronNoPathCompressed}})_{\min}$, modulo $\wplmbis$, of the poset model $\map(\calP_{\ref{fig:PolyhedronNoPathCompressed}})$ shown in Figure~\ref{subfig:PolyhedronNoPathPosetCompressed}.
$\map(\calP_{\ref{fig:PolyhedronNoPathCompressed}})_{\min}$ is built using the procedure that will be described in detail in Section~\ref{sec:EtaMinimisation}. 
Note that $\map(\calP_{\ref{fig:PolyhedronNoPathCompressed}})_{\min}$ is not a poset model, but it is a reflexive Kripke model.
As we can see in the figure, we have four equivalence classes. More specifically, the classes are:
$C_1= \SET{\relint{A}}$, represented by the grey node with orange border,
$C_2= \SET{\relint{B}, \relint{C}, \relint{AB}, \relint{AC}, \relint{BC}, \relint{BD}, \relint{CD},
\relint{ABC}, \relint{BCD}}$, represented by the red node with cyan border,
$C_3 =\SET{\relint{D},\relint{E},\relint{F},\relint{CE},\relint{DE},\relint{DF},\relint{EF},
\relint{DEF}}$, represented by the grey node with brown border, and, finally,
$C_4 = \SET{\relint{CDE}}$, represented by the green node with violet border.

As we will see in Section~\ref{sec:EtaMinimisation}, the fact that
$\relint{D} \preccurlyeq \relint{CD}$ holds, with 
$\relint{D} \in C_3$ and $\relint{CD} \in C_2$, implies 
that $(C_3,C_2)$ belongs to the accessibility relation $R_{\min}$ of 
the Kripke model $\map(\calP_{\ref{fig:PolyhedronNoPathCompressed}})_{\min}$.
Similarly, we have that the fact that
$\relint{C} \preccurlyeq \relint{CE}$ holds, with
$\relint{C} \in C_2$ and $\relint{CE} \in C_3$, implies
that $(C_2,C_3)\in R_{\min}$.
With the same rationale, since $\relint{D} \preccurlyeq \relint{D}$ holds, we have
that $(C_3,C_3)\in R_{\min}$.
Finally, since $\relint{A} \preccurlyeq \relint{AB}$ and 
$\relint{CD} \preccurlyeq \relint{CDE}$, we get 
that $\SET{(C_1,C_2),(C_2,C_4)} \subseteq R_{\min}$ whereas we can see
from Figure~\ref{subfig:MinEmodel} that $(C_1,C_4)\not\in R_{\min}$.
The presence of cycles as the above, as well as the fact that transitivity of the accessibility relation is not guaranteed, imply that the minimal model of
a poset model, modulo $\wsibis$, is not necessarily a poset model. Anyway, it is guaranteed, by construction, to be a reflexive Kripke model.

Note that cell $\relint{A}$ of the poset model of Figure~\ref{subfig:PolyhedronNoPathPosetCompressed} is in a different equivalence class, namely $C_1$, than any other grey cell of the poset model: the latter
cells belong to $C_3$. In fact, it is easy to see that there is no weak \plm-bisimulation $Z$ such that $Z(\relint{A},w)$ for any $w\in C_3$. This is because condition (2) of Definition~\ref{def:WPLMBis} cannot be satisfied, as shown in the sequel. Suppose for instance $Z(\relint{A},\relint{D})$ for some weak bisimulation relation $Z$. Then, with reference to  Definition~\ref{def:WPLMBis}, 
take $w_1=\relint{D}$ and $u_1= d_1=\relint{CDE}$: clearly  
$w_1 \dircnv{\preccurlyeq} u_1 \succcurlyeq  d_1$.
Any $\pi_2$ from $\relint{A}$ should end in $\relint{CDE}$, otherwise $B(d_1,\pi_2(\ell_2))$ would not hold, since $\invpeval{\map(\calP_{\ref{fig:PolyhedronNoPathCompressed}})}(d_1)=\mathbf{green}$ and $\peval{\map(\calP_{\ref{fig:PolyhedronNoPathCompressed}})}(\mathbf{green})=\SET{\relint{CDE}}$.
But any path from $\relint{A}$ and ending in  $\relint{CDE}$ would necessarily pass by a
cell, say $\pi_2(j)$, for some $j\in (0;\ell_2)$ such that $\pi_2(j)\in \peval{\map(\calP_{\ref{fig:PolyhedronNoPathCompressed}})}(\mathbf{red})$.
For such a $j$ we would have that neither $Z(w_1,\pi_2(j))$ would hold, since $w_1 = \relint{D}\not\in \peval{\map(\calP_{\ref{fig:PolyhedronNoPathCompressed}})}(\mathbf{red})$, nor $Z(u_1,\pi_2(j))$, for the same reason. So, there exists no
weak \plm-bisimulation containing $(\relint{A},\relint{D})$.
And, in fact, we also have that 
$\map(\calP_{\ref{fig:PolyhedronNoPathCompressed}}), \relint{D} \models \eta(\mathbf{green}\, \lor \, \mathbf{grey}, \mathbf{green})$
whereas 
$\map(\calP_{\ref{fig:PolyhedronNoPathCompressed}}), \relint{A} \not\models \eta(\mathbf{green}\, \lor \, \mathbf{grey}, \mathbf{green})$.

As another example, suppose $Z(\relint{A},\relint{DEF})$ for some weak bisimulation relation $Z$ and
let $w_1=u_1=\relint{DEF}$ and $d_1=\relint{D}$. Any $\pi_2$ from from $\relint{A}$ should necessarily
end in a grey cell. But such a cell cannot be $\relint{A}$, since we already know that no
 \plm-bisimulation can contain $(\relint{A},\relint{D})$. And, on the other hand, if 
 $\pi_2(\ell_2)\in C_3$, then we would have a similar problem as above, with the unavoidable red elements of $\pi_2$. From the logical perspective, we see that
$\map(\calP_{\ref{fig:PolyhedronNoPathCompressed}}), \relint{DEF} \models \eta(\mathbf{grey},\eta(\mathbf{green}\, \lor \, \mathbf{grey}, \mathbf{green}))$
whereas
$\map(\calP_{\ref{fig:PolyhedronNoPathCompressed}}), \relint{A} \not\models \eta(\mathbf{grey},\eta(\mathbf{green}\, \lor \, \mathbf{grey}, \mathbf{green}))$.
The reasoning for all  the other cases is similar.
Finally, the reader can easily check that both 
$
\map(\calP_{\ref{fig:PolyhedronNoPathCompressed}}),\relint{D} \models \eta(\mathbf{grey} \lor \mathbf{red},
\mathbf{red})
$
and
$
\map(\calP_{\ref{fig:PolyhedronNoPathCompressed}}),\relint{E} \models \eta(\mathbf{grey} \lor \mathbf{red},
\mathbf{red}).
$
Actually, any grey point satisfies the above formula.

Weak \plm-bisimilarity ensures that, for  each \plm-path in the poset model, there is a corresponding 
\plm-path in the minimal model and vice-versa. For instance the \plm-path 
$
(C_2,C_2,C_2, C_2,C_3)
$
in the minimal model corresponds to \plm-path
$
(\relint{AB}, \relint{ABC},\relint{BC},\relint{BCD},\relint{D})
$ 
in the poset model --- witnessing, in both cases, $\map(\calP_{\ref{fig:PolyhedronNoPathCompressed}}),\relint{AB} \models \eta(\mathbf{red}, \mathbf{grey})$.
The correspondence, of course, is not unique: for instance, the above \plm-path in the minimal model corresponds also to the \plm-path
$
(\relint{AB}, \relint{ABC},\relint{C},\relint{CD},\relint{D})
$.

Finally, in Figure~\ref{subfig:MinGmodel}  the minimal model of
$\map(\calP_{\ref{fig:PolyhedronNoPathCompressed}})$ with respect to~$\slcsGeq$ is shown. Note that the minimal model is a poset model and, in fact, in the figure its  Hasse diagram is shown. We have 10 equivalence classes, namely
$C'_0=\SET{\relint{B}, \relint{AB}, \relint{AC},\relint{BC},\relint{BD},\relint{ABC},\relint{BCD}}$,
$C'_1=\SET{\relint{CD}}$,
$C'_2=\SET{\relint{D}}$,
$C'_3=\SET{\relint{C}}$,
$C'_4=\SET{\relint{CDE}}$,
$C'_5=\SET{\relint{E}}$,
$C'_6=\SET{\relint{CE}}$,
$C'_7=\SET{\relint{DE}}$,
$C'_8=\SET{\relint{F},\relint{DF},\relint{EF},\relint{DEF}}$, and
$C'_9=\SET{\relint{A}}$.
\closeex
\end{exa}

\section{Building the Minimal Model Modulo Logical Equivalence}\label{sec:EtaMinimisation}

In this section we present a minimisation procedure for finite poset
models modulo weak \plm-bisimilarity or, equivalently, modulo
$\slcsEeq$.
Given a finite poset model $\calF = (W, \preccurlyeq, \peval{\calF})$,
the procedure consists of three steps: \\[0.5em] 
\textbf{Step 1:} The poset model~$\calF$ is encoded as an
\lts{} denoted~$\posToltsC(\calF)$. The set of states
of~$\posToltsC(\calF)$ is~$W$ itself.
The encoding is such that it is ensured that logically equivalent
elements of $\calF$ are mapped into branching bisimilar states of $\posToltsC(\calF)$. 
Thus, for 
$w_1, w_2 \in W$ that are logically equivalent with respect to
\slcsE{} in the poset model~$\calF$, i.e.\ $w_1 \slcsEeq^{\calF} w_2$,
we have that they are branching bisimilar as states in the
\lts~$\posToltsC(\calF)$, i.e.\ $w_1 \beq^{\posToltsC(\calF)} w_2$.
\\[1em]
\textbf{Step 2:} The \lts~$\posToltsC(\calF)$ is reduced modulo
branching bisimilarity using available software tools, such as
\mcrltwo~\cite{Gr+17}. This step yields the set of equivalence classes
of~$W$ for~$\beq^{\posToltsC(\calF)}$. Because of the correspondence
of logical equivalence and branching bisimilarity,
we obtain $W / {\slcsEeq^{\calF}}$. \\[0.5em]
\textbf{Step 3:} The minimal model 
$\calF_{\min} = (W_{\min}, R_{\mkern1mu \min},\peval{\calF_{\min}})$
is built. It turns out that this
model is not necessarily a poset model (see the example in
Figure~\ref{subfig:MinEmodel}). 
However, it is a reflexive Kripke model 
where $W_{\min}=W / {\slcsEeq^{\calF}}$, $R_{\mkern1mu \min}$~is a
relation induced by the ordering~$\preccurlyeq$ of~$\calF$, and, most importantly,
\slcsE{} is preserved and reflected, i.e.
 for each $w \in W$ and \slcsE{} formula $\form$ the following
  holds:
  $ \calF,w \models \form \mbox{ if and only if } \calF_{\min},
  [w]_{\slcsEeq} \models \form.  $\\[0.5em]
In the remainder of this section we focus on Step~1 and Step~3.

\subsection{The Encoding of~$\calF$ as $\posToltsC(\calF)$}

We obtain the \lts~$\posToltsC(\calF) = (S,L,{\rightarrow})$ from the poset~$\calF$ as
specified in Definition~\ref{def:LTSetaConc} below.
$\posToltsC(\calF)$ is an \lts{}
representing each node $w\in W$ of $\calF$ as a distinct state. So, we put $S=W$.
For example, the set of states of the \lts{} $\posToltsC(\calF_{\ref{fig:encodeB}})$ of Figure~\ref{subfig:concreteLTS}  is 
$
\SET{\relint{D}, \relint{E},\relint{F},\relint{DE},\relint{EF}}
$,
i.e. the same as that of the nodes  of $\calF_{\ref{fig:encodeB}}= \map(\calP_{\ref{fig:encodeB}})$.

The set $L$ of transition labels includes all predicate letters in $\ap$, plus the 
``silent move''~$\tau$, typical of \lts{s} in concurrency theory, and the two special labels $\cact$ and $\dact$,
the meaning of which will be discussed later. In our example of Figure~\ref{fig:encodeB}, we have
$L=\SET{\mathbf{blue}, \mathbf{red},\tau,\cact, \dact}$.
We use transitions in $\posToltsC(\calF)$ for several purposes, as follows.
For each state $w$, the fact that $w$ (represents a node of $\calF$ that) satisfies a 
predicate letter $p$ is represented by a self-loop:
each predicate letter $p \in \ap$ such that  $w \in \peval{\calF}(p)$ is represented 
in $\posToltsC(\calF)$ by a transition from $w$ to itself, labelled by $p$ (Rule (PLC)). 
The transitions labelled by $\tau$ relate those states in $\posToltsC(\calF)$ representing
nodes in $\calF$ that are related by $\preccurlyeq$ or by $\succcurlyeq$ 
and satisfy the same set of predicate letters (Rule (TAU)).
Intuitively, this represents in the \lts{} the fact that ``nothing changes'' when moving from one such node $w$ to another one, $w'$ (including $w$ itself).

On the contrary, the fact that two states $w$ and $w'$ represent ``adjacent'' nodes of $\calF$ --- i.e. nodes related by $\dircnv{\preccurlyeq}$ --- which do {\em not} satisfy the same set of predicate letters, is modelled by  transitions $w \trans{\cact} w'$ and 
$w' \trans{\cact} w$, where $\cact$ stands for ``change'', with the obvious meaning (see Rule (CNG)).

Finally, Rule (DWN) makes sure that whenever $w  \succcurlyeq w'$ in $\calF$, a transition labelled
$\dact$ goes from (the state representing) $w$ to (that representing) $w'$. The label $\dact$ stands for ``down''. ``Marking'' the pair $(w,w')$ with the transition $w \trans{\dact}w'$ is relevant for identifying (the end of) \dwn-paths. Recall that such paths are the most fundamental ones for the semantics and the properties of  \slcsE.
We invite the reader to check that all the transitions in the \lts{} of Figure~\ref{subfig:concreteLTS}
are generated according to the above mentioned rules.

\begin{defi}
  \label{def:LTSetaConc}
  For a finite poset model $\calF = (W,\preccurlyeq,\peval{\calF})$ and
  symbols $\tau, \cact, \dact \notin \ap$, the \lts{}
  $\posToltsC(\calF)$ is defined by
  $\posToltsC(\calF) = (S,L,{\rightarrow})$ where
  \begin{itemize}
  \item the set of states~$S$ is the set~$W$;
  \item the set of labels~$L$ consists of $\ap \cup \SET{\tau, \cact,
      \dact}$;   
  \item the transition relation~$\rightarrow$ is the smallest relation
    on $S \times L \times S$ induced by the following transition
    rules.
    \begin{displaymath}
      \begin{array}{l}
        \mathrm{(PLC)} \ 
        \sosrule{w \in \peval{\calF}(p)}{w\trans{p} w}
        \qquad \quad
        \mathrm{(TAU)} \ 
          \sosrule{w \dircnv{\preccurlyeq} w'
          \quad
          \invpeval{\calF}(\SET{w}) = \invpeval{\calF}(\SET{w'})}{w \trans{\tau} w'}
          \medskip \\\\\\
        \mathrm{(CNG)} \ 
        \sosrule{w \dircnv{\preccurlyeq} w'
        \quad
        \invpeval{\calF}(\SET{w}) \neq \invpeval{\calF}(\SET{w'})}{w\trans{\cact} w'}
        \qquad 
        \mathrm{(DWN)} \ 
        \sosrule{w \succcurlyeq w'}{w \trans{\dact} w'}
        \quad \! {\bullet}
      \end{array}
    \end{displaymath}
  \end{itemize}
\end{defi}

\begin{figure}[t!]
\centering
\subfloat[\label{subfig:polyM}]{
  \begin{tikzpicture}[baseline={(1.5,-1.125)}]
    \tikzset{node distance=1cm}
    \tikzstyle{point}=[circle,inner sep=0pt,minimum width=4pt,minimum height=4pt]
    \node (A)[point,draw=red,fill=red] at (0,0){};
    \node (Aname) [above of = A,yshift = -1.5cm] {D};
    \node (B)[point,draw=blue,fill=blue, right of= A] {};
    \node (Bname) [above of = B,yshift = -1.5cm] {E};
    \node (C)[point,draw=blue,fill=blue, right of= B] {};
    \node (Cname) [above of = C,yshift = -1.5cm] {F};

    \draw [thick,red] (A) edge (B);
    \draw [thick,blue] (B) edge (C);
  \end{tikzpicture}
}
\qquad \qquad
\subfloat[\label{subfig:posetM}]{
  \begin{tikzpicture}[%
    scale=0.85, baseline={(1.5,-0.5)},
    every edge/.style={draw, thick}]
    
    \tikzstyle{kstate}=[rectangle,draw=black,fill=white]
    \node [kstate,fill=red!50] (P0) at (0,0) {$\relint{D}$};
    \node [kstate,fill=blue!50] (P1) at (1.5,0) {$\relint{E}$};
    \node [kstate,fill=blue!50] (P2) at (3,0) {$\relint{F}$};
    \node [kstate,fill=red!50] (E0) at (0.75,1.25) {$\relint{DE}$};
    \node [kstate,fill=blue!50] (E1) at (2.25,1.25) {$\relint{EF}$};
    \draw (P0) to (E0);
    \draw (P1) to (E0);
    \draw (P1) to (E1);
    \draw (P2) to (E1);
  \end{tikzpicture}
}
\qquad
\subfloat[\label{subfig:minM}]{%
  \begin{tikzpicture}[scale=0.8,
    every edge/.style={draw, ->, >=Stealth, shorten >=0pt, shorten <=0pt, thick}]

    \tikzstyle{kstate}=[rectangle, line width=1.75pt, draw]
    
    \node [kstate,line width=1mm,draw=brown,fill=blue!50] (P0) at (1,0) {$\phantom{B}$};
    \node [kstate,line width=1mm,draw=green,fill=red!50] (E0) at (1,1.5) {$\phantom{B}$};

    \draw (P0) edge (E0);
    \draw (P0) edge [->, >=Stealth, loop right, min distance=10mm, thick] (P0);
    \draw (E0) edge [->, >=Stealth, loop right, min distance=10mm, thick] (E0);

    \node [draw=none] (dummy) at (-0.625,-0.5) {\ } ;
  \end{tikzpicture}
} \\
\subfloat[\label{subfig:concreteLTS}]{
  \begin{tikzpicture}[scale=0.8, baseline={(0,-1.5)},
    every edge/.style={draw, ->, >=Stealth, thick},
    every loop/.style={min distance=8mm}]
    
    \tikzstyle{kstate}=[rectangle,draw=black,fill=white]
    
    \node[kstate,fill=red!50] (P00) at (0,2.5) {$\relint{D}$};
    \node[kstate,fill=red!50] (P0) at (2.5,2.5) {$\relint{DE}$};
    \node[kstate,fill=blue!50] (P1) at (2.5,0) {$\relint{E}$};
    \node[kstate,fill=blue!50] (P2) at (0,0) {$\relint{EF}$};
    \node[kstate,fill=blue!50] (P3) at (-2.5,0) {$\relint{F}$};
    \path (P00.south east) edge [<->, bend right] node [below]{$\tau$}
    (P0.south west);
    \path (P0.north west) edge [bend right] node [above] {$\dact$}
    (P00.north east);
    \path (P00) edge [loop above] node {$\tau,\dact$} (P00);
    \path (P00) edge [loop left] node [left, yshift=2pt]
    {\color{red}{{\bf red}}} (P00); 
    \path (P0) edge [loop above] node {$\tau,\dact$} (P0);
    \path (P0) edge [loop right] node [right, yshift=2pt]
    {\color{red}{{\bf red}}} (P0); 
    \path (P0) edge [<->,bend right] node [left] {$\cact$} (P1);
    \path (P0) edge [bend left] node [right] {$\dact$} (P1);
    \path (P1) edge [loop below] node {$\tau,\dact$} (P1);
    \path (P1) edge [loop right] node {\color{blue}{{\bf blue}}} (P1);
    \path (P2.south east) edge [<->, bend right] node [below, pos=0.35] {$\tau$}
    (P1.south west);
    \path (P2.north east) edge [bend left] node [above, pos=0.35] {$\dact$} (P1.north west);
    \path (P2.north west) edge [->, bend right] node [above] {$\dact$}
    (P3.north east);
    \path (P2) edge [loop below] node [below] {\color{blue}{{\bf blue}}} (P2);
    \path (P2) edge [loop above] node [above] {$\tau,\dact$} (P2);
    \path (P3.south east) edge [<->, bend right] node [below] {$\tau$}
    (P2.south west);
    \path (P3) edge [loop above] node [above] {$\tau,\dact$} (P3);
    \path (P3) edge [loop left] node [left, yshift=2pt] {\color{blue}{{\bf blue}}} (P3);
  \end{tikzpicture}
}
\hspace{-1cm}
\subfloat[\label{subfig:abstractLTS}]{
  \begin{tikzpicture}[scale=0.8, baseline={(0,-1.5)},
    every edge/.style={draw, ->, >=Stealth, thick},
    every loop/.style={min distance=8mm}]
    \tikzstyle{kstate}=[rectangle,draw=black,fill=white]

    \node [kstate,fill=red!50] (P0) at (0,2.5) {$\SET{\relint{D},\relint{DE}}$};
    \node [kstate,fill=blue!50] (P1) at (0,0) {$\SET{\relint{E},\relint{EF},\relint{F}}$};

    \path (P0) edge [loop above] node {$\dact,\sact$} (P0);
    \path (P1) edge [loop below] node {$\dact,\sact$} (P1);
    \path (P0) edge [<->, bend right] node [left]{$\sact$} (P1);
    \path (P0) edge [->, bend left] node [right, pos=0.45]{$\dact$} (P1);
    \node [kstate,fill=none, draw=none] (Q0) at ((0.7,2.5) {$\quad$} ;
    \path (Q0) edge [loop right, looseness=10] node [right] {\color{red}{$\SET{\mathbf{red}}$}} (Q0);
    \node [kstate,fill=none, draw=none] (Q1) at ((0.95,0) {$\quad$} ;Q
    \path (Q1) edge [loop right,looseness=10] node [right]
    {\color{blue}{$\SET{\mathbf{blue}}$}} (Q1);

    \node [draw=none] (dummy) at (-3.5,-0.5) {\ };
\end{tikzpicture}
}

\caption{
  (\ref{subfig:polyM}) A polyhedral model
  $\calP_{\ref{fig:encodeB}}$; (\ref{subfig:posetM})~Hasse diagram of the poset model $\calF_{\ref{fig:encodeB}}=
  \map(\calP_{\ref{fig:encodeB}})$; (\ref{subfig:minM})~minimal Kripke model
  ${\calF_{\ref{fig:encodeB}}}_{\min}$; (\ref{subfig:concreteLTS})~the
  \lts~$\posToltsC(\calF_{\ref{fig:encodeB}})$ obtained
  from~$\calF_{\ref{fig:encodeB}}$ by the encoding of Definition~\ref{def:LTSetaConc};
  (\ref{subfig:abstractLTS}) The
  \lts~$\posToltsA(\calF_{\ref{fig:encodeB}})$ obtained
  from~$\calF_{\ref{fig:encodeB}}$ by the encoding of Definition~\ref{def:LTSetaAbs}. Note that
  whenever $w \trans{\ell} w'$ and $w' \trans{\ell}
  w$ a ``double transition'' $w \stackrel{\ell}{\longleftrightarrow}
  w'$ is drawn in the figure between $w$ and~$w'$.}\label{fig:encodeB}
\end{figure}

\smallskip

In order to show that the above definition establishes that
$w_1 \slcsEeq^{\calF} w_2$ if and only if $w_1 \beq^{\posToltsC(\calF)} w_2$, it
is convenient to consider an intermediate structure, that is an \lts{}
too.
We denote this second \lts{} by~$\posToltsA(\calF)$. This structure
helps in the proofs to separate concerns related to the various
equivalences that are involved.
Suppose that 
nodes $w_1$ and~$w_2$ of~$\calF$
are encoded by the
states $s_1$ and~$s_2$ in~$\posToltsA(\calF)$, respectively. We will
have that  $w_1$ and~$w_2$ are logically equivalent in~$\calF$
with respect to \slcsE{} if and only if states $s_1$ and~$s_2$ are strongly
bisimilar (in the classical sense~\cite{Mil89}) in~$\posToltsA(\calF)$,
written $s_1 \seq^{\mkern1mu \posToltsA(\calF)} \!\!
s_2$. Furthermore, it will hold that $s_1$ and~$s_2$ are strongly
bisimilar in~$\posToltsA(\calF)$ if and only if $w_1$ and~$w_2$ are branching
bisimilar
in~$\posToltsC(\calF)$, thus providing the correctness of
the construction.

\lts~$\posToltsA(\calF)$ is more abstract
than~$\posToltsC(\calF)$ 
in the sense that all the  nodes of $\calF$
that satisfy the same proposition letters and that are  connected via 
$\dircnv{\preccurlyeq}$ are mapped to the same state of $\posToltsA(\calF)$.
Thus, intuitively, a state of $\posToltsA(\calF)$ corresponds to a class of states of
$\posToltsC(\calF)$. This is a class of states representing  nodes $w$ and $w'$ in $\calF$
for which ``nothing changes'' when moving from $w$ to $w'$, 
as discussed above.
More precisely, define
$\Theta = \ZET{\, \invpeval{\calF}(\SET{w})\, }{\, w \in W\, }$ and
consider, for $\alpha \in \Theta$, the $\alpha$-connected components
of~$\calF$. Then, each state~$s$ of~$\posToltsA(\calF)$ is an
$\alpha$-connected component of~$\calF$, for some~$\alpha$ as
above. So, we group together all the
nodes in~$W$ that can reach one
another only via a path in~$\calF$ composed of elements all satisfying
exactly the same proposition letters.
The above intuition
is formalised by the following definition.

\begin{defi}
  \label{def:LCCeq}
  Given  a finite poset model $\calF = (W,\preccurlyeq,\peval{\calF})$, we
  define relation ${\lcceq} \subseteq {W \times W}$ as the set of pairs
  $(w_1,w_2)$ such that an undirected path~$\pi$ of some length~$\ell$
  exists with $\pi(0) = w_1, \pi(\ell) = w_2$, and
  $\invpeval{\calF}(\SET{\pi(i)}) = \invpeval{\calF}(\SET{\pi(j)})$,
  for all $i,j \in [0;\ell]$.  \closedefi
\end{defi}

\noindent
The relevant definitions lead straightforwardly to the following
observation.

\begin{prop}
  Let $\calF = (W,\preccurlyeq,\peval{\calF})$ be a finite poset
  model. Then $\lcceq$~is an equivalence relation on~$W$. \qed
\end{prop}

\noindent
The encoding to the more ``abstract'' \lts{} is defined in Definition~\ref{def:LTSetaAbs} below.
The states of $\posToltsC(\calF)$ are the equivalence classes of $W$ modulo the equivalence relation 
$\lcceq$, i.e. $S=W{/}{\lcceq}$. With reference to Figure~\ref{fig:encodeB}, we obtain two states,
namely $\SET{\relint{D}, \relint{DE}}$ and $\SET{\relint{E},\relint{F},\relint{EF}}$, as shown in
Figure~\ref{subfig:abstractLTS}.
The set $L$ of transition labels includes the powerset  of the set of predicate letters in $\pws{\ap}$,
plus the two special labels $\sact, \dact$. 
In our example of Figure~\ref{fig:encodeB}, we have
$L_{\ref{fig:encodeB}}=\SET{\emptyset, \SET{\mathbf{blue}}, \SET{\mathbf{red}},\SET{\mathbf{blue},\mathbf{red}},\sact, \dact}$.

Similarly to Rule (PLC) for the definition of $\posToltsC(\calF)$, Rule (PL) induces a self-loop
in each state of $\posToltsA(\calF)$ (representing equivalence class) $[w]_{\lcceq}$.
This transition is labelled with the {\em set} of predicate letters $\invpeval{\calF}(w)$ satisfied by the elements
of the class. Note that, by definition of $\lcceq$, all the elements of such an equivalence class
satisfy the same set of predicate letters. 
Transitions labelled by $\dact$ (Rule (Down)) have the same interpretation as in the definition of $\posToltsC(\calF)$
while those labelled by $\sact$ (Rule (Step)) model a single step in $\dircnv{\preccurlyeq}$, regardless of 
there being ``a change'' or not.

\begin{defi}\label{def:LTSetaAbs}
  Given a finite poset model $\calF = (W,\preccurlyeq,\peval{\calF})$, and
  $\sact,\dact \notin \ap$, we define the \lts{}
  $\posToltsA(\calF) = (S,L,{\rightarrow})$ where
  \begin{itemize}
  \item the set $S$ of states is the quotient $W{/}{\lcceq}$ of~$W$
    modulo~$\lcceq$;
  \item the set~$L$ of labels is
    $\pws{\mkern1mu \ap} \cup \SET{\sact,\dact}$;
  \item the transition relation is the smallest relation on
    $W \times L \times W$ induced by the following transition rules:
       \[
      \begin{array}{c}
        \mathrm{(PL)}
        \ [w]_{\lcceq} \trans{\invpeval{\calF}(\SET{w})} [w]_{\lcceq}
        \medskip \\\\\\
        \mathrm{(Step)} \
        \sosrule{w \dircnv{\preccurlyeq} w'}{[w]_{\lcceq} \trans{\sact\phantom{\dact}}
        [w']_{\lcceq}}
        \qquad
        \mathrm{(Down)} \
        \sosrule{w \succcurlyeq w'}{[w]_{\lcceq} \trans{\dact} [w']_{\lcceq}}
      \end{array}
       \]
\end{itemize}
\closedefi
\end{defi}

The following theorem ensures that  any two elements $w_1$
and~$w_2$ of a finite poset model~$\calF$ are logically equivalent in
$\calF$ with respect to~\slcsE{} if and only if their equivalence
classes $[w_1]_{\lcceq}$ and~$[w_2]_{\lcceq}$ are strongly bisimilar
in~$\posToltsA(\calF)$.
The theorem uses the following lemma,
proven in Appendix~\ref{apx:prf:lem:LccImplSlcsEeq}:

\begin{lem}
  \label{lem:LccImplSlcsEeq}
  Given a finite poset model $\calF = (W,\preccurlyeq,\peval{\calF})$
  and $w_1, w_2 \in W$ the following holds: if $w_1 \lcceq w_2$, then
  $w_1 \slcsEeq w_2$.\qed
\end{lem}

\begin{thm}
  \label{thm:StrongEqLogeq}
  Let $\calF = (W,\preccurlyeq,\peval{\calF})$ be a finite poset
  model. For all $w_1,w_2 \in W$ it holds that
  $w_1 \slcsEeq^{\calF} w_2$ if and only if
  ${[w_1]_{\lcceq}} \seq^{\posToltsA(\calF)} {[w_2]_{\lcceq}}$.
\end{thm}

\begin{proof}
  We first prove that if
  $[w_1]_{\lcceq} \seq^{\posToltsA(\calF)} [w_2]_{\lcceq}$ then
  $w_1 \slcsEeq^{\calF} w_2$. We proceed by induction on \slcsE{}
  formulas and consider only the case $\eta(\form_1,\form_2)$, since
  the other cases are straightforward. Suppose
  $[w_1]_{\lcceq} \seq^{\posToltsA(\calF)} [w_2]_{\lcceq}$ and
  $\calF,w_1 \models \eta(\form_1,\form_2)$. Since
  $\calF,w_1 \models \eta(\form_1,\form_2)$, there is (a \plm-path,
  and so, by Proposition~\ref{prop:interchangeableE}) a
  \dwn-path~$\pi_1$ from~$w_1$ of some length $\ell_1 \geqslant 1$
  such that $\calF,\pi_1(\ell_1) \models \form_2$ and
  $\calF,\pi_1(i) \models \form_1$ for all $i \in [0;\ell_1)$. At this
  point, we use induction on~$\ell_1$, together with structural
  induction on the formulas, for showing that also $\calF,w_2 \models
  \eta(\form_1,\form_2)$ holds.\\[1em]
  \textbf{Base case:} $\ell_1 = 1$.\\
  In this case we have $\calF,w_1 \models\form_1$ and
  $\calF,\pi_1(1) \models\form_2$, with $w_1 \succcurlyeq \pi_1(1)$.
  Moreover, by the Induction Hypothesis on formulas, we also have
  $\calF,w_2 \models \form_1$.  In addition, by Rule (Down), we get
  $[w_1]_{\lcceq}\trans{\dact} [\pi_1(1)]_{\lcceq}$. Since
  $[w_1]_{\lcceq} \seq [w_2]_{\lcceq}$ by hypothesis, we also get
  $[w_2]_{\lcceq}\trans{\dact} [w'_2]_{\lcceq}$, for some
  $[w'_2]_{\lcceq}$ with $[w'_2]_{\lcceq} \seq
  [\pi_1(1)]_{\lcceq}$. Note that, by definition of ${\lcceq}$ and
  since $[w_2]_{\lcceq}\trans{\dact} [w'_2]_{\lcceq}$, there is a
  path~$\pi'_2$ from~$w_2$ of some length~$\ell'_2$ such that
  $\pi'_2(j) \lcceq w_2$ for all $j \in [0;\ell'_2]$ and
  $\pi'_2(\ell'_2) \succcurlyeq w''_2$, with
  $w''_2 \in [w'_2]_{\lcceq}$.  Recalling that
  $\calF,w_2 \models\form_1$, by Lemma~\ref{lem:LccImplSlcsEeq}, we
  also get that $\calF,\pi'_2(j) \models\form_1$ for all
  $j \in [0;\ell'_2]$.  Recalling also that
  $\calF,\pi_1(1) \models\form_2$, again by the Induction Hypothesis
  on formulas, from $[w'_2]_{\lcceq} \seq [\pi_1(1)]_{\lcceq}$, we get
  $\calF,w'_2 \models\form_2$ and, by Lemma~\ref{lem:LccImplSlcsEeq},
  we also get $\calF,w''_2 \models\form_2$. Consider now path
  $\pi_2:[0;\ell'_2+1] \to W$ defined as follows:
\[
\pi_2(j) =
\left\{
\begin{array}{l l}
\pi'_2(j) & \text{if  $j \in [0;\ell'_2]$}, \smallskip \\
w''_2 & \text{if $j = \ell'_2+1$}.
\end{array}
\right.
\]
Clearly, $\pi_2$~is a \dwn-path from~$w_2$ since $\pi'_2$~is an
undirected path and $\pi_2(\ell'_2) \succcurlyeq
\pi_2(\ell'_2+1)$. Furthermore, we have shown above that
$\calF,\pi_2(\ell'_2+1) \models\form_2$ and
$\calF,\pi_2(j) \models \form_1$ for all $j \in [0;\ell'_2+1)$.

Thus, we have that
$\calF,w_2 \models \eta(\form_1,\form_2)$, witnessed by~$\pi_2$.\\[1em]
\textbf{Induction step:} We assume the assertion holds for $\ell_1 =
n$, for $n \geqslant 1$ and we show it holds for $\ell_1 = n+1$.\\
Since $w_1 \dircnv{\preccurlyeq} \pi_1(1)$, by Rule (Step), we have
that $[w_1]_{\lcceq} \trans{\sact} [\pi_1(1)]_{\lcceq}$, and since, by
hypothesis, $[w_1]_{\lcceq} \seq [w_2]_{\lcceq}$, we also know that
$[w_2]_{\lcceq} \trans{\sact} [w'_2]_{\lcceq}$ for some~$w'_2$ such
that $[w'_2]_{\lcceq} \seq [\pi_1(1)]_{\lcceq}$. Furthermore,
$\calF,\pi_1(1) \models \eta(\form_1,\form_2)$ since
$\ell_1\geqslant 2$ and that this is witnessed by $\pi_1 \uparrow 1$,
which is a \dwn-path of length~$n$.  Thus, by the Induction Hypothesis
on~$\ell_1$, we get that $\calF,w'_2 \models \eta(\form_1,\form_2)$
since $[w'_2]_{\lcceq} \seq [\pi_1(1)]_{\lcceq}$ (see above).
From $[w_2]_{\lcceq} \trans{\sact} [w'_2]_{\lcceq}$, by Rule (Step),
we know that $w \in [w_2]_{\lcceq}$ and $w' \in [w'_2]_{\lcceq}$ exist
such that $w \dircnv{\preccurlyeq} w'$. Since $w \in [w_2]_{\lcceq}$
an undirected path~$\pi'_2$ exists from~$w_2$ to~$w$, of some
length~$\ell'_2$, such that $\pi'_2(j) \lcceq w_2$ for all
$j\in [0;\ell'_2]$. By the Induction Hypothesis on formulas, we know
that $\calF,w_2 \models \form_1$, and so, by
Lemma~\ref{lem:LccImplSlcsEeq}, we get also
$\calF,\pi'_2(j) \models \form_1$ for all $j \in
[0;\ell'_2]$. Moreover, since
$\calF,w'_2 \models \eta(\form_1,\form_2)$ (see above) and
$w' \lcceq w'_2$, again by Lemma~\ref{lem:LccImplSlcsEeq}, we get
$\calF,w' \models \eta(\form_1,\form_2)$. This means that there is a
\plm-path~$\pi''_2$ from~$w'$ of some length~$\ell''_2$ witnessing
$\calF,w' \models \eta(\form_1,\form_2)$. Define~$\pi_2$ as follows:
$\pi'_2\cdot (w,w') \cdot \pi''_2$. It is easy to see that $\pi_2$ is
a \dwn-path witnessing $\calF,w_2 \models \eta(\form_1,\form_2)$.

Now we prove that if $w_1 \slcsEeq^{\calF} w_2$ then
$[w_1]_{\lcceq} \seq^{\posToltsA(\calF)} [w_2]_{\lcceq}$. We do this
by showing that the following binary relation~$B$ on~$W$ is a strong
bisimulation:
\begin{displaymath}
  B =
  \ZET{(s_1,s_2) \in S \times S}{%
    \text{there are $w_1 \in s_1$, $w_2 \in s_2$ such that $w_1
      \slcsEeq w_2$}
    }.
\end{displaymath}
Let, without loss of generality, $s_1 = [w_1]_{\lcceq}$ and
$s_2 = [w_2]_{\lcceq}$, for some $w_1,w_2 \in W$ with
$w_1 \slcsEeq w_2$ and suppose $B([w_1]_{\lcceq}, [w_2]_{\lcceq})$,
with $w_1 \slcsEeq w_2$.
We distinguish three cases:\\[1em]
\textbf{Case A:} $[w_1]_{\lcceq} \trans{\alpha} [w_1']_{\lcceq}$ with
$\alpha \in \pws{\ap}$.\\ 
If $[w_1]_{\lcceq} \trans{\alpha} [w_1']_{\lcceq}$ for some
$\alpha \in \pws{\ap}$ and $w'_1\in W$, then, by Rule (PL), we know
that $ [w_1']_{\lcceq} = [w_1]_{\lcceq}$. Furthermore, since
$w_1 \slcsEeq w_2$, we also know that
$\invpeval{\calF}(\SET{w_2})= \invpeval{\calF}(\SET{w_1}) = \alpha$.
In addition, again by Rule (PL), we get that $[w_2]_{\lcceq} \trans{\alpha} [w_2]_{\lcceq}$ and, by hypothesis $B([w_1]_{\lcceq}, [w_2]_{\lcceq})$. \\[1em]
\textbf{Case B:} $[w_1]_{\lcceq} \trans{\dact} [w_1']_{\lcceq}$ \\
If $[w_1]_{\lcceq} \trans{\dact} [w_1']_{\lcceq}$ for some
$w'_1\in W$, then, by Rule (Down) there are $w \in [w_1]_{\lcceq}$ and
$w' \in [w'_1]_{\lcceq}$ such that $w \succcurlyeq w'$. Note that
$(w,w')$ is a \dwn-path witnessing
$\calF,w\models \eta(\chi(w),\chi(w'))$, where $\chi$~is as in
Definition~\ref{def:chiE} on page~\pageref{def:chiE}.  Since
$w \lcceq w_1$, we have that $\calF,w_1\models \eta(\chi(w),\chi(w'))$
holds, by Lemma~\ref{lem:LccImplSlcsEeq}. Moreover, since, by
hypothesis, $w_1 \slcsEeq w_2$, we also have
$\calF,w_2\models \eta(\chi(w),\chi(w'))$. Then a \plm-path
$\pi:[0;\ell] \to W$ exists from~$w_2$ such that
$\calF,\pi(\ell) \models \chi(w')$ and $\calF,\pi(j) \models \chi(w)$
for all $j \in [0;\ell)$.  This in turn, by
Proposition~\ref{prop:chiE}, means that $\pi(\ell) \slcsEeq w'$ and
$\pi(j) \slcsEeq w$ for all $j \in [0;\ell)$.
By Lemma~\ref{lem:LccImplSlcsEeq}, since $w' \lcceq w'_1$, we get
$w' \slcsEeq w'_1$, and by transitivity, since $\pi(\ell) \slcsEeq w'$
(see above), we also have $\pi(\ell) \slcsEeq w'_1$.  Similarly, we
get $\pi(j) \slcsEeq w \slcsEeq w_1$, which implies
$\invpeval{\calF}(\SET{\pi(j)})= \invpeval{\calF}(\SET{w_1})$, for all
$j\in [0;\ell)$.  Recall that $w_1 \slcsEeq w_2$, which implies
$\invpeval{\calF}(w_2)= \invpeval{\calF}(\SET{w_1})$ and so we get
also $\invpeval{\calF}(\SET{\pi(j)})= \invpeval{\calF}(\SET{w_2})$,
for all $j\in [0;\ell)$. In addition, for all $j\in [0;\ell)$ we have
that $\pi|[0;j]$ connects $\pi(0) = w_2$ to $\pi(j)$. This means that,
for all $j \in [0;\ell)$,
$\pi(j) \in [w_2]_{\lcceq} = [\pi(\ell-1)]_{\lcceq}$ and since
$\pi(\ell-1) \succcurlyeq \pi(\ell)$, by Rule (Down) we deduce
$[\pi(\ell-1)]_{\lcceq} \trans{\dact} [\pi(\ell)]_{\lcceq}$, that is
$[w_2]_{\lcceq} \trans{\dact} [\pi(\ell)]_{\lcceq}$.  Recall that
$\pi(\ell) \slcsEeq w'_1$, so that, by definition of relation $B$, we
finally get $B([w'_1]_{\lcceq},[\pi(\ell)]_{\lcceq})$.\\[1em]
\textbf{Case C:}
$[w_1]_{\lcceq} \trans{\sact} [w_1']_{\lcceq}$ \\
Suppose, finally, that $[w_1]_{\lcceq} \trans{\sact} [w_1']_{\lcceq}$
for some $w'_1\in W$.
We distinguish two cases:\\
\textbf{Case C1:} $w'_1 \in [w_1]_{\lcceq}$. In this case, by
Lemma~\ref{lem:LccImplSlcsEeq}, we have also $w'_1 \slcsEeq
w_1$. Furthermore, $w_1 \slcsEeq w_2$ by hypothesis, thus we get
$w'_1 \slcsEeq w_2$. But then, since $w_2 \dircnv{\preccurlyeq} w_2$,
by Rule (Step), we know that
$[w_2]_{\lcceq} \trans{\sact} [w_2]_{\lcceq}$ and since
$w'_1 \slcsEeq w_2$, by definition of relation~$B$, we
finally get $B([w'_1]_{\lcceq},[w_2]_{\lcceq})$. \\
\textbf{Case C2:} $w'_1 \notin [w_1]_{\lcceq}$. We know there are
$w \in [w_1]_{\lcceq}$ and $w' \in [w'_1]_{\lcceq}$ such that
$w \dircnv{\preccurlyeq} w'$. Since $w \lcceq w_1$, then
$\invpeval{\calF}(\SET{w}) = \invpeval{\calF}(\SET{w_1})$ and since
$w' \lcceq w'_1$, then
$\invpeval{\calF}(\SET{w'}) = \invpeval{\calF}(\SET{w'_1})$.
Furthermore, since $w \dircnv{\preccurlyeq} w'$, there is path
$(w,w')$ connecting $w$ with $w'$.  So there is a path connecting
$w_1$ to $w'_1$ and if
$\invpeval{\calF}(\SET{w_1}) = \invpeval{\calF}(\SET{w'_1})$ would
hold, it could not be that $w'_1 \notin [w_1]_{\lcceq}$.
Consequently, it must be
$\invpeval{\calF}(\SET{w_1}) \neq \invpeval{\calF}(\SET{w'_1})$,
which in turn implies $w_1 \not\slcsEeq w'_1$. We note that the
following holds:
\[
\calF,w_1 \models \eta(\chi(w_1),\eta(\chi(w_1) \lor \chi(w'_1),\chi(w'_1)))
\]
and, since $w_1 \slcsEeq w_2$ we also have
\[
\calF,w_2 \models \eta(\chi(w_1),\eta(\chi(w_1) \lor \chi(w'_1),\chi(w'_1))).
\]
Let~$\pi$ be a \plm-path from~$w_2$ witnessing the above formula and
let~$k$ be the first index such that
$\calF,\pi(k) \models \chi(w'_1)$.  We have that, for all
$j \in [0;k)$, $\calF,\pi(j) \models \chi(w_1)$ and $\pi|[0;j]$
connects $\pi(0)=w_2$ to $\pi(j)$. Furthermore, for all such~$j$, we
have $\pi(j) \slcsEeq w_1$, by Proposition~\ref{prop:chiE}, which
entails
$\invpeval{\calF}(\SET{\pi(j)}) = \invpeval{\calF}(\SET{w_1})$. Thus
$\pi(j) \in [w_2]_{\lcceq}$ for all $j\in [0;k)$ and since
$\pi(k-1) \dircnv{\preccurlyeq} \pi(k)$ we have, by Rule (Step)
$ [w_2]_{\lcceq} \trans{\sact} [\pi(k)]_{\lcceq}$.  Finally, recalling
that, again by Proposition~\ref{prop:chiE}, $w'_1 \slcsEeq\pi(k)$, we
get $B([w'_1]_{\lcceq}, [\pi(k)]_{\lcceq})$.
\end{proof}

The following theorem ensures that $[w_1]_{\lcceq}$
and~$[w_2]_{\lcceq}$ are strongly bisimilar in~$\posToltsA(\calF)$ if
and only if $w_1$ and~$w_2$ are branching bisimilar
in~$\posToltsC(\calF)$.  The theorem uses the following
lemma, proven in Appendix~\ref{apx:prf:lem:StrongLab}:

\begin{lem}
  \label{lem:StrongLab}
  Consider a finite poset model
  $\calF = (W,\preccurlyeq,\peval{\calF})$. Then for all
  $w_1, w_2 \in W$ the following holds: if
  $[w_1]_{\lcceq} \seq^{\posToltsA(\calF)} [w_2]_{\lcceq}$, then
  $\invpeval{\calF}(\SET{w_1}) = \invpeval{\calF}(\SET{w_2})$. \qed
\end{lem}

\begin{thm}
  \label{thm:StrongEqBranching}
  Let $\calF = (W,\preccurlyeq,\peval{\calF})$ be a finite poset
  model.  For all $w_1, w_2 \in W$ it holds that
  ${[w_1]_{\lcceq}} \seq^{\posToltsA(\calF)} {[w_2]_{\lcceq}}$ if and
  only if $w_1 \beq^{\posToltsC(\calF)} w_2$.
\end{thm}

\begin{proof}
  We first prove that if
  $[w_1]_{\lcceq} \seq^{\posToltsA(\calF)} [w_2]_{\lcceq}$ then
  $w_1 \beq^{\posToltsC(\calF)} w_2$. We show that the following
  relation is a branching bisimulation:
\[
B_C = \ZET{(w_1,w_2) \in {W \times W}}{[w_1]_{\lcceq}
  \seq^{\posToltsA(\calF)} [w_2]_{\lcceq}}. 
\]
Let us assume $B_C(w_1,w_2)$. We have to consider a few cases: \\[1em]
\textbf{Case A:} $w_1\trans{p}w_1$. \\
If $w_1\trans{p}w_1$, then, by Rule~(PLC), we have
$p \in \invpeval{\calF}(\SET{w_1})$. By definition of~$B_C$ and by
hypothesis we know that $[w_1]_{\lcceq} \seq [w_2]_{\lcceq}$ and so,
by Lemma~\ref{lem:StrongLab}, we get
$\invpeval{\calF}(\SET{w_1}) = \invpeval{\calF}(\SET{w_2})$.  It
follows then that $p \in \invpeval{\calF}(\SET{w_2})$ and, again by
Rule~(PLC), we finally get $w_2\trans{p}w_2$, which is the required
mimicking step since $B(w_1,w_2)$. \\[1em]
\textbf{Case B:} $w_1\trans{\tau}w'_1$. \\
If $w_1\trans{\tau}w'_1$ for some $w'_1 \in W$, then, by Rule~({TAU}),
we know that $w_1 \dircnv{\preccurlyeq} w'_1$, with
$\invpeval{\calF}(\SET{w_1}) = \invpeval{\calF}(\SET{w'_1})$, which,
by definition of $\lcceq$, means $[w'_1]_{\lcceq} = [w_1]_{\lcceq}$
and since $[w_1]_{\lcceq} \seq^{\posToltsA(\calF)} [w_2]_{\lcceq}$ by
definition of $B_C$, given that $B_C(w_1,w_2)$, we get
$[w'_1]_{\lcceq} \seq^{\posToltsA(\calF)} [w_2]_{\lcceq}$. This, in
turn, again by definition of $B_C$, means $B_C(w'_1,w_2)$.\\[1em]
\textbf{Case C:} $w_1\trans{\cact}w'_1$. \\
If $w_1\trans{\cact}w'_1$ for some $w'_1\in W$, then, by Rule (CNG),
we know that $w_1 \dircnv{\preccurlyeq} w'_1$, with
$\invpeval{\calF}(\SET{w_1}) \not= \invpeval{\calF}(\SET{w'_1})$, and,
by Rule (Step), we have $[w_1]_{\lcceq} \trans{\sact}[w'_1]_{\lcceq}$.
Since, by definition of $B_C$ and by hypothesis,
$[w_1]_{\lcceq} \seq^{\posToltsA(\calF)} [w_2]_{\lcceq}$, we also have
$[w_2]_{\lcceq} \trans{\sact}[w'_2]_{\lcceq}$ for some
$[w'_2]_{\lcceq} \seq^{\posToltsA(\calF)} [w'_1]_{\lcceq}$.  From
$[w_2]_{\lcceq} \trans{\sact}[w'_2]_{\lcceq}$, by Rule (Step), we know
there are $w_3 \in [w_2]_{\lcceq} $ and $w'_3 \in [w'_2]_{\lcceq} $
such that $w_3 \dircnv{\preccurlyeq} w'_3$.  By
Lemma~\ref{lem:StrongLab}, since
$[w_1]_{\lcceq} \seq^{\posToltsA(\calF)} [w_2]_{\lcceq}$ by hypothesis
and $[w'_1]_{\lcceq} \seq^{\posToltsA(\calF)} [w'_2]_{\lcceq}$ (see
above), we have
$\invpeval{\calF}(\SET{w_1}) = \invpeval{\calF}(\SET{w_2})$ and
$\invpeval{\calF}(\SET{w'_1}) = \invpeval{\calF}(\SET{w'_2})$ and
since $\invpeval{\calF}(\SET{w_1}) \not= \invpeval{\calF}(\SET{w'_1})$
(see above), we get
$\invpeval{\calF}(\SET{w_2})=\invpeval{\calF}(\SET{w_1})
\not=\invpeval{\calF}(\SET{w'_1}) = \invpeval{\calF}(\SET{w'_2})$.
Consequently, since $w_3 \in [w_2]_{\lcceq}$ and
$w'_3 \in [w'_2]_{\lcceq}$, we also finally get that
$\invpeval{\calF}(\SET{w_3}) \not= \invpeval{\calF}(\SET{w'_3})$.
Thus, by rule (CNG), we know that $w_3\trans{\cact}w'_3$.  Now, since
$w_3 \in [w_2]_{\lcceq} $, by definition of $\lcceq$ and by
construction of $\posToltsC(\calF)$ we know there are
$s_0, \ldots s_n \in W$ with $s_0=w_2$, $s_n = w_3$ such that
$s_i \trans{\tau} s_{i+1}$ and $s_{i+1} \trans{\tau} s_i$, for all
$i \in [0;n)$.  We note that $B_C(w_1,s_i)$ for all $i \in [0;n]$.  In
fact for each $i \in [0;n]$ we have that
$[s_i]_{\lcceq} = [w_2]_{\lcceq}$ by definition of $\lcceq$ and we
also know that
$[w_2]_{\lcceq} \seq^{\posToltsA(\calF)} [w_1]_{\lcceq}$, since
$B_C(w_1,w_2)$ by hypothesis.  Thus we get
$[s_i]_{\lcceq} \seq^{\posToltsA(\calF)} [w_1]_{\lcceq}$,
i.e. $B_C(w_1,s_i)$.  Furthermore, we also note that $B_C(w'_1,w'_3)$.
In fact $[w'_3]_{\lcceq} = [w'_2]_{\lcceq} $, since
$w'_3 \in [w'_2]_{\lcceq} $.  Furthermore,
$[w'_2]_{\lcceq} \seq^{\posToltsA(\calF)} [w'_1]_{\lcceq}$ (see
above).  So, we get
$[w'_3]_{\lcceq} \seq^{\posToltsA(\calF)} [w'_1]_{\lcceq}$,
i.e. $B_C(w'_1,w'_3)$. In conclusion, we have that if
$w_1\trans{\cact}w'_1$ for some $w'_1\in W$, then
$w_2=s_0 \trans{\tau} s_1 \trans{\tau} \ldots \trans{\tau} s_n = w_3
\trans{\cact} w'_3$ with
$B_C(w'_1,w'_3)$ and $B_C(w_1,s_i)$ for all $i \in [0;n]$.\\[1em]
\textbf{Case D:} $w_1\trans{\dact}w'_1$.\\
If $w_1\trans{\dact}w'_1$ for some $w'_1\in W$, then, by Rule (DWN),
we know that $w_1 \succcurlyeq w'_1$, and, by Rule ({Down}), we have
$[w_1]_{\lcceq} \trans{\dact}[w'_1]_{\lcceq}$.  Since, by definition
of~$B_C$ and by hypothesis,
$[w_1]_{\lcceq} \seq^{\posToltsA(\calF)} [w_2]_{\lcceq}$, we also have
$[w_2]_{\lcceq} \trans{\dact}[w'_2]_{\lcceq}$ for some
$[w'_2]_{\lcceq} \seq^{\posToltsA(\calF)} [w'_1]_{\lcceq}$.
From $[w_2]_{\lcceq} \trans{\dact}[w'_2]_{\lcceq}$, by Rule (Down), we
know there are $w_3 \in [w_2]_{\lcceq} $ and
$w'_3 \in [w'_2]_{\lcceq} $ such that $w_3 \succcurlyeq w'_3$ and, by
Rule (DWN) we know that $w_3 \trans{\dact}w'_3$.  Now, since
$w_3 \in [w_2]_{\lcceq} $, by definition of~$\lcceq$ and by
construction of~$\posToltsC(\calF)$ we know there are
$s_0, \ldots s_n \in W$ with $s_0 = w_2$, $s_n = w_3$ such that
$s_i \trans{\tau} s_{i+1}$ and $s_{i+1} \trans{\tau} s_i$, for all
$i \in [0;n)$.  We note that $B_C(w_1,s_i)$ for all $i \in [0;n]$. In
fact for each $i \in [0;n]$ we have that
$[s_i]_{\lcceq} = [w_2]_{\lcceq}$ by definition of~$\lcceq$ and we
also know that
$[w_2]_{\lcceq} \seq^{\posToltsA(\calF)} [w_1]_{\lcceq}$, since
$B_C(w_1,w_2)$ by hypothesis. Thus we get
$[s_i]_{\lcceq} \seq^{\posToltsA(\calF)} [w_1]_{\lcceq}$,
i.e. $B_C(w_1,s_i)$. Furthermore, we also note that $B_C(w'_1,w'_3)$.
In fact $[w'_3]_{\lcceq} = [w'_2]_{\lcceq} $, since
$w'_3 \in [w'_2]_{\lcceq} $. In addition,
$[w'_2]_{\lcceq} \seq^{\posToltsA(\calF)} [w'_1]_{\lcceq}$ (see
above). So, we get
$[w'_3]_{\lcceq} \seq^{\posToltsA(\calF)} [w'_1]_{\lcceq}$,
i.e. $B_C(w'_1,w'_3)$. In conclusion, we have that if
$w_1\trans{\dact}w'_1$ for some $w'_1 \in W$, then
$w_2 = s_0 \trans{\tau} s_1 \trans{\tau} \ldots \trans{\tau} s_n = w_3
\trans{\dact} w'_3$ with
$B_C(w'_1,w'_3)$ and $B_C(w_1,s_i)$ for all $i\in [0;n]$.\\[1em]

We now prove that if
$w_1 \beq^{\posToltsC(\calF)} w_2$, then
$[w_1]_{\lcceq} \seq^{\posToltsA(\calF)} [w_2]_{\lcceq}$.
We show that the following relation is a strong bisimulation:
\[
B_A = \ZET{(s_1,s_2) \in S \times S}{
\text{there are $w_1\in s_1$, $w_2 \in s_2$ such that $w_1
  \beq^{\posToltsC(\calF)} w_2$}}. 
\]
Let, without loss of generality, $s_1 = [w_1]_{\lcceq}$ and
$s_2 = [w_2]_{\lcceq}$ for some $w_1,w_2 \in W$ with
$w_1 \beq^{\posToltsC(\calF)} w_2$, and suppose
$B_A([w_1]_{\lcceq},[w_2]_{\lcceq})$.
We distinguish three cases:\\[1em]

\noindent
\textbf{Case A:} $[w_1]_{\lcceq} \trans{\alpha} [w_1']_{\lcceq}$ with
$\alpha \in \pws{\ap}$:\\
By Rule (PL), if $[w_1]_{\lcceq} \trans{\alpha}
[w_1']_{\lcceq}$ for $\alpha \in \pws{\ap}$ and $w'_1 \in
W$, then $[w_1']_{\lcceq} = [w_1]_{\lcceq}$ and $\alpha =
\invpeval{\calF}(\SET{w_1})$.
On the one hand, if $p \in \alpha$ then $w_1 \trans{p} w_1$ by
rule~(PLC).
Since $w_2 \beq^{\posToltsC(\calF)} w_1$ it follows that
$w_2 \trans{\tau} \ldots \trans{\tau}  \bar{w}_2 \trans{p} w'_2$ for
$\bar{w}_2, w'_2 \in W$ such that $p \in \invpeval{\calF}(\SET{\bar{w}_2})$,
$\bar{w}_2 \beq^{\posToltsC(\calF)} w_1$, and
$w'_2 \beq^{\posToltsC(\calF)} w_1$.
By rule~(TAU), $p \in \invpeval{\calF}(\SET{w_2})$.
Thus, $\alpha \subseteq \invpeval{\calF}(\SET{w_2})$.
On the other hand, if $p \in \invpeval{\calF}(\SET{w_2})$ then $w_2 \trans{p}
w_2$ by rule~(PLC).
Since $w_1 \beq^{\posToltsC(\calF)} w_2$ we have that $w_1
\trans{\tau} \ldots \trans{\tau} \bar{w}_1 \trans{p} w'_1$ for $\bar{w}_1, w'_1 \in W$
such that $p \in \invpeval{\calF}(\SET{\bar{w}_1})$, $\bar{w}_1
\beq^{\posToltsC(\calF)} w_2$, $w'_1 \beq^{\posToltsC(\calF)} w_2$.
By rule~(TAU) we obtain that $p \in \invpeval{\calF}(\SET{\bar{w}_1})$.
Thus,~$p \in \alpha$.
Hence, $\invpeval{\calF}(\SET{w_2}) \subseteq \alpha$.
So, $\invpeval{\calF}(\SET{w_2}) = \alpha$.
Therefore, $[w_2]_{\lcceq} \trans{\alpha} [w_2]_{\lcceq}$ by
rule~(PL).
By assumption, $B_A([w_1]_{\lcceq},[w_2]_{\lcceq})$ for target states
$[w_1]_{\lcceq}$ and $[w_2]_{\lcceq}$ as required.\\[1em]
\textbf{Case B:} $[w_1]_{\lcceq} \trans{\dact} [w_1']_{\lcceq}$\\
If $[w_1]_{\lcceq} \trans{\dact} [w_1']_{\lcceq}$ for some
$w'_1 \in W$, then, by Rule ({Down}), we know that there are
$w_3 \in [w_1]_{\lcceq}$ and $w'_3 \in [w'_1]_{\lcceq}$ such that
$w_3 \succcurlyeq w'_3$.  This implies, by Rule (DWN), that
$w_3 \trans{\dact} w'_3$. By definition of~$\lcceq$ and by
construction of~$\posToltsC(\calF)$ we know that there are
$m \geqslant 0$ and $t_0, \ldots, t_m \in W$ with $t_0 = w_1$,
$t_m = w_3$ such that $t_i \trans{\tau} t_{i+1}$ and
$t_{i+1} \trans{\tau} t_i$, for all $i \in [0;m)$.  This implies that
$w_1 \beq^{\posToltsC(\calF)} w_3$, and consequently
$w_2 \beq^{\posToltsC(\calF)} w_3$, since
$w_1 \beq^{\posToltsC(\calF)} w_2$ by hypothesis.  Furthermore, since
$w_3 \beq^{\posToltsC(\calF)} w_2$, there are $n \geqslant 0$ and
$v_0, \ldots, v_n, v_{n+1} \in W$ with
$w_2= v_0 \trans{\tau} \cdots \trans{\tau} v_{n} \trans{\dact}
v_{n+1}$, such that $w'_3 \beq^{\posToltsC(\calF)} v_{n+1}$ and
$w_3 \beq^{\posToltsC(\calF)} v_i$ for all $i\in [0;n]$.  Moreover, by
Rule (DWN), we have $v_{n} \succcurlyeq v_{n+1}$ which imples, by Rule
(Down), that $[v_n]_{\lcceq} \trans{\dact} [v_{n+1}]_{\lcceq}$.  Note
that, by construction of $\posToltsC(\calF)$ we also have
$\invpeval{\calF}(w_2) = \invpeval{\calF}(v_0)= \ldots =
\invpeval{\calF}(v_n)$ and so $[v_i]= [w_2]_{\lcceq}$ for all
$i\in [0;n]$.  Thus,
$[w_2]_{\lcceq} = [v_n]_{\lcceq} \trans{\dact} [v_{n+1}]_{\lcceq}$.
Furthermore, $B_A([w'_3]_{\lcceq}, [v_{n+1}]_{\lcceq})$ holds, since
$w'_3 \beq^{\posToltsC(\calF)} v_{n+1}$ (see above)
and, recalling that $[w'_3]_{\lcceq} = [w'_1]_{\lcceq}$, we also know
that $B_A([w'_1]_{\lcceq}, [v_{n+1}]_{\lcceq})$.\\[1em]
\textbf{Case C:} $[w_1]_{\lcceq} \trans{\sact} [w_1']_{\lcceq}$\\
If $[w_1]_{\lcceq} \trans{\sact} [w_1']_{\lcceq}$ for some
$w'_1 \in W$, then, by Rule ({Step}), we know that there are
$w_3 \in [w_1]_{\lcceq}$ and $w'_3 \in [w'_1]_{\lcceq}$ such that
$w_3 \dircnv{\preccurlyeq} w'_3$.
We distinguish two cases:\\
\textbf{Case C1:}
$\invpeval{\calF}(\SET{w_3}) = \invpeval{\calF}(\SET{w'_3})$.\\
If $\invpeval{\calF}(\SET{w_3}) = \invpeval{\calF}(\SET{w'_3})$, then,
by Rule (TAU), we know $w_3 \trans{\tau} w'_3$. But then, by
definition of~$\lcceq$, we get $[w_3]_{\lcceq} = [w'_3]_{\lcceq}$ and
since $[w_3]_{\lcceq} = [w_1]_{\lcceq}$ and
$[w'_3]_{\lcceq} = [w'_1]_{\lcceq}$ (see above), we get
$[w'_1]_{\lcceq} = [w_1]_{\lcceq}$. On the other hand, since,
trivially, $w_2 \dircnv{\preccurlyeq} w_2$, by Rule (Step), we also
get that $[w_2]_{\lcceq} \trans{\sact} [w_2]_{\lcceq}$.  Moreover,
since by hypothesis, we also have
$B_A([w_1]_{\lcceq},[w_2]_{\lcceq})$, we finally get
that also $B_A([w'_1]_{\lcceq},[w_2]_{\lcceq})$.\\
\textbf{Case C2:} $\invpeval{\calF}(\SET{w_3})\not=\invpeval{\calF}(\SET{w'_3})$.\\
If $\invpeval{\calF}(\SET{w_3})\not=\invpeval{\calF}(\SET{w'_3})$,
then, by Rule (CNG), we know $w_3 \trans{\cact} w'_3$.  By definition
of $\lcceq$ and by construction of~$\posToltsC(\calF)$ we know that
there are $m\geqslant 0$ and $t_0, \ldots, t_m \in W$ with $t_0 = w_1$,
$t_m = w_3$ such that $t_i \trans{\tau} t_{i+1}$ and
$t_{i+1} \trans{\tau} t_i$, for all $i \in [0;m)$. This implies that
$w_1 \beq^{\posToltsC(\calF)} w_3$, and consequently
$w_2 \beq^{\posToltsC(\calF)} w_3$, since
$w_1 \beq^{\posToltsC(\calF)} w_2$ by hypothesis.  Furthermore, since
$w_3 \beq^{\posToltsC(\calF)} w_2$, there are $n \geqslant 0$ and
$v_0, \ldots, v_n, v_{n+1} \in W$ with
$w_2 = v_0 \trans{\tau} \cdots \trans{\tau} v_{n} \trans{\cact}
v_{n+1}$, such that $w'_3 \beq^{\posToltsC(\calF)} v_{n+1}$ and
$w_3 \beq^{\posToltsC(\calF)} v_i$ for all $i\in [0;n]$. Moreover, by
Rule (CNG), we have $v_{n} \dircnv{\preccurlyeq} v_{n+1}$ which
imples, by Rule (Step), that
$[v_n]_{\lcceq} \trans{\sact} [v_{n+1}]_{\lcceq}$. Note that, by
construction of $\posToltsC(\calF)$ we also have
$\invpeval{\calF}(w_2) = \invpeval{\calF}(v_0) = \ldots =
\invpeval{\calF}(v_n)$ and so $[v_i]= [w_2]_{\lcceq}$ for all
$i\in [0;n]$.  Thus,
$[w_2]_{\lcceq} = [v_n]_{\lcceq} \trans{\sact} [v_{n+1}]_{\lcceq}$.
Furthermore, $B_A([w'_3]_{\lcceq}, [v_{n+1}]_{\lcceq})$ holds, since
$w'_3 \beq^{\posToltsC(\calF)} v_{n+1}$ (see above) and, recalling
that $[w'_3]_{\lcceq} = [w'_1]_{\lcceq}$, we also know that
$B_A([w'_1]_{\lcceq}, [v_{n+1}]_{\lcceq})$.
\end{proof}

From Theorems~\ref{thm:StrongEqLogeq} and~\ref{thm:StrongEqBranching}
we finally obtain our claim: 

\begin{cor}
  \label{cor:LogeqEqBranch}
  Let $\calF = (W,\preccurlyeq,\peval{\calF})$ be a finite poset
  model. For all $w_1,w_2 \in W$ the following holds:
  $w_1 \slcsEeq^{\calF} w_2$ if and only if
  $w_1 \beq^{\posToltsC(\calF)} w_2$.\qed
\end{cor}

\noindent
Now that we have characterised logical equivalence~$\slcsEeq$ for~\slcsE{}
for the elements of a finite poset model~$\calF$ in terms of branching
bisimilarity~$\beq$ for the \lts~$\posToltsC(\calF)$, we can
compute the minimal \lts{} modulo branching bisimilarity with standard
techniques available, such as branching bisimilarity
minimisation provided by the \mcrltwo{} toolset.

\subsection{Building the Minimal Model}

Via the correspondence of \slcsE{} logical equivalence for a poset
model and branching bisimilarity of its encoding, one can obtain the
equivalence classes of~$\slcsEeq$ by identifying the branching
bisimilar states in the~\lts. With the equivalence classes
modulo~$\slcsEeq$ for the poset model available, we can consider the
ensued quotient model. We obtain a Kripke model that is minimal with
respect to~$\slcsEeq$, but which is not necessarily a poset model.

\begin{defi}[$\calF_{\min}$]
  \label{def:MinE} 
  For a finite poset model $\calF = (W,\preccurlyeq,\peval{\calF})$ let
  the Kripke model
  $\calF_{\min} = (W_{\min}, R_{\mkern1mu \min},\peval{\calF_{\min}})$
  have
  \begin{itemize}
  \item set of nodes $W_{\min} = W / {\slcsEeq}$, the equivalence
    classes of~$W$ with respect to~$\slcsEeq$,
  \item accessibility relation
    $R_{\min}\subseteq {W_{\min} \times W_{\min}}$ satisfying
    \smallskip
    \begin{center}
      $R_{\min}([w_1],[w_2])$ if and only if $w'_1 \preccurlyeq
      w'_2$ for some $w'_1 
      \slcsEeq w_1$ and $w'_2 \slcsEeq w_2$
    \end{center}
    \smallskip
    for $w_1, w_2 \in W\!$, and
  \item valuation $\peval{\calF_{\min}} : \ap \to
    \pws{W_{\min}}$ such that
    \smallskip 
    \begin{center}
      $\peval{\calF_{\min}}(p) = \ZET{%
        \, [w] \in W_{\min}
      }{%
        \text{$w' \in \peval{\calF} (p)$ for some $w' \slcsEeq w$} \,
      }$
    \end{center}
    for~$p \in \ap$.  \closedefi
  \end{itemize}
\end{defi}

\noindent
Clearly,
$\calF_{\min}$ is a finite reflexive Kripke model. Reflexivity of the
accessibility relation~$R_{\mkern1mu
  \min}$ is immediate from reflexivity of the ordering~$\preccurlyeq$.
Furthermore, it is minimal with respect to
\slcsE{} by definition of $\slcsEeq$ and
$W/\slcsEeq$. An example of the minimal Kripke model of the
polyhedral model in Figure~\ref{subfig:polyM} is shown in
Figure~\ref{subfig:minM}. The following theorem ensures that the model
defined above is sound and complete with respect to the logic, so that
the minimisation procedure is correct.

\begin{thm}\label{theo:MinE}
  Given a finite poset model $\calF = (W,\preccurlyeq,\peval{\calF})$
  let $\calF_{\min}$ be defined as in Definition~\ref{def:MinE}.
  Then, for each $w \in W$ and \slcsE{} formula $\form$ the following
  holds:
  $ \calF,w \models \form \mbox{ if and only if } \calF_{\min},
  [w]_{\slcsEeq} \models \form.  $
\end{thm}

\begin{proof}
  We first prove that $\calF,w \models \form$ implies
  $\calF_{\min}, [w]_{\slcsEeq} \models \form$. We proceed by
  induction on the structure of~$\form$ and we show the proof only for
  $\form = \eta(\form_1,\form_2)$ the other cases being
  straightforward. Suppose $\calF,w \models \eta(\form_1,\form_2)$.
  This means there is a \plm-path~$\pi$ of some
  length~$\ell \geqslant 2$ such that $\pi(0) = w$,
  $\calF, \pi(\ell) \models \form_2$, and
  $\calF, \pi(i) \models \form_1$ for all $i \in [0;\ell)$. Now define
  $\pi_{\min}: [0;\ell] \to W_{\min}$ with $\pi_{\min}(i) = [\pi(i)]$
  for all $i \in [0;\ell]$. We show that $\pi_{\min}$ is a \plm-path
  with respect to~$R_{\min}$. We have that
  $R_{\min}(\pi_{\min}(0),\pi_{\min}(1))$ by definition of~$R_{\min}$
  because $\pi(0)\in [\pi(0)] = \pi_{\min}(0)$,
  $\pi(1) \in [\pi(1)] = \pi_{\min}(1)$ and
  $\pi(0) \preccurlyeq \pi(1)$ by assumption. Similarly, we have that
  $\cnv{R_{\min}}(\pi_{\min}(\ell-1),\pi_{\min}(\ell))$ and also that
  $\dircnv{R_{\min}}(\pi_{\min}(i),\pi_{\min}(i+1))$ for all
  $i \in (0;\ell-1)$. Furthermore, since
  $\calF, \pi(\ell)\models \form_2$, by the Induction Hypothesis, we
  have that $\calF_{\min}, \pi_{\min}(\ell)\models
  \form_2$. Similarly, we have that
  $\calF_{\min}, \pi_{\min}(i)\models \form_1$ for all $i\in [0;\ell)$
  since $\calF, \pi(i)\models \form_1$. So
  $\calF_{\min}, [w]_{\slcsEeq} \models \eta(\form_1,\form_2)$.\\[1em]
  Now we prove that $\calF_{\min}, [w]_{\slcsEeq} \models \form$
  implies $\calF,w \models \form$. Also in this case we proceed by
  induction on the structure of~$\form$ and we show the proof only for
  $\form = \eta(\form_1,\form_2)$. Suppose
  $\calF_{\min}, [w]_{\slcsEeq} \models \eta(\form_1,\form_2)$. 
Hence
  there is a \plm-path~$\pi_{\min}$ such that
  $\pi_{\min}(0) = [w]_{\slcsEeq}$,
  $\calF_{\min}, \pi(\ell_{\min}) \models \form_2$, and
  $\calF_{\min}, \pi_{\min}(i) \models \form_1$ for all
  $i \in [0;\ell_{\min})$. Since $R_{\min}$~is reflexive, using
  Lemma~\ref{lem:pm2ud}, we know that there is also an
  \upd-path~$\hat{\pi}_{\min}$ from~$[w]_{\slcsEeq}$ of some
  length~$2k$, for $k \geqslant 1$, with the same starting-/ending
  points and the same intermediate points as~$\pi_{\min}$ and that
  obviously witnesses $\eta(\form_1,\form_2)$ for $[w]_{\slcsEeq}$.
  By induction on~$k$, in the sequel, we show that there is a
  \plm-path~$\pi$ from $w$ witnessing $\eta(\form_1,\form_2)$.\\[1em] 
\noindent
\textbf{Base case:} $k=1$.\\
In this case, we have that $\hat{\pi}_{\min}(0) = [w]_{\slcsEeq}$,
$\calF_{\min}, \hat{\pi}_{\min}(0) \models \form_1$
$\calF_{\min}, \hat{\pi}_{\min}(1) \models \form_1$, and
$\calF_{\min}, \hat{\pi}_{\min}(2) \models \form_2$. Furthermore,
since $\hat{\pi}_{\min}$~is an \upd-path with respect to~$R_{\min}$,
we know that
\[
\hat{\pi}_{\min}(0)= [w]_{\slcsEeq}, 
R_{\min}(\hat{\pi}_{\min}(0),\hat{\pi}_{\min}(1)),
\cnv{R_{\min}}(\hat{\pi}_{\min}(1),\hat{\pi}_{\min}(2))
\] and, by definition of~$R_{\min}$, there are
$w_0 \in \hat{\pi}_{\min}(0) = [w]_{\slcsEeq}$,
$w'_1, w''_1 \in \hat{\pi}_{\min}(1)$, and
$w_2 \in \hat{\pi}_{\min}(2)$ such that $w_0 \preccurlyeq w'_1$ and
$w''_1 \succcurlyeq w_2$.  Moreover, by the Induction Hypothesis with
respect to the structure of formulas, we have that
$\calF, w_0 \models \form_1$, $\calF, w'_1 \models \form_1$,
$\calF, w''_1 \models \form_1$, and $\calF, w_2 \models \form_2$. Note
that $\calF, w''_1 \models \eta(\form_1,\form_2)$, witnessed by the
following \plm-path: $(w''_1,w''_1,w_2)$.  But then we have that also
$\calF, w'_1 \models \eta(\form_1,\form_2)$ holds since
$w'_1 \slcsEeq w''_1$, recalling that
$w'_1, w''_1 \in \hat{\pi}_{\min}(1) \in W/\slcsEeq.$ There is then
a \plm-path $\pi':[0;\ell'] \to W$ from~$w'_1$ of some length~$\ell'$
such that $\calF,\pi'(\ell')\models \form_2$ and
$\calF,\pi'(i)\models \form_1$ for all $i\in [0;\ell')$.  Furthermore,
$w_0 \preccurlyeq w'_1$ by hypothesis and so
$\pi = (w_0,w'_1)\cdot \pi':[0;\ell'+1]\to W$ is a \plm-path from~$w_0$
witnessing $\calF, w_0 \models \eta(\form_1,\form_2)$. Finally,
recalling that $w, w_0 \in \hat{\pi}_{\min}(0) \in W/\slcsEeq$, we
know that $w \slcsEeq w_0$ and so we have proven the assertion
$\calF, w \models \eta(\form_1,\form_2)$.\\[1em]
\noindent
\textbf{Induction step:} $k=n{+}1$ assuming the assertion holds for
$k=n$, for $n>0$.\\ 
Since $k>1$, we know that
$\calF_{\min}, \hat{\pi}_{\min}(1) \models \form_1$ and
$\calF_{\min}, \hat{\pi}_{\min}(2) \models \form_1\land \lneg \form_2.$
Furthermore, 
\[
\hat{\pi}_{\min}(0)= [w]_{\slcsEeq}, 
R_{\min}(\hat{\pi}_{\min}(0),\hat{\pi}_{\min}(1)),
\cnv{R_{\min}}(\hat{\pi}_{\min}(1),\hat{\pi}_{\min}(2))
\]
because $\hat{\pi}_{\min}$ is an \upd-path. By definition of
$R_{\min}$, there are $w_0 \in \hat{\pi}_{\min}(0)=[w]_{\slcsEeq}$,
$w'_1, w''_1 \in \hat{\pi}_{\min}(1)$ and
$w_2 \in \hat{\pi}_{\min}(2)$ such that $w_0 \preccurlyeq w'_1$ and
$w''_1 \succcurlyeq w_2$. By the Induction Hypothesis with respect to
the structure of the formula, we get that
$\calF, w_0 \models \form_1$, $\calF, w'_1 \models \form_1$,
$\calF, w''_1 \models \form_1$, and
$\calF, w_2 \models \form_1\land \lneg \form_2$.  We consider now the
\upd-path~$\hat{\pi}_{\min}\uparrow 2$ from~$\hat{\pi}_{\min}(2)$ of
length~$2n$, noting that it witnesses $\eta(\form_1,\form_2)$, since
so does $\hat{\pi}_{\min}$ and $k>1$. In other words, we have that
$\calF_{\min}, \hat{\pi}_{\min}(2) \models \eta(\form_1,\form_2)$ with
$w_2 \in \hat{\pi}_{\min}(2)$. By the Induction Hypothesis with
respect to $k$, we then have that
$\calF, w_2 \models \eta(\form_1,\form_2).$ So there is a \upd-path
$\pi_2:[0;\ell_2] \to W$ from $w_2$ of some length $\ell_2$ such that
$\calF,\pi_2(\ell_2) \models \form_2$ and
$\calF,\pi_2(i) \models \form_1$ for $i\in [0;\ell_2)$. Note that
$\calF, \pi_2(0)\models \form_1$ as well, since $\pi_2(0)=w_2$ and
$\calF, w_2 \models \form_1\land \lneg \form_2$ (see above).  Let us
consider now the path $\pi''=(w''_1,w''_1,w_2)\cdot \pi_2$. Such a
path is an \upd-path since so is $\pi_2$, and $w''_1 \succcurlyeq w_2$
by hypothesis. Note that \upd-path $\pi''$ witnesses
$\calF, w''_1 \models \eta(\form_1,\form_2)$.  But then we have that
also $\calF, w'_1 \models \eta(\form_1,\form_2)$ holds since
$w'_1 \slcsEeq w''_1$, recalling that
$w'_1, w''_1 \in \hat{\pi}_{\min}(1) \in W/\slcsEeq.$ Thus, we have
that the following holds:
$\calF, w'_1 \models \form_1 \land \eta(\form_1,\form_2)$. There is
then a \plm-path $\pi':[0;\ell']\to W$ from $w'_1$ of some length
$\ell'$ such that $\calF,\pi'(\ell')\models \form_2$ and
$\calF,\pi'(i)\models \form_1$ for all $i\in [0;\ell')$.  Furthermore,
$w_0 \preccurlyeq w'_1$ by hypothesis and so
$\pi= (w_0,w'_1)\cdot \pi':[0;\ell'+1]\to W$ is a \plm-path from $w_0$
witnessing $\calF, w_0 \models \eta(\form_1,\form_2)$.  Finally,
recalling that $w, w_0 \in \hat{\pi}_{\min}(0) \in W /\slcsEeq$, we
know that $w \slcsEeq w_0$ and so we have proven the assertion
$\calF, w \models \eta(\form_1,\form_2)$.
\end{proof}

Finally, the following theorem turns out to be useful for simplifying
the procedure for the effective construction of~$\calF_{\min}$:

\begin{thm}
  \label{theo:d}
  For any poset model $\calF = (W,\preccurlyeq,\peval{\calF})$ and
  $\calF_{\min}$ as of Definition~\ref{def:MinE} and for all
  $\alpha_1, \alpha_2 \in W_{\min}$, it holds that
  $R_{min}(\alpha_1,\alpha_2)$ if and only if
  $\alpha_2 \trans{\dact} \alpha_1$ is a transition of the minimal
  \lts{} obtained from $\posToltsC(\calF)$ via branching bisimilarity.
\end{thm}

\begin{proof}
  In the sequel, we let $\posToltsC(\calF)/\beq$ denote the minimal
  \lts{} obtained from $\posToltsC(\calF)$ via branching bisimilarity.
  First of all, by Corollary~\ref{cor:LogeqEqBranch}, $W_{\min}$
  coincides with the quotient of the set of states $W$ of
  $\posToltsC(\calF)$ modulo branching bisimilarity.  Now, suppose that
  $\alpha_2 \trans{\dact} \alpha_1$ is a transition of
  $\posToltsC(\calF)/\beq$.  By standard construction of the
  minimal \lts{} modulo an equivalence on its state set, we know that
  $w_1 \in \alpha_1$ and $w_2 \in \alpha_2$ exist such that
  $w_2 \trans{\dact} w_1$ is a transition of $\posToltsC(\calF)$.  But
  then, by Rule (DWN), we get that $w_1 \preccurlyeq w_2$ and so, by
  definition of $\calF_{\min}$, we finally get $R_{\min}(\alpha_1,\alpha_2)$.
  If, on the other hand, $R_{\min}(\alpha_1,\alpha_2)$ holds, then we know
  that there exist $w_1 \in \alpha_1$ and $w_2 \in \alpha_2$ such that
  $w_1 \preccurlyeq w_2$, by definition of $\calF_{\min}$. But then,
  by Rule (DWN), we get that $w_2 \trans{\dact} w_1$ is a transition
  of $\posToltsC(\calF)$. Again, by standard construction of the
  minimal \lts{} modulo an equivalence on its state set, we know that
  $\alpha_2 \trans{\dact} \alpha_1$ is a transition
  of~$\posToltsC(\calF)/\beq$.
\end{proof}

\begin{rem}\label{rem:NotPoset}
The fact that the minimal model might not be a poset model  does not constitute a problem, at any (i.e. theoretical, implementation, user) level.
More specifically, at the theoretical level, 
Theorem~\ref{theo:MinE} guarantees that \slcsE{} interpreted on a finite poset model
$\calF$ is preserved and reflected by the minimisation result $\calF_{\min}$, despite the finite reflexive Kripke model $\calF_{\min}$  is not necessarily a poset model. 
The above, via Theorem~\ref{theo:calMPresForm},
guarantees that \slcsE{}  is preserved and reflected  by the full chain of translations, from the polyhedral model $\calP$ to the minimal model $\map(\calP)_{\min}$ via finite poset $\map(\calP)$.

In summary, we have:
\begin{equation}\label{correctness}
\calP,x \models \form \quad \mbox{ iff }\quad
\map(\calP),\map(x) \models \form \quad\mbox{ iff }\quad
\map(\calP)_{\min},  [\map(x)]_{\slcsEeq} \models \form.
\end{equation}
Taking the first and the last statements of (\ref{correctness}) above we get the following: 
a point $x$ of a polyhedral model $\calP$,
laying in a cell $\relint{\sigma}$ of $\calP$, satisfies a \slcsE{} formula $\form$ in the polyhedral interpretation
of $\form$ on $\calP$ if and only if 
the node of the Kripke model $\map(\calP)_{\min}$ that (uniquely) represents
the equivalence class $[\map(x)]_{\slcsEeq}$ of $\map(x)=\relint{\sigma}$ modulo $\slcsEeq$ (or, equivalently modulo weak \plm-bisimilarity) satisfies $\form$ in the relational interpretation
of $\form$ on $\map(\calP)_{\min}$.
At the implementation level, an experimental prototype of a variant of \polylogica{} has been developed that is capable to deal with general Kripke models and $\eta$ semantics, as briefly discussed in Section~\ref{sec:toolchain} below.
At  the user level, we observe that the user deals only with the description of the polyhedral model $\calP$ and the input formula $\form$ as input and the (figure showing the) cells satisfying $\form$ as output of model checking.  All the details of the minimisation  procedure are hidden to the user.\closerem
\end{rem}

\section{An Experimental Minimisation Toolchain}
\label{sec:toolchain}

In this section we provide a brief overview of an experimental
toolchain to study the minimisation procedure for polyhedral models
and to illustrate the practical potential of the theory presented in
the previous section. The further development and a thorough analysis
of the toolchain will be the subject of future work.
Figure~\ref{fig:toolchain} illustrates the elements of the toolchain
that, starting from a polyhedral model in~\texttt{json} format,
produces the set of equivalence classes and the minimal Kripke
model. The former may serve as input for the \polyvisualizer{}
tool\footnote{http://ggrilletti2.scienceontheweb.net/polyVisualizer/polyVisualizer\_static\_maze.html}~\cite{Be+22},
a polyhedra visualizer, to inspect the results, whereas the latter can
be used for spatial model checking. 
For that purpose, a variant of \polylogica{} is required, since minimal models may turn out not to  be posets. In particular, they  might not be transitive (see the discussion in Example~\ref{ex:min} and in Section~\ref{sec:EtaMinimisation}).
In addition, the variant has to accomodate for the different semantics of the reachability operators $\gamma$ and $\eta$. An experimental prototype of the tool has been developed and it is publicly available.$^{\ref{ftn:swAvailability}}$ The complexity of the model checking algorithm is linear in the size of the model and the number of sub-formulas to be checked. A fully fledged implementation and efficiency study is left for future work.

The toolchain is also
able to map the results obtained on the minimal Kripke model back to
the original polyhedral model, because of the direct correspondence
between the states of the Kripke model and the equivalence classes.

\begin{figure}[h!]
\centering
\tikzstyle{block} = [rectangle, draw, fill=blue!25,text width=6em, text centered, rounded corners, minimum height=2em, line width=1pt ]
\tikzstyle{line} = [draw, -latex', line width=1pt]
\tikzstyle{mucrlts} = [text=black, fill=green!25, draw=green]
\resizebox{!}{0.95in}{
\begin{tikzpicture}[node distance = 2cm, auto]
    \node [block] (poly) {Poly2Poset};
    \node [block, right of=poly, node distance=3cm] (poset) {Poset2mcrl2};
    \node [block, right of=poset, node distance=3cm,mucrlts] (mcrl) {mcrl2lps};
    \node [block, below of=mcrl, mucrlts] (lps) {lps2lpspp};
    \node [block, right of=mcrl, node distance=3cm,mucrlts] (lts) {lps2lts};
    \node [block, right of=lps, node distance=3cm] (findStates) {findStates\\renameLps};
    \node [block, right of=lts, node distance=3cm,mucrlts] (mini) {ltsMinimise};
    \node [block, below of=mini] (classes) {Classes $+$\\Kripke model};
    \path [line] (poly) -> (poset);
    \path [line] (poset) -> (mcrl);
    \path [line] (mcrl) -> (lps);
    \path [line] (lps) -> (findStates);
    \path [line] (lts) -> (mini);
    \path [line] (mini) -> (classes);
     \path [line] (findStates) -> (classes);
     \path [line] (findStates) -> (lts);
    
\end{tikzpicture}
}
\caption{Toolchain for polyhedral model minimisation. Parts in green
  are command line operations of the \mcrltwo{} toolset. Parts in
  blue are developed in Python in the context of the current
  paper.}\label{fig:toolchain}
  
\end{figure}

The toolchain uses several command line operations provided by the
\mcrltwo{} toolset \cite{Bu+19} (shown in green in
Figure~\ref{fig:toolchain}) and a number of operations developed in the
context of this paper (shown in blue in Figure~\ref{fig:toolchain}). The
prototype aims to demonstrate the feasibility of our approach from a
qualitative perspective, providing support for examples that
illustrate the practical usefulness of the theory. 
The operation \texttt{Poly2Poset} transforms the polyhedral model into
a poset model. The operation \texttt{Poset2mcrl2} encodes the poset
model into a \mcrltwo{} specification of an LTS following the
procedure defined in Definition~\ref{def:LTSetaConc}. The operations
\texttt{mcrl2lps} and \texttt{lps2lts} transform the encoding into a
linearised LTS-representation which is then minimised
(\texttt{ltsMinimise}) via branching bisimulation. The operation
\texttt{lps2lpspp} provides a textual version of the linear process
which is used to obtain the correspondence between internal state
labels of the minimised LTS and the cells of the original polyhedral
model present in the equivalence classes. The latter, in turn, are
essential for the generation of the result files of model checking the
minimised model and form the input to the \polyvisualizer{} (together
with the original polyhedral model and a colour definition
file). Figure~\ref{fig:cube3x3} and Figure~\ref{fig:cube3x3more} in the next
section show an example.\footnote{\label{ftn:swAvailability}The software and examples are available at \url{https://github.com/VoxLogicA-Project/Polyhedra-minimisation}.} Maintaining the relation between internal
state labels of the minimised LTS and the original states of the poset
and polyhedral model is the most tricky part of the toolchain as such
internal state labels are assigned dynamically in the \texttt{lps2lts}
procedure. This aspect is dealt with by the \texttt{findStates} and
\texttt{renameLps}
procedures.

\section{Minimisation at Work}\label{sec:Experiments}

In this section, we show, as a proof of concept,  an example of use of the experimental toolchain presented in Section~\ref{sec:toolchain}.
Figure~\ref{subfig:maze3D} 
shows a simple symmetric
3D maze composed of one white room in the middle, 
26~green rooms, and connecting grey corridors.  
Like in the previous examples, the cells of the white and green rooms satisfy only predicate letter
$\mathbf{white}$ and $\mathbf{green}$, respectively. Those of corridors satisfy only $\mathbf{corridor}$.
In total, the structure consists of 2,619
cells. We have chosen a symmetric structure on purpose. 
This makes it easy to interpret
the various equivalence classes as nodes of the minimal Kripke model of
this structure, shown in Figure~\ref{subfig:posetMaze}.
Note the considerable reduction that is obtained: from 2,619 cells
to just~7 in the minimal model (observe furthermore that, for this example, the minimal model is also a poset model).

Figure~\ref{subfig:minLTSMaze}  
shows the minimal \lts{} with respect
to branching bisimilarity as produced by \mcrl.\footnote{The numbering of
the states is as generated by \mcrltwo. 
} 
The minimal Kripke model with respect to $\slcsEeq$ obtained (see Theorem~\ref{theo:d})  from the \lts{} of 
Figure~\ref{subfig:minLTSMaze} is shown in Figure~\ref{subfig:posetMaze}.
The Kripke model has seven nodes --- of course, in direct correspondence with the seven states of the minimal  \lts.
Node~$\mathtt{C}1$ represents the class of the cells of the white room and is coloured in white in the figure, three nodes ($\mathtt{C}3$, $\mathtt{C}0$, and~$\mathtt{C}5$) correspond to cells of corridors and are coloured in grey, and the other three ($\mathtt{C}4$, $\mathtt{C}2$, and~$\mathtt{C}6$) correspond to  cells of green rooms, and are coloured in green.  
Green node~$\mathtt{C}4$ (visualised on the original polyhedron in
Figure~\ref{subfig:C4}) represents the class of (the cells of) green rooms that are
directly connected to the white room by a corridor. Green node~$\mathtt{C}2$
(visualised in Figure~\ref{subfig:C2}) represents the class of (the cells of) green
rooms situated on the edges of the maze. Green node~$\mathtt{C}6$ (visualised
in Figure~\ref{subfig:C6}) represents the class of green rooms situated
at the corners of the maze.

\begin{figure}
\centering
\subfloat[\label{subfig:maze3D}Maze]{\includegraphics[width=0.25\textwidth]{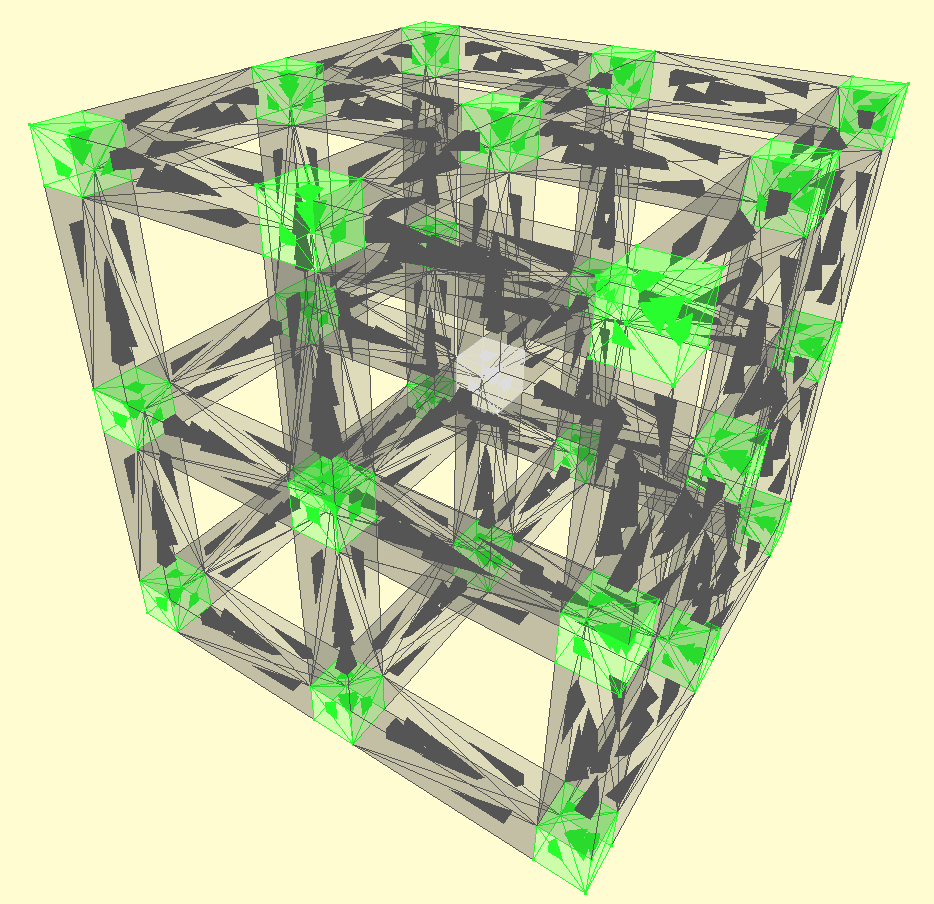}}\quad\quad
\subfloat[\label{subfig:minLTSMaze}Minimal LTS]{\scriptsize
\resizebox{4in}{!}{
\begin{tikzpicture}[%
  scale=0.275,
  every loop/.style={looseness=10, min distance=30mm},
  every state/.style={inner sep=0pt, minimum size=6mm},
  every edge/.style={draw, ->, >=Stealth}]

  \node [state, fill=gray!50] at (0,0) (state0) {{C0}};
  \node [state] at (-13.5,0) (state1) {{C1}};
  \node [state, fill=green] at (4.5,0) (state2) {{C2}};
  \node [state, fill=gray!50] at (-9,0) (state3) {{C3}};
  \node [state, fill=green] at (-4.5,0) (state4) {{C4}};
  \node [state, fill=gray!50] at (9,0) (state5) {{C5}};
  \node [state, fill=green] at (13.5,0) (state6) {{C6}};
  \node [draw=none] (dummy) at (0,-5.5) {} ;

  \draw (state5) edge [loop above] node [above] {\ftncorridor} (state5) ;
  \draw (state5) edge [bend right=60] node [below] {\ftndact} (state6);
  \draw (state5) edge node [below] {\ftncact} (state2);
  \draw (state5) edge [bend left=50] node [above] {\ftncact} (state6);
  \draw (state5) edge [loop below] node [below] {\ftndact} (state5);
  \draw (state5) edge [bend left=60] node [below] {\ftndact} (state2);

  \draw (state0) edge [bend left=60] node [below] {\ftndact} (state4);
  \draw (state0) edge [loop above] node [above] {\ftncorridor} (state0);
  \draw (state0) edge [bend left=50] node [above] {\ftncact} (state2);
  \draw (state0) edge node [below] {\ftncact} (state4);
  \draw (state0) edge [loop below] node [below] {\ftndact} (state0);
  \draw (state0) edge [bend right=60] node [below] {\ftndact} (state2);

  \draw (state3) edge [bend left=60]node [below] {\ftndact} (state1);
  \draw (state3) edge [loop above] node [above] {\ftncorridor} (state3);
  \draw (state3) edge [bend right=60] node [below] {\ftndact} (state4);
  \draw (state3) edge [bend left=50] node [above] {\ftncact} (state4);
  \draw (state3) edge node [below] {\ftncact} (state1);
  \draw (state3) edge [loop below] node [below] {\ftndact} (state3);

  \draw (state6) edge node [below] {\ftncact} (state5);
  \draw (state6) edge [out=-45, in=0, looseness=6] node [below, yshift=-2pt] {\ftndact} (state6);
  \draw (state6) edge [out=+45, in=0, looseness=6] node [above, yshift=+2pt] {\ftngreen} (state6);

  \draw (state2) edge [loop above, min distance=45mm] node [above] {\ftngreen} (state2);
  \draw (state2) edge [loop below] node [below] {\ftndact} (state2);
  \draw (state2) edge node [below] {\ftncact} (state0);
  \draw (state2) edge [bend left=50] node [above] {\ftncact} (state5);

  \draw (state4) edge node [below] {\ftncact} (state3);
  \draw (state4) edge [loop above,  min distance=45mm] node [above] {\ftngreen} (state4);
  \draw (state4) edge [loop below] node [below] {\ftndact} (state4);
  \draw (state4) edge [bend left=50] node [above] {\ftncact} (state0);

  \draw (state1) edge [out=135, in=180, looseness=6] node [above, yshift=+2pt] {\ftnwhite} (state1);
  \draw (state1) edge [out=225, in=180, looseness=6] node [below, yshift=-2pt] {\ftndact} (state1);
  \draw (state1) edge [bend left=50] node[above] {\ftncact} (state3);
\end{tikzpicture}
}
} 
\\
 \subfloat[Min. Kripke model\label{subfig:posetMaze}]{
    \phantom{AA}
    \begin{tikzpicture}[scale=1.0, every node/.style={transform shape}]
    \tikzstyle{kstate}=[rectangle,draw=black,fill=white]
    \tikzset{->-/.style={decoration={
		markings,
		mark=at position #1 with {\arrow{>}}},postaction={decorate}}}
   
    \node[kstate,line width= 1mm, draw=cyan,fill=white!50,label={270:$$}  ] (P0) at (  0,0) {$C1$}; 
    \node[kstate,line width= 1mm, draw=red,fill=green!50,label={270:$$}   ] (P1) at (  1,0) {$C4$}; 
    \node[kstate,line width= 1mm, draw=blue,fill=green!50,label={270:$$}   ] (P2) at (  2,0) {$C2$}; 
    \node[kstate,line width= 1mm, draw=yellow,fill=green!50,label={270:$$}   ] (P3) at (  3,0) {$C6$}; 

    \node[kstate,line width= 1mm, draw=magenta,fill=gray!50,label={90:$$} ] (E0) at (0.5,1) {$C3$}; 
    \node[kstate,line width= 1mm, draw=orange,fill=gray!50,label={90:$$} ] (E1) at (1.5,1) {$C0$};
    \node[kstate,line width= 1mm, draw=black,fill=gray!50,label={90:$$} ] (E2) at (2.5,1) {$C5$}; 

    \draw (P0) edge[->,thick] (E0);
    \draw (P1) edge[->,thick] (E0);
    \draw (P1) edge[->,thick] (E1);

    \draw (P2) edge[->,thick] (E1);
    \draw (P2) edge[->,thick] (E2);
    \draw (P3) edge[->,thick] (E2);
    
    \path (P0) edge[->, loop below,thick] (P0);
    \path (P1) edge[->, loop below,thick] (P1);
    \path (P2) edge[->, loop below,thick] (P2);
    \path (P3) edge[->, loop below,thick] (P3);
    
    \path (E0) edge[->, loop above,thick] (E0);
    \path (E1) edge[->, loop above,thick] (E1);
    \path (E2) edge[->, loop above,thick] (E2);

\end{tikzpicture}  
\phantom{AAAAAAA}     
}\\
\subfloat[\label{subfig:C4}C4]{\includegraphics[width=0.2\textwidth]{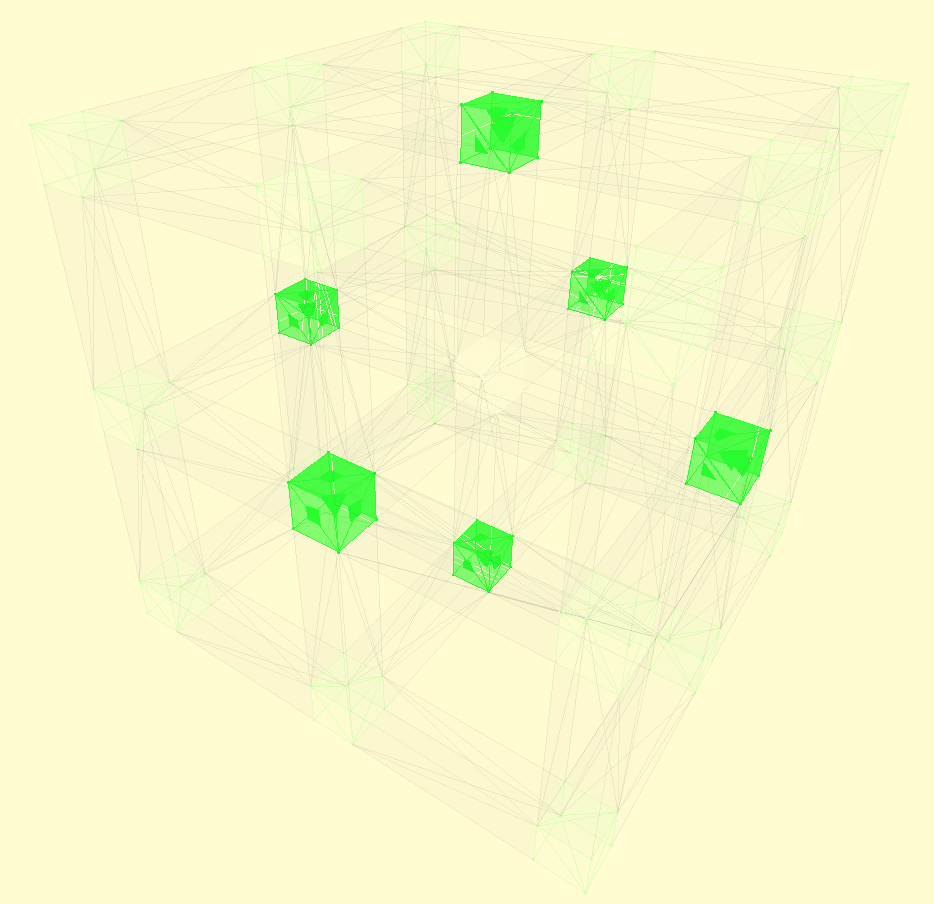}\phantom{AA}}\quad
\subfloat[\label{subfig:C2}C2]{\includegraphics[width=0.2\textwidth]{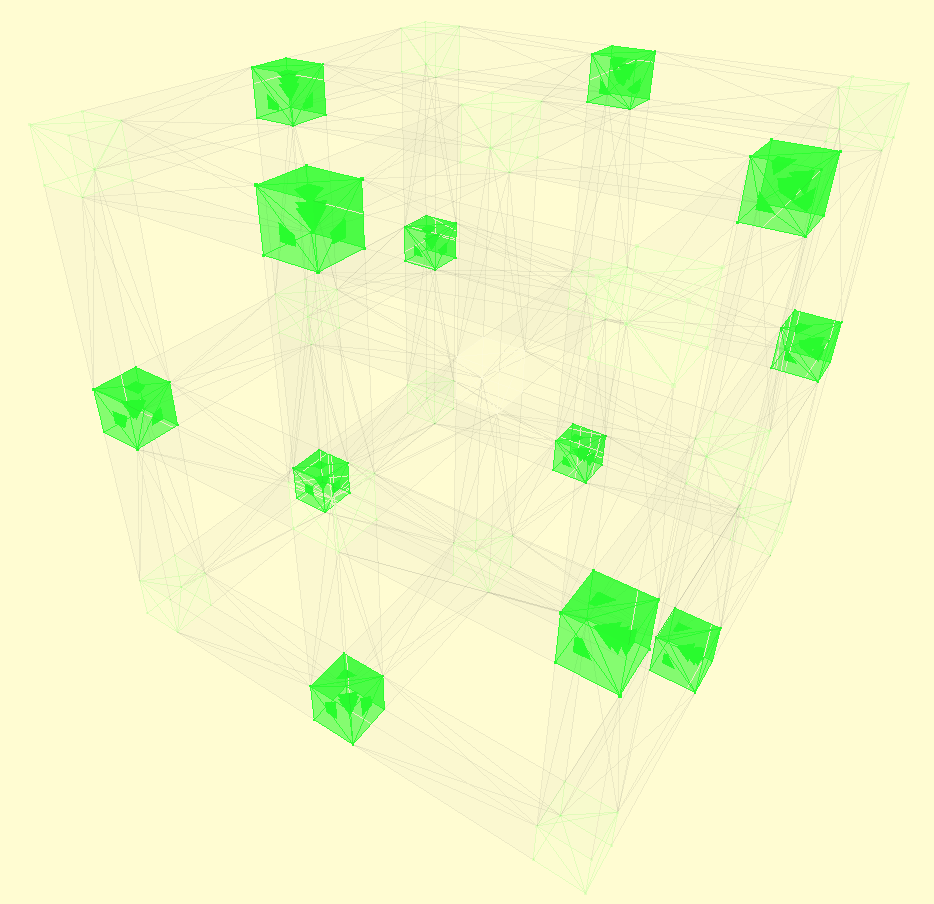}\phantom{AA}}\quad
\subfloat[\label{subfig:C6}C6]{\includegraphics[width=0.2\textwidth]{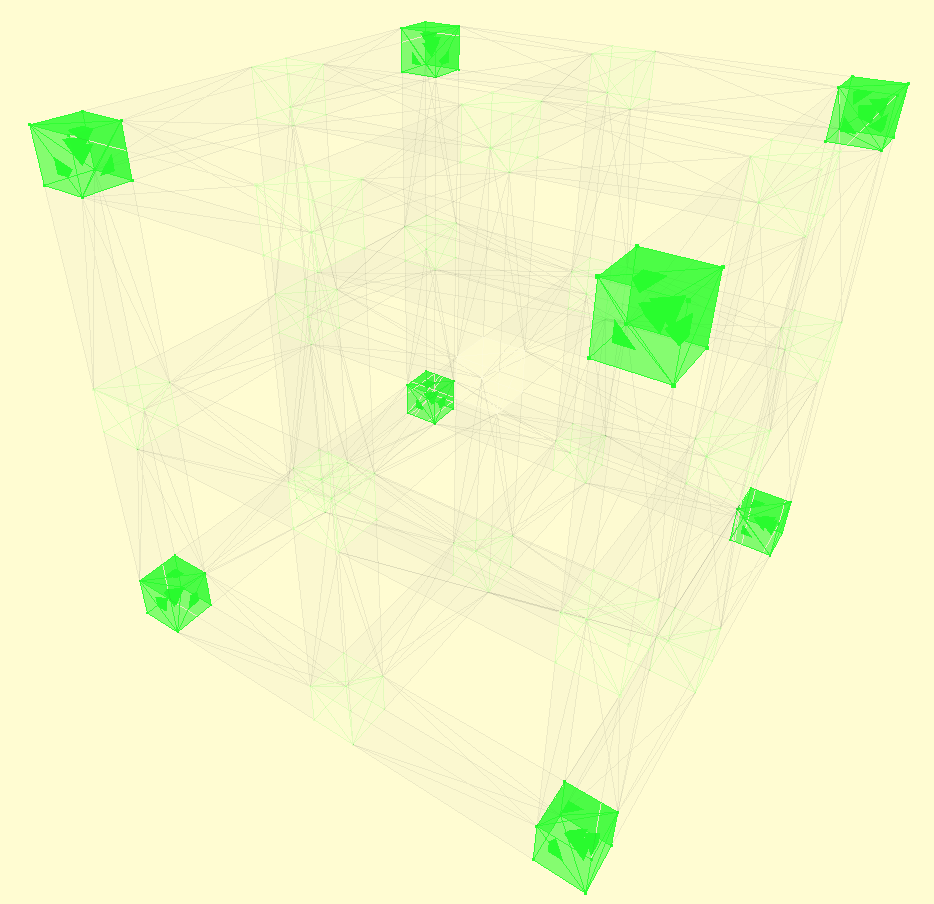}\phantom{AA}}\\
\caption{A maze with 27 rooms: 26 green and one white in the middle.}\label{fig:cube3x3}
\end{figure}

It is not difficult to find \slcsE{} formulas that distinguish the
various green classes. For example, the cells in~$\mathtt{C}4$ satisfy
$\phi_1 = \eta \mkern1mu (\mathbf{green} \lor \eta \mkern1mu
(\mathbf{corridor},\mathbf{white}), \mathbf{white})$, whereas no cell in
$\mathtt{C}2$ or~$\mathtt{C}6$ satisfies~$\phi_1$. 
To distinguish class~$\mathtt{C}2$ from $\mathtt{C}6$, one can observe that cells in~$\mathtt{C}2$ satisfy
$\phi_2 = \eta \mkern1mu (\mathbf{green} \lor \eta \mkern1mu
(\mathbf{corridor},\phi_1), \phi_1)$ whereas those in~$\mathtt{C}6$ do not
satisfy~$\phi_2$. Figure~\ref{fig:cube3x3more} shows the result of
\polylogica{} model checking for the formulas $\phi_1$ (see
Figure~\ref{subfig:phi1}) and~$\phi_2$ (see
Figure~\ref{subfig:phi2}).\footnote{All tests were performed on a
  workstation equipped with an Intel(R) Core(TM) i9-9900K CPU @ 3.60
  GHz (8~cores, 16~threads).}

\begin{figure}
\centering
\subfloat[\label{subfig:3DMaze}]{\includegraphics[width=0.2\textwidth]{maze3x3x3LC_full.png}\phantom{AA}}\quad
\subfloat[\label{subfig:phi1}$\phi_1$]{\includegraphics[width=0.2\textwidth]{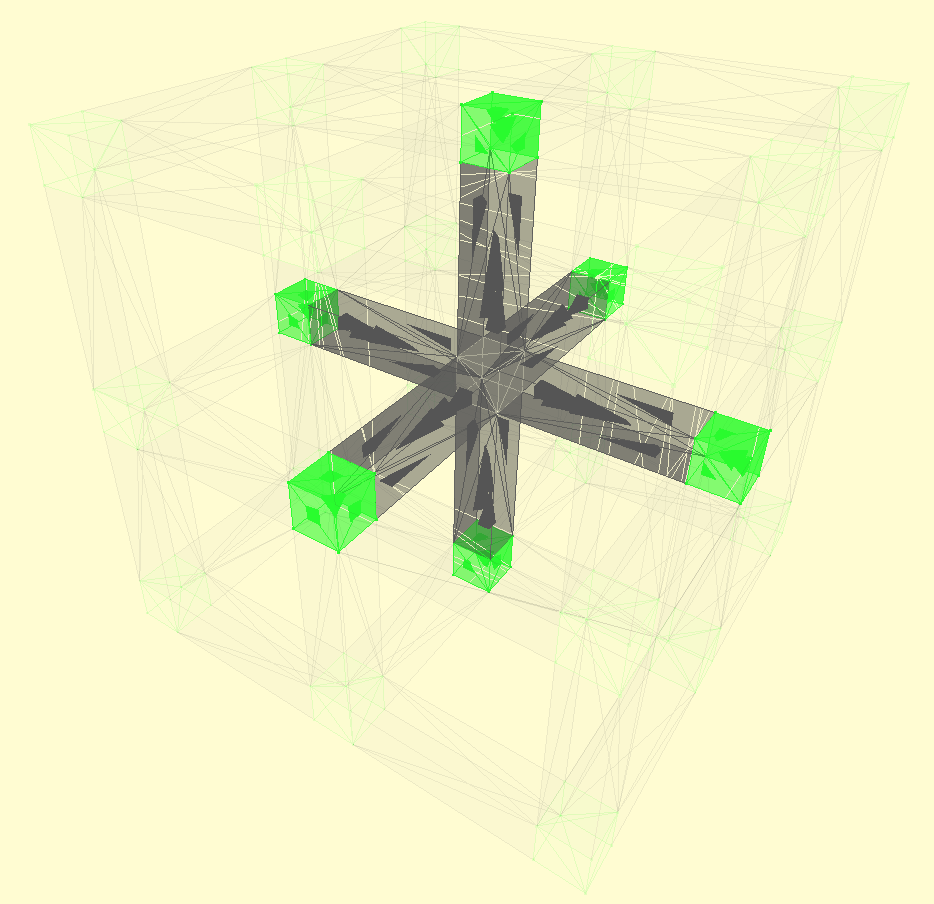}\phantom{AA}}\quad
\subfloat[\label{subfig:phi2}$\phi_2$]{\includegraphics[width=0.2\textwidth]{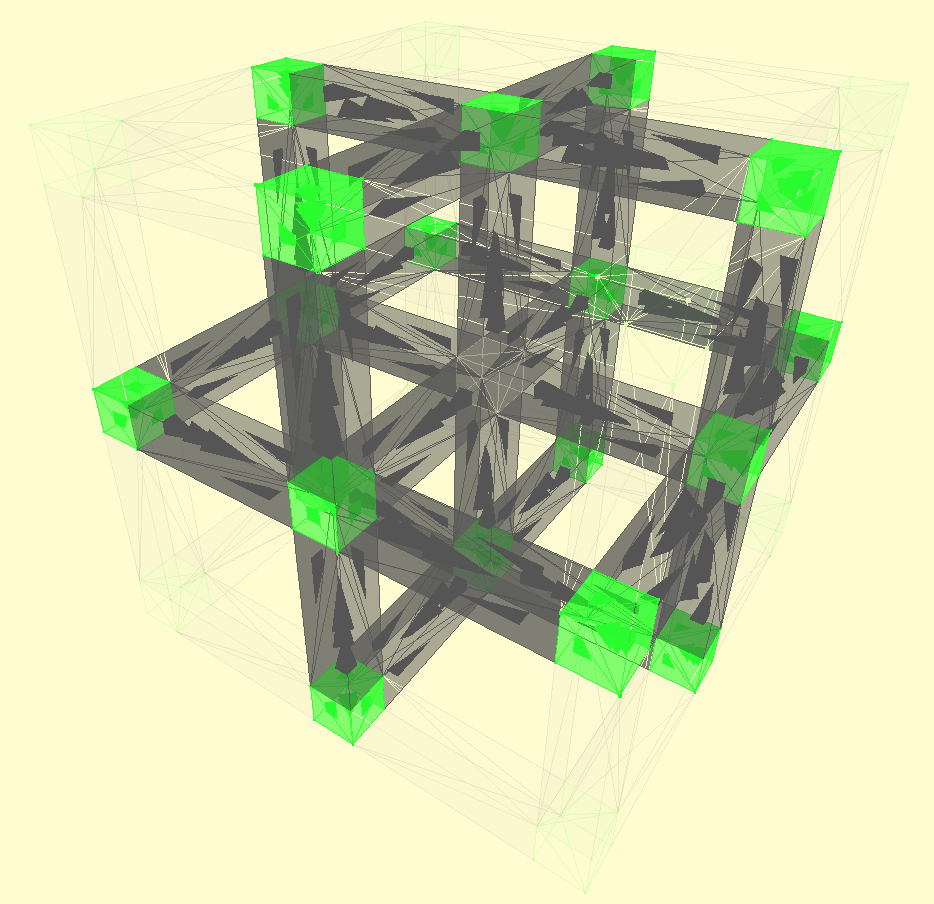}\phantom{AA}}\\
\caption{(\ref{subfig:3DMaze}) The 3D maze. Results of \polylogica{} model checking of the formulas $\phi_1$ (\ref{subfig:phi1})  and $\phi_2$ (\ref{subfig:phi2}) on the minimised model  as they are shown to the user by \polyvisualizer{} --- results are mapped back automatically by the procedure onto the full 3D maze. 
}\label{fig:cube3x3more}
\end{figure}

Table~\ref{tab:toolchain} provides a  detailed overview regarding the time
performance of the various components of the toolchain (see Figure~\ref{fig:toolchain})
 on four models of the maze of different sizes.\footnote{Maze 3x3x3 is shown in Figure~\ref{subfig:maze3D}, Maze 3x5x3 in Figure~\ref{subfig:c3x5x3} (in Appendix~\ref{apx:AdditionalExamples}), and Maze 3x5x4 in Figure~\ref{subfig:c3x5x4}.} 
 In each model all green rooms form the outer
frame of the maze and white rooms are positioned inside the maze. 
The table has one separate column for each  maze. 
The first horizontal block shows the number of cells and vertices for the models, as well as the
number of the equivalence classes.
The names of the components of the toolchain are listed in the first column of the
second horizontal block of the table. In the list two additional activities appear,
namely, loading of the model (\texttt{loadData}) and the production of
the equivalence classes and of the minimal Kripke model (\texttt{createJsonFiles} and
\texttt{createModelFile}, respectively). The remaining columns show the computing time of each component, in seconds.
The third block shows the model checking times for formulas $\phi_1$ and $\phi_2$, in the original as well as the minimal models.

Note the substantial reduction in size (several orders of magnitude) of the
minimised model, where the number of states corresponds to the number
of equivalence classes, compared to the full model (number of
cells). This leads to a similar reduction in model checking time (see
last two lines of Table~\ref{tab:toolchain}).
Clearly, the time for encoding (\texttt{poset2mcrl2}) and minimising (see~\texttt{ltsMinimise}) 
the model is very small, whereas there seems to be a bottleneck of computing time needed for the \mcrl{} procedure~\texttt{lps2lts}. However, the latter step may be avoided by
implementing the encoding directly into the binary \mcrl{} LTS
format. This requires usage of the \mcrl{} \texttt{C++}
\emph{application programming interface}, and is left to future work.

In summary, the considerable reduction of the models and their relative model checking times are very encouraging, also considering that the minimised model, once obtained, can be used for multiple model checking sessions.

\begin{table}
\begin{center}
\caption{Performance for 3D maze example. All times are in seconds.}\label{tab:toolchain}
\npdecimalsign{.}
\nprounddigits{2}
\resizebox{4.5in}{!}{
\begin{tabular}{|l|n{5}{2}|n{5}{2}|n{5}{2}|n{5}{2}|}
\hline \hline \rule{0pt}{11pt} 
&{Maze 3x3x3} & {Maze 3x5x3} & {Maze 3x5x4}& {Maze 5x5x5}\\
\hline\hline \rule{0pt}{11pt}\!\! 
Nr. of classes & \hfill{ 7} &\hfill{ 21}&\hfill {38} &\hfill {21}\\ 
Nr. of cells &\hfill { 2,619} &\hfill { 3,568}& \hfill {6,145}&\hfill{13,375}\\
Nr. of vertices &\hfill { 216}&\hfill { 288}& \hfill{480}&\hfill{1,000}\\
\hline\hline \hline\hline \rule{0pt}{11pt} 
 \!\!{\tt poly2poset} & 0.353486&0.344858&0.430944& 1.0970571\\
{\tt loadData} & 0.002311 &0.004278& 0.005411& 0.016074\\
{\tt poset2mcrl2} & 0.155092&0.299725&0.418008& 0.94669890 \\
{\tt mcrl2lps} & 1.706578&3.507069& 5.424817& 23.722838\\
{\tt lps2lpspp} & 0.241190 &0.410788& 0.574084 & 1.946931\\
{\tt findStates} & 0.171570&0.311257&0.405276& 4.177323 \\
{\tt renamelps} &  0.544962&0.946846& 1.341600&4.467334\\
{\tt lps2lts} & 21.412016 &78.257813&135.222544&794.327176\\
{\tt ltsMinimise} & 0.063677 &0.227126& 0.239059&0.350995\\
{\tt createJsonFiles} & 6.352409 &51.370290& 160.529789&587.993786\\
{\tt createModelFile} & 0.008330& 0.010408& 0.011927&0.026501\\
\hline\hline \rule{0pt}{11pt}
\!\!Model checking original model & 8.76 & 24.9 & 64.5 & 671.3 \\
Model checking minimised model & 0.02 & 0.03 & 0.03 & 0.03\\
\hline\hline
\end{tabular}
}
\end{center}
\end{table}

\section{Conclusions}
\label{sec:ConclusionsFW}

Polyhedral models are widely used in domains that exploit mesh processing such as 3D computer graphics. These models are typically huge, consisting of very many cells. Spatial model checking of such models is an interesting, novel approach to verify properties of such models and to visualise the results in a graphically appealing way. In previous work the polyhedral model checker \polylogica{} was developed for this purpose~\cite{Be+22}.

In~\cite{Be+22} simplicial bisimilarity was proposed for polyhedral models --- i.e. models of continuous space --- while  \plm-bisimilarity, the corresponding equivalence for cell-poset models --- discrete representations of polyhedral models --- was first introduced in~\cite{Ci+23c}.
In order to support large model reductions, in this paper the novel notions of weak  simplicial bisimilarity and  weak  \plm-bisimilarity have been presented, and the correspondence between the two has been studied. 
We have also presented \slcsE{,} a weaker version of the Spatial Logic for Closure Spaces on polyhedral models, and we have shown that simplicial bisimilarity enjoys the Hennessy-Milner property (Theorem~\ref{thm:HMPpoly}). Furthermore, we have shown that the property holds for \plm-bisimilarity on poset models and the interpretation of \slcsE{} on such models (Theorem~\ref{theo:PMbisEqSLCSEeq}).
\slcsE{} can be used in the geometric spatial model checker  \polylogica{} for checking spatial reachability properties of polyhedral models. 
Model checking results can be visualised by projecting them onto the original polyhedral structure, showing in a specific colour all the cells satisfying the property of interest.

In order to reduce model checking time and computing resources, we have proposed an effective procedure that computes the minimal model, modulo logical equivalence with respect to the logic \slcsE{,} of a polyhedral model. Such minimised models are also amenable to model checking with a variant of \polylogica{} dealing with general Kripke models and with the $\eta$ modality. 

The procedure has been formalised and proven correct. A prototype implementation of the procedure has been developed in the form of a toolchain, that also involves operations provided by the \mcrl{} toolset, to study the practical feasibility of the approach and to identify possible bottlenecks. 
We have also shown how the model checking results of the minimal model can be projected back onto the original polyhedral model. This provides a direct 3D visual inspection of the results through the polyhedra visualizer \polyvisualizer.

In future work we aim at a more sophisticated implementation of the procedure, possibly using in a more direct way the minimisation operations provided by \mcrl{} and integrating the various steps in the procedure. 
Such an implementation, would also enable us to experiment applying our methodology and supporting tools to real-world case studies.
On the theoretical side, an interesting issue that is beyond the scope of the present paper, and that 
we would like to address in future work, is the relationship between \slcsE, \slcsG, and~$\Diamond$. 
Finally, we would be interested in extending \slcsE/\slcsG{} with additional operators, for example those concerning notions of distance, and in applying our spatial model checking framework to a larger number of case studies.

\section*{Acknowledgment}
Research partially supported by bilateral project between CNR (Italy)
and SRNSFG (Georgia) ``Model Checking for Polyhedral Logic''
(\#CNR-22-010); European Union -- Next GenerationEU -- National Recovery
and Resilience Plan (NRRP), Investment~1.5 Ecosystems of Innovation,
Project “Tuscany Health Ecosystem” (THE), CUP: B83C22003930001;
European Union -- Next-GenerationEU -- National Recovery and Resilience
Plan (NRRP) – MISSION~4 COMPONENT~2, INVESTMENT N.~1.1, CALL PRIN 2022
D.D.\ 104 02-02-2022 – (Stendhal) CUP N.~B53D23012850006; MUR project
PRIN 2020TL3X8X ``T-LADIES''; CNR project "Formal Methods in Software
Engineering 2.0", CUP B53C24000720005; Shota Rustaveli National
Science Foundation of Georgia grant
\#FR\nobreakdash-22\nobreakdash-6700.

\bibliographystyle{alphaurl}
\bibliography{main}

\appendix

\pagebreak
\section{Detailed Proofs}\label{apx:DetailedProofs}

\subsection{Proof of Lemma~\ref{lem:pm2ud}}\label{apx:prf:lem:pm2ud}$ $\\

\noindent
{\bf Lemma~\ref{lem:pm2ud}.}
{\em 
Given a reflexive Kripke frame $(W,R)$ and  a \plm-path $\pi:[0;\ell]\to W$,
there is a \upd-path $\pi':[0;\ell']\to W$, for some $\ell'$, and a 
total, surjective, monotonic non-decreasing function $f:[0;\ell'] \to [0;\ell]$ such that 
$\pi'(j)=\pi(f(j))$ for all $j\in [0;\ell']$.
}

\begin{proof}
We proceed by induction on the length $\ell$ of \plm-path $\pi$.\\
{\bf Base case:} $\ell =2$.\\
In this case, by definition of \plm-path, we have $R(\pi(0),\pi(1))$ and $\cnv{R}(\pi(1),\pi(2))$, which, by definition of \upd-path, implies that $\pi$ itself is an \upd-path and $f:[0;\ell] \to [0;\ell]$ is just the identity function.\\

\noindent
{\bf Induction step.} We assume the assertion holds for all \plm-paths of length $\ell$ and we prove it for $\ell+1$.
Let $\pi:[0;\ell+1] \to W$ be a \plm-path. 
Then $\cnv{R}(\pi(\ell),\pi(\ell+1))$, since $\pi$ is a $\pm$-path.
We consider the following cases:\\
{\bf Case A:} $\cnv{R}(\pi(\ell-1),\pi(\ell))$ and  $\cnv{R}(\pi(\ell),\pi(\ell+1))$.\\
In this case, consider the prefix $\pi_1 = \pi | [0;\ell]$ of $\pi$, noting that $\pi_1$ is a \plm-path of length $\ell$. By the Induction Hypothesis there is an \upd-path $\pi'_1$ of some length $\ell'_1$
and a total, surjective, monotonic non-decreasing function $g:[0;\ell'_1] \to [0;\ell]$ such that 
$\pi'_1(j)=\pi_1(g(j))=\pi(g(j))$ for all $j\in [0;\ell'_1]$. 
Note that $\pi'_1(\ell'_1) = \pi(\ell)$ so that  the sequentialisation
of $\pi'_1$ with the two-element path  $(\pi(\ell),\pi(\ell+1))$ is well-defined.
Consider  path $\pi'=(\pi'_1 \cdot (\pi(\ell),\pi(\ell+1)))\leftarrow \ell'_1$, of length $\ell'_1+2$
consisting of $\pi'_1$ followed by $\pi(\ell)$ followed in turn by $\pi(\ell+1)$. 
In other words, $\pi' = (\pi'_1(0) \ldots \pi'_1(\ell'_1),\pi(\ell),\pi(\ell+1))$, with $\pi'_1(\ell'_1) = \pi(\ell)$ --- recall that $R$ is reflexive.
It is easy to see that $\pi'$ is an \upd-path and that function $f:[0;\ell'_1+2] \to [0;\ell+1]$, with
$f(j)=g(j)$ for $j\in [0;\ell'_1]$, $f(\ell'_1+1)=\ell$ and $f(\ell'_1+2)=\ell+1$, is 
total, surjective, and monotonic non-decreasing.\\
{\bf Case B:} $R(\pi(\ell-1),\pi(\ell))$ and $\cnv{R}(\pi(\ell),\pi(\ell+1))$.\\
In this case the prefix $\pi | [0;\ell]$ of $\pi$ is {\em not} a \plm-path.
We then consider the path consisting of prefix $\pi|[0;\ell-1]$ where we add a copy
of $\pi(\ell-1)$, i.e. the path $\pi_1=(\pi|[0;\ell-1])\leftarrow (\ell-1)$ --- we can do that because $R$ is reflexive.
Note that $\pi_1$ is a \plm-path and has length $\ell$.
By the Induction Hypothesis there is an \upd-path $\pi'_1$ of some length $\ell'_1$
and a total, surjective, monotonic non-decreasing function $g:[0;\ell'_1] \to [0;\ell]$ such that 
$\pi'_1(j)=\pi_1(g(j))=\pi(g(j))$ for all $j\in [0;\ell'_1]$. 
Consider path $\pi'=\pi'_1 \cdot (\pi(\ell-1),\pi(\ell), \pi(\ell+1))$, of length $\ell'_1 +2$,
that is well defined since $\pi'_1(\ell'_1)=\pi(\ell-1)$ by definition of $\pi_1$.
In other words, $\pi'=(\pi'_1(0),\ldots,\pi'_1(\ell'_1),\pi(\ell),\pi(\ell+1))$, with $\pi'_1(\ell'_1) = \pi(\ell-1)$.
Path $\pi'$ is an \upd-path. In fact $\pi'|[0;\ell'_1]=\pi'_1$ is an \upd-path. Furthermore, 
$\pi'(\ell'_1)=\pi(\ell-1)$, $R(\pi(\ell-1),\pi(\ell))$, $\cnv{R}(\pi(\ell),\pi(\ell+1))$ and $\pi(\ell+1)=\pi'(\ell'_1 +2)$.
Finally, function $f:[0;\ell'_1+2] \to [0;\ell+1]$, with
$f(j)=g(j)$ for $j\in [0;\ell'_1]$, $f(\ell'_1+1)=\ell$ and $f(\ell'_1+2)=\ell+1$, is 
total, surjective, and monotonic non-decreasing.
\end{proof}

\subsection{Proof of Lemma~\ref{lem:d2ud}}\label{apx:prf:lem:d2ud}$ $\\

\noindent
{\bf Lemma~\ref{lem:d2ud}.}
{\em
Given a reflexive Kripke frame  $(W,R)$ and  a \dwn-path $\pi:[0;\ell]\to W$,
there is an \upd-path $\pi':[0;\ell']\to W$, for some $\ell'$, and a 
total, surjective, monotonic non-decreasing function $f:[0;\ell'] \to [0;\ell]$ such that 
$\pi'(j)=\pi(f(j))$ for all $j\in [0;\ell']$.
}

\begin{proof}
The proof is carried out by induction on the length $\ell$ of $\pi$.\\
{\bf Base case.} $\ell=1$.
Suppose $\ell=1$, i.e. $\pi:[0;1] \to W$ with $\cnv{R}(\pi(0),\pi(1))$. Then let 
$\pi':[0;2]\to W$ be such that $\pi'(0)=\pi'(1)=\pi(0)$ and $\pi'(2)=\pi(1)$ --- we can do that since $R$ is reflexive --- and
$f:[0;2] \to [0;1]$ be such that $f(0)=f(1)=0$ and $f(2)=1$.
Clearly $\pi'$ is an \upd-path and $\pi'(j)=\pi(f(j))$ for all $j\in [0;2]$.\\
{\bf Induction step.} We assume the assertion  holds for all \dwn-paths of length $\ell$ and we prove it for $\ell+1$.
Let $\pi:[0;\ell+1] \to W$ a \dwn-path and suppose the assertion  holds for all \dwn-paths of length $\ell$. 
In particular, it holds for $\pi\uparrow 1$, i.e., there is an \upd-path $\pi''$ 
of some length $\ell''$ with $\pi''(0) = \pi(1)$, and 
total, monotonic non-decreasing surjection $g:[0;\ell'']\to W$ such that 
$\pi''(j) = \pi(g(j))$ for all $j\in [0;\ell'']$.
Suppose $R(\pi(0),\pi(1))$ does not hold. Then, since $R$ is reflexive, we let $\pi'= (\pi(0),\pi(0),\pi(1))\cdot \pi''$ and
$f :[0;\ell''+2] \to [0;\ell+1]$ with $f(0)=f(1)=0$ and
$f(j)=g(j-2)$ for all $j\in[2;\ell''+2]$. 
If instead $R(\pi(0),\pi(1))$, then we let $\pi'= (\pi(0),\pi(1),\pi(1))\cdot \pi''$ and
$f :[0;\ell''+2] \to [0;\ell+1]$ with $f(0)=0, f(1)=1$ and
$f(j)=g(j-2)$ for all $j\in[2;\ell''+2]$.
\end{proof}

\subsection{Proof of Lemma~\ref{lem:d2plm}}\label{apx:prf:lem:d2plm}$ $\\

\noindent
{\bf Lemma~\ref{lem:d2plm}.}
{\em 
Given a reflexive Kripke frame  $(W,R)$ and a \dwn-path $\pi:[0;\ell]\to W$,
there is a \plm-path $\pi':[0;\ell'']\to W$, for some $\ell'$, and a 
total, surjective, monotonic, non-decreasing function $f:[0;\ell'] \to [0;\ell]$ with  
$\pi'(j)=\pi(f(j))$ for all $j\in [0;\ell']$.
}

\begin{proof}
The assertion  follows directly from Lemma~\ref{lem:d2ud} since every \upd-path is also a \plm-path.
\end{proof}

\subsection{Proof of Lemma~\ref{lem:etgaCorrectG}}\label{apx:Prf:lem:etgaCorrectG}$ $\\

\noindent
{\bf Lemma~\ref{lem:etgaCorrectG}.}
{\em
Let $\calP=(|K|,\peval{\calP})$ be a polyhedral model, $x \in |K|$ and $\form$ a \slcsE{} formula.
Then $\calP,x \models \form$ iff $\calP,x \models \etga(\form)$.
}

\begin{proof}
By induction on the structure of $\form$. We consider only the case $\eta(\form_1,\form_2)$.
Suppose $\calP,x \models \eta(\form_1,\form_2)$. By definition there is a topological path 
$\pi$ such that $\calP,\pi(1) \models \form_2$ and 
$\calP,\pi(r) \models \form_1$ for all $r \in [0,1)$. By the Induction Hypothesis this is the same to say that $\calP,\pi(1) \models \etga(\form_2)$ and 
$\calP,\pi(r) \models \etga(\form_1)$ for all $r \in [0,1)$, i.e.
$\calP,x \models \etga(\form_1)$, 
$\calP,\pi(1) \models \etga(\form_2)$ and 
$\calP,\pi(r) \models \etga(\form_1)$ for all $r \in (0,1)$. In other words, we have
$\calP,x \models \etga(\form_1) \land \gamma(\etga(\form_1), \etga(\form_2))$ that,
by Definition~\ref{def:etga} on page~\pageref{def:etga} means $\calP,x \models \etga(\eta(\form_1,\form_2))$.

Suppose now $\calP,x \models \etga(\eta(\form_1,\form_2))$, i.e.
$\calP,x \models \etga(\form_1) \land \gamma(\etga(\form_1), \etga(\form_2))$, by Definition~\ref{def:etga} on page~\pageref{def:etga}.
Since $\calP,x \models \gamma(\etga(\form_1), \etga(\form_2))$, there is a path $\pi$
such that $\calP,\pi(1) \models \etga(\form_2)$ and 
$\calP,\pi(r) \models \etga(\form_1)$ for all $r \in (0,1)$. Using the Induction Hypothesis
we know the following holds:
$\calP,x \models  \form_1$, $\calP,\pi(1) \models \form_2$, and
$\calP,\pi(r) \models \form_1$ for all $r \in (0,1)$, i.e.
$\calP,\pi(1) \models \form_2$ and $\calP,\pi(r) \models \form_1$ for all $r \in [0,1)$.
So, we get $\calP,x \models \eta(\form_1,\form_2)$.
\end{proof}

\subsection{Proof concerning the example of Remark~\ref{rem:weakGweakerG}}
\label{apx:prf:rem:weakGweakerG}$ $\\

The assertion  can be proven by induction on the structure of formulas.
The case for proposition letters, negation  and conjunction are straightforward and omitted.

Suppose $\calP_{\ref{fig:AltTriAndPoset}}, A \models \eta(\form_1, \form_2)$. 
Then there is a topological path $\pi_A:[0,1] \to P_{\ref{fig:AltTriAndPoset}}$ from $A$ such that
$\calP_{\ref{fig:AltTriAndPoset}}, \pi_A(1) \models \form_2$ and $\calP_{\ref{fig:AltTriAndPoset}}, \pi_A(r)\models \form_1$ for all 
$r \in [0,1)$. 
Since $\calP_{\ref{fig:AltTriAndPoset}}, A \models \form_1$, by the Induction Hypothesis, we have that $\calP_{\ref{fig:AltTriAndPoset}}, x \models \form_1$ for all $x \in \relint{ABC}$.
For each $x\in\relint{ABC}$, define $\pi_x: [0,1] \to P_{\ref{fig:AltTriAndPoset}}$ as follows, for arbitrary $v \in (0,1)$:
\[
\pi_x(r)=
\left\{
\begin{array}{l}
\frac{r}{v}A + \frac{v-r}{v}x, \mbox{ if } r\in [0,v),\\\\
\pi_A(\frac{r-v}{1-v}), \mbox{ if }r\in [v,1].
\end{array}
\right.
\]
Function $\pi_x$ is continuous. Furthermore, for all $y\in [0,v)$, we have that $\calP_{\ref{fig:AltTriAndPoset}}, \pi_x(y) \models \form_1$,  since $\pi_x(y) \in \relint{ABC}$. Also, for all $y\in [v,1)$ we have that $\calP_{\ref{fig:AltTriAndPoset}}, \pi_x(y) \models \form_1$,  since $\pi_x(y) = \pi_A(\frac{y-v}{1-v})$,
$0\leq \frac{y-v}{1-v}<1$ and for $y \in [0,1)$ we have that $\calP_{\ref{fig:AltTriAndPoset}},\pi_A(y), \models \form_1$.  Thus
$\calP_{\ref{fig:AltTriAndPoset}},\pi_x(r) \models \form_1$ for all $r \in [0,1)$. Finally, $\pi_x(1)=\pi_A(1)$ and 
$\calP_{\ref{fig:AltTriAndPoset}}, \pi_A(1) \models \form_2$ by hypothesis. 
Thus, $\pi_x$ is a topological path that witnesses $\calP_{\ref{fig:AltTriAndPoset}}, x \models \eta(\form_1, \form_2)$.

The proof of the converse is similar, using instead function $\pi_A: [0,1] \to P_{\ref{fig:AltTriAndPoset}}$ defined as follows, for arbitrary $v \in (0,1)$:
\[
\pi_A(r)=
\left\{
\begin{array}{l}
\frac{r}{v}p + \frac{v-r}{v}A, \mbox{ if } r\in [0,v),\\\\
\pi_p(\frac{r-v}{1-v}), \mbox{ if }r\in [v,1].
\end{array}
\right.
\]

\subsection{Proof of Proposition~\ref{prop:interchangeableE}}\label{apx:prf:prop:interchangeableE}$ $\\

\noindent
{\bf Proposition~\ref{prop:interchangeableE}.}
{\em
Given a finite poset model $\calF = (W,\preccurlyeq,\peval{\calF})$,
  $w \in W$, and  \slcsE{} formulas $\form_1$ and $\form_2$, the
  following statements are equivalent:
\begin{enumerate}
\item
There exists a \plm-path $\pi:[0;\ell] \to W$ for some $\ell$ with
$\pi(0)=w$, 
$\calF, \pi(\ell)\models \form_2$ and 
$\calF, \pi(i)\models \form_1$  for all $i\in [0;\ell)$.
\item
There exists a \dwn-path $\pi:[0;\ell'] \to W$ for some $\ell'$ with
$\pi(0)=w$, 
$\calF, \pi(\ell')\models \form_2$ and 
$\calF, \pi(i)\models \form_1$  for all $i\in [0;\ell)$.
\end{enumerate}
}
\begin{proof}
The equivalence of statements (\ref{enu:plmE}) and (\ref{enu:dwnE}) follows directly 
from Lemma~\ref{lem:d2plm} and the fact that \plm-paths are also \dwn-paths.
\end{proof}

\subsection{Proof of Proposition~\ref{prop:chiE}}\label{apx:prf:prop:chiE}$ $\\

\noindent
{\bf Proposition~\ref{prop:chiE}.}
{\em
Given a finite poset model $\calF=(W,\preccurlyeq,\peval{\calF})$, for $w_1,w_2 \in W$, it holds that 
\[
  \calF,w_2 \models \chi(w_1) \mbox{ if and only if } w_1 \slcsEeq w_2.
\]
}
\begin{proof}
Suppose $w_1 \not\slcsEeq w_2$, then we have $\calF,w_2 \not\models \delta_{w_1,w_2}$, and so
  $\calF,w_2 \not\models \bigwedge_{w \in W} \: \delta_{w_1,w}$.
  If, instead, $w_1 \slcsEeq w_2$, then we have: $\delta_{w_1,w_1} \equiv \delta_{w_1,w_2} \equiv \ltrue  $ by definition, since $w_1 \slcsEeq w_1$ and $w_1 \slcsEeq w_2$. Moreover, for any other $w$, we have that, in any case,
  $\calF,w_1\models \delta_{w_1,w}$ holds and since $w_1 \slcsEeq w_2$, also $\calF,w_2\models \delta_{w_1,w}$ holds.
  So, in conclusion, $\calF,w_2\models \bigwedge_{w \in W} \: \delta_{w_1,w}$.
\end{proof}

\subsection{Proof of Lemma~\ref{lem:etgaCorrectE}}\label{apx:prf:lem:etgaCorrectE}$ $\\

\noindent
{\bf Lemma~\ref{lem:etgaCorrectE}.}
{\em
Let $\calF=(W,\preccurlyeq,\peval{\calF})$ be a 
finite poset model, 
$w \in W$ and $\form$ a \slcsE{} formula.
Then $\calF,w \models \form$ iff $\calF,w \models \etga(\form)$.
}

\begin{proof}
Similar to that of Lemma~\ref{lem:etgaCorrectG}, but with reference to the finite poset intepretation of the logic.
  \end{proof}

\subsection{Proof concerning the example of Remark~\ref{rem:weakPMeakerPM}}\label{apx:prf:rem:weakPMeakerPM}$ $\\

We prove the assertion  by induction on the structure of formulas. The case for atomic proposition letters, negation  and conjunction are straightforward and omitted.
Suppose $\calF,\relint{A} \models \eta(\form_1, \form_2)$. Then, there is a \plm-path $\pi$ of 
some length $\ell\geq2$ such that
$\pi(0) = \relint{A}$, $\pi(\ell) \models \form_2$ and $\pi(i)\models \form_1$ for all 
$i \in [0;\ell)$. Since $\calF,\relint{A} \models \form_1$, by the Induction Hypothesis, we have that $\calF,\relint{ABC} \models \form_1$.
Consider then path $\pi'= (\relint{ABC},\relint{ABC},\relint{A})\cdot \pi$. Path $\pi'$ is a \plm-path and it witnesses $\calF, \relint{ABC} \models \eta(\form_1, \form_2)$.

Suppose now $\calF, \relint{ABC} \models \eta(\form_1, \form_2)$ and let $\pi$ be a \plm-path witnessing it.
Then, path $(\relint{A},\relint{ABC},\relint{ABC}) \cdot \pi$ is a \plm-path witnessing $\calF, \relint{A} \models \eta(\form_1, \form_2)$.

\subsection{Proof of Lemma~\ref{lem:VB}}\label{apx:prf:lem:VB}$ $\\

The proof of the lemma uses a similar result, for the $\gamma$ operator, that we 
have already proven in~\cite{Be+22} namely:\\

{\bf Theorem 4.4 of~\cite{Be+22}}. 
{\em Let $\calP=(P,\peval{\calP})$ be a polyhedral model
and $x \in P$. Then, for every formula $\form$ of \slcsG{} we have that:
$\calP,x \models \form$ if and only if $\map(\calP),\map(x) \models \form$.\\
}

\noindent
{\bf Lemma~\ref{lem:VB}.}
{\em 
Given a polyhedral model $\calP=(|K|,\peval{\calP})$,
for all $x\in |K|$ and formulas $\form$ of \slcsE{}
the following holds: $\calP,x \models \form$ if and only if $\map(\calP),\map(x) \models \etga(\form)$.
}

\begin{proof}
The proof is by induction on the structure of $\form$. 
We consider only the case $\eta(\form_1,\form_2)$.
Suppose $\calP,x \models \eta(\form_1,\form_2)$. 
By Lemma~\ref{lem:etgaCorrectG}  we get 
$\calP,x \models \etga(\eta(\form_1,\form_2))$
and then, by Definition~\ref{def:etga}, we have $\calP,x \models \etga(\form_1) \land \gamma(\etga(\form_1),\etga(\form_2))$, 
that is $\calP,x \models \etga(\form_1)$ and  $\calP,x \models\gamma(\etga(\form_1),\etga(\form_2))$. Again by Lemma~\ref{lem:etgaCorrectG} on page~\pageref{lem:etgaCorrectG}, we get also 
$\calP,x \models \form_1$ and so, by the Induction Hypothesis, we have
$\map(\calP),\map(x) \models \etga(\form_1)$.
Furthermore, by Theorem 4.4 of~\cite{Be+22} we also get 
$\map(\calP),\map(x) \models \gamma(\etga(\form_1),\etga(\form_2))$.
Thus we get $\map(\calP),\map(x) \models \etga(\form_1)\land \gamma(\etga(\form_1),\etga(\form_2))$, that is
$\map(\calP),\map(x) \models \etga(\eta(\form_1,\form_2))$.\\
Suppose now $\map(\calP),\map(x) \models \etga(\eta(\form_1,\form_2))$.
This means $\map(\calP),\map(x) \models \etga(\form_1) \land \gamma(\etga(\form_1),\etga(\form_2))$,
that is $\map(\calP),\map(x) \models \etga(\form_1)$  and 
$\map(\calP),\map(x) \models \gamma(\etga(\form_1),\etga(\form_2))$.
By the Induction Hypothesis we get that $\calP,x \models \form_1$.
Furthermore, by Theorem 4.4 of~\cite{Be+22} we also get
$\calP,x \models \gamma(\etga(\form_1),\etga(\form_2))$.
This means that there is topological path $\pi$ such that 
$\calP,\pi(1) \models \etga(\form_2)$ and $\calP,\pi(r) \models \etga(\form_1)$
for all $r \in (0,1)$. 
Using Lemma~\ref{lem:etgaCorrectG}  we also get
$\calP,\pi(1) \models \form_2$ and $\calP,\pi(r) \models \form_1$
for all $r \in (0,1)$ and since also $\calP,x \models \form_1$ (see above), we
get $\calP,\pi(1) \models \form_2$ and $\calP,\pi(r) \models \form_1$
for all $r \in [0,1)$, that is
$\calP,x \models \eta (\form_1,\form_2)$.
  \end{proof}

\subsection{Proof of Theorem~\ref{theo:calMPresForm}}\label{apx:prf:theo:calMPresForm}$ $\\

\noindent
{\bf Theorem~\ref{theo:calMPresForm}.}
{\em 
Given a polyhedral model $\calP=(|K|,\peval{\calP})$,
for all $x\in |K|$ and formulas $\form$ of \slcsE{}
it holds that: $\calP,x \models \form$ if and only if $\map(\calP),\map(x) \models \form$.
}

\begin{proof}
Using Lemma~\ref{lem:VB}, we know that 
$\calP,x \models \form$ if and only if $\map(\calP),\map(x) \models \etga(\form)$.
Moreover, by Lemma~\ref{lem:etgaCorrectE}, we know that 
$\map(\calP),\map(x) \models \etga(\form)$ if and only if
$\map(\calP),\map(x) \models \form$, which brings us to the result.
  \end{proof}

\subsection{Proof of Lemma~\ref{lem:dPExists}}\label{apx:prf:lem:dPExists}$ $\\

\noindent
{\bf Lemma~\ref{lem:dPExists}.}
{\em Given a finite poset model $\calF=(W,\preccurlyeq, \peval{\calF})$ and  
weak \plm-bisimulation $Z \subseteq W \times W$, for all $w_1,w_2$ such that $Z(w_1,w_2)$, the following holds:
for each  \dwn-path $\pi_1:[0;k_1] \to W$ from $w_1$ 
there is a \dwn-path $\pi_2:[0;k_2] \to W$ from $w_2$ 
such that $Z(\pi_1(k_1),\pi_2(k_2))$ and 
for each $j\in [0;k_2)$ there is $i\in [0;k_1)$ such that $Z(\pi_1(i),\pi_2(j))$.
}

\begin{proof}
Let $\pi_1:[0;k_1] \to W$ be a \dwn-path from $w_1$. By Lemma~\ref{lem:d2ud} on page~\pageref{lem:d2ud} we know that
there is an \upd-path $\hat{\pi}_1:[0;2h] \to W$ and total, monotonic non-decreasing surjection
$f:[0;2h] \to [0;k_1]$ such that $\hat{\pi}_1(j)=\pi_1(f(j))$  for all $j\in [0;2h]$. Furthermore, by Lemma~\ref{lem:dPExists4udP} below,
we know that
there is a \dwn-path $\pi_2:[0;k_2] \to W$ from $w_2$ 
such that $Z(\hat{\pi}_1(2h),\pi_2(k_2))$ and 
for each $j\in [0;k_2)$ there is $i\in [0;2h)$ such that $Z(\hat{\pi}_1(i),\pi_2(j))$.
In addition, 
$\hat{\pi}_1(0)=\pi_1(0)=w_1$,  $Z(\pi_1(k_1),\pi_2(k_2))$ since $Z(\hat{\pi}_1(2h),\pi_2(k_2))$ and $\hat{\pi}_1(2h)=\pi_1(k_1)$. Finally, 
for each $j\in [0;k_2)$ there is $i \in [0;k_1)$ such that $Z(\pi_1(i),\pi_2(j))$,
since there is $n\in [0;2h)$ such that $Z(\hat{\pi}_1(n),\pi_2(j))$ and $f(n)=i$ for some $i \in [0;k_1)$.
\end{proof}

\begin{lem}\label{lem:dPExists4udP}
Given a finite poset model $\calF=(W,\preccurlyeq, \peval{\calF})$ and  a
weak \plm-bisimulation $Z \subseteq W \times W$, for all $w_1,w_2$ such that $Z(w_1,w_2)$, the following holds:
for each  \upd-path $\pi_1:[0;2h] \to W$ from $w_1$ 
there is a \dwn-path $\pi_2:[0;k] \to W$ from $w_2$ 
such that $Z(\pi_1(2h),\pi_2(k))$ and 
for each $j\in [0;k)$ there is $i\in [0;2h)$ such that $Z(\pi_1(i),\pi_2(j))$.
\end{lem}

\begin{proof}
We prove the assertion  by induction on $h$.\\
{\bf Base case.} $h=1$.\\
If $h=1$, the assertion  follows directly from Definition~\ref{def:WPLMBis} on page~\pageref{def:WPLMBis}
where $w_1= \pi_1(0), u_1=\pi_1(1)$ and $d_1=\pi_1(2)$.\\
{\bf Induction step.} We assume the assertion  holds for \upd-paths of length $2h$ or less and we prove it for \upd-paths of length $2(h+1)$. \\
Suppose $\pi_1$ is a \upd-path of length $2h+2$ and consider \upd-path $\pi'_1=\pi_1|[0;2h]$.
By the Induction Hypothesis, we know that there is a \dwn-path $\pi'_2:[0;k']\to W$ from $w_2$
such that $Z(\pi'_1(2h),\pi'_2(k'))$ and for each $j\in[0;k')$ there is $i\in [0;2h)$ such that
$Z(\pi'_1(i),\pi'_2(j))$. Clearly, this means that $Z(\pi_1(2h),\pi'_2(k'))$ and for each $j\in[0;k')$ there is $i\in [0;2h)$ such that $Z(\pi_1(i),\pi'_2(j))$.
Furthermore, since $Z(\pi_1(2h),\pi'_2(k'))$ and $Z$ is a weak \plm-bisimulation, 
we also know that there is a \dwn-path 
$\pi''_2:[0;k'']\to W$ from $\pi'_2(k')$ such that $Z(\pi_1(2h+2),\pi''_2(k''))$ and for each $j\in[0;k'')$
there is $i\in [2h;2h+2)$ such that $Z(\pi_1(i),\pi'_2(j))$. Let $\pi_2:[0;k'+k'']\to W$ be defined as
$\pi_2= \pi'_2 \cdot \pi''_2$. Clearly $\pi_2$ is a \dwn-path, since so is $\pi''_2$.
Furthermore $Z(\pi_1(2h+2),\pi_2(k'+k''))$ since $Z(\pi_1(2h+2),\pi''_2(k''))$ and $\pi''_2(k'')=\pi_2(k'+k'')$. Finally, it is straightforward to check  for all $j\in [0;k'+k'')$ 
there is $i\in [0;2h+2)$ such that $Z(\pi_1(i),\pi_2(j))$.
\end{proof}

\subsection{Proof of Lemma~\ref{lem:dPtoTP}}\label{apx:prf:lem:dPtoTP}$ $\\

\noindent
{\bf Lemma~\ref{lem:dPtoTP}.}
{\em
Given a polyhedral model $\calP=(|K|,\peval{\calP})$, and associated cell poset model $\map(\calP)=(W,\preccurlyeq,\peval{\map(\calP)})$, for any  
\dwn-path  $\pi:[0;\ell] \to W$, 
there is a topological path $\pi':[0,1] \to |K|$ such that: (i) $\map(\pi'(0))=\pi(0)$, (ii) $\map(\pi'(1))=\pi(\ell)$, and
(iii) for all $r \in (0,1)$ there is $i<\ell$ such that $\map(\pi'(r))=\pi(i)$.
}

\begin{proof}
Since $\pi$ is a \dwn-path, we have that either $\closure_T(\map^{-1}(\pi(k-1))) \sqsubseteq \closure_T(\map^{-1}(\pi(k)))$  
or $\closure_T(\map^{-1}(\pi(k))) \sqsubseteq \closure_T(\map^{-1}(\pi(k-1)))$,  
for each $k\in (0;\ell]$\footnote{We recall here that $\sigma_1 \sqsubseteq \sigma_2$ iff 
$\relint{\sigma_1} \preccurlyeq \relint{\sigma_2}$ and that $\sigma = \closure_T(\relint{\sigma})$.}. 
It follows that there is a continuous map 
$\pi'_k:[\frac{k-1}{\ell},\frac{k}{\ell}]\to |K|$ such that, in the first case, 
$\map(\pi'_k(\frac{k-1}{\ell})) = \pi(k-1)$ and
$\pi'_k((\frac{k-1}{\ell},\frac{k}{\ell}]) \subseteq \closure_T(\map^{-1}(\pi(k)))$, 
while in the second case, 
$\pi'_k([\frac{k-1}{\ell},\frac{k}{\ell})) \subseteq \closure_T(\map^{-1}(\pi(k-1)))$ and 
$\map(\pi'_k(\frac{k}{\ell}))=\pi(k)$.
In fact $\pi'_k$ can be realised as a linear bijection to the line segment connecting the barycenters  
in the corresponding cell, either in $\map^{-1}(\pi(k))$ or in $\map^{-1}(\pi(k-1))$, respectively.

For each $k\in(0;\ell)$, both $\pi'_k(\frac{k}{\ell})$ and $\pi'_{k+1}(\frac{k}{\ell})$ coincide with the barycenter of $\map^{-1}(\pi(k))$, so that defining $\pi'(r)=\pi'_k(r)$ for $r\in[\frac{k-1}{\ell},\frac{k}{\ell}]$ correctly defines a  topological path (actually a piece-wise linear path), satisfying (i) and (ii).
Finally since $\pi$ is a \dwn-path, $\pi(\ell) \preccurlyeq\pi(\ell-1)$, so that 
$\pi'([\frac{\ell-1}{\ell},1))\subseteq\map^{-1}(\pi(\ell-1))$. This implies (iii) above.
  \end{proof}

\subsection{Proof of Lemma~\ref{lem:ExisDwnPath}}\label{apx:prf:lem:ExisDwnPath}$ $\\[0.5em]
\noindent
{\bf Lemma~\ref{lem:ExisDwnPath}.}
{\em
Given a polyhedral model $\calP=(|K|,\peval{\calP})$, and associated cell poset model $\map(\calP)=(W,\preccurlyeq,\peval{\map(\calP)})$, for any topological path $\pi:[0,1] \to |K|$ the following holds:
$\map(\pi([0,1]))$ is a connected subposet of $W$ and there is $k>0$ and
a \dwn-path $\hat\pi: [0;k] \to W$ from $\map(\pi(0))$ to $\map(\pi(1))$
such that for all $i\in [0;k)$ there is $r\in [0,1)$ with $\hat\pi(i)=\map(\pi(r))$.
}

\begin{proof}
Continuity of $\map \circ \pi$ ensures that $\map(\pi([0,1]))$ is a connected subposet of $W$.
Thus there is an undirected path $\hat\pi: [0;k] \to W$ from $\map(\pi(0))$ to $\map(\pi(1))$
of some length $k>0$. In particular, $\hat\pi(k - 1) \succcurlyeq \hat\pi(k)$, 
as shown in the sequel, by contradiction.
Suppose that $\hat\pi(k - 1) \prec \hat\pi(k)$.  This would mean
that there is $\epsilon<1$, with 
$\pi(\epsilon) \in \map(\pi(\epsilon))=\hat\pi(k - 1)$, such that 
$\pi(r') \in \hat\pi(k)= \map(\pi(1))$ for no $r' \in (\epsilon,1)$  --- otherwise 
$\hat\pi(k - 1) = \hat\pi(k)$ would hold. But the fact that no such an $r'$ exists contradicts the fact that  $\pi$ is continuous, since continuity requires that
for each neighbourhood $N_1(\pi(1))$ of $\pi(1)$ there is a neighbourhood $N_2(1) \subseteq [0,1]$ of $1$ such that $\pi(t) \in N_1(\pi(1))$ whenever $t \in N_2(1)$.
We thus conclude that $\hat\pi(k - 1) \succcurlyeq \hat\pi(k)$, and so $\hat\pi_1$
is a \dwn-path. By definition and connectedness of $\map(\pi([0,1]))$ we finally get
that for all $i\in [0;k)$ there is $r\in [0,1)$ with $\hat\pi(i)=\map(\pi(r))$.
  \end{proof}

\subsection{Proof of Lemma~\ref{lem:LccImplSlcsEeq}}\label{apx:prf:lem:LccImplSlcsEeq}$ $\\

\noindent
{\bf Lemma~\ref{lem:LccImplSlcsEeq}.}
{\em
Given a finite poset model $\calF=(W,\preccurlyeq,\peval{\calF})$ and $w_1, w_2 \in W$ the following holds:
if $w_1 \lcceq w_2$, then $w_1 \slcsEeq w_2$.
}

\begin{proof}
By induction on the structure of \slcsE{} formulas. We show only the case for
$\eta(\form_1,\form_2)$ since the others are straightforward.
Suppose $\calF,w_1 \models \eta(\form_1,\form_2)$. Then there is a 
\plm-path $\pi$ from $w_1$ of some length $\ell$ such that
$\calF,\pi(\ell)\models \form_2$ and $\calF,\pi(i)\models \form_1$ for all $i\in [0;\ell)$.
In particular, we have that $\calF,w_1 \models \form_1$. So, by the Induction Hypothesis,
since $w_1 \lcceq w_2$, we get that also $\calF,w_2 \models \form_1$. In addition,
by definition of $\lcceq$, and given that $w_2 \lcceq w_1$,
there is an undirected path $\pi'$ of some length $\ell'$ such that
$\pi'(0)=w_2, \pi(\ell')=w_1$ and $\invpeval{\calF}(\SET{\pi'(i)})=\invpeval{\calF}(\SET{\pi'(j)})$,
for all $i,j\in [0;\ell']$. Note that,  by definition of $\lcceq$, we have that
$\pi'(k) \lcceq w_1$ for all $k\in [0;\ell']$. Thus, again by the Induction Hypothesis, we also get
$\calF,\pi'(k) \models \form_1$ for all $k\in [0;\ell']$. 
Clearly, the sequentialisation $\pi'\cdot\pi$ of $\pi'$ with $\pi$ is a \dwn-path since
$\pi$ is a \plm-path. Furthermore, by Lemma~\ref{lem:d2plm}, there is 
a \plm-path $\pi''$ with the same starting and ending points as $\pi'\cdot\pi$, 
and with the same set of intermediate points, occurring in the same order.
Thus $\pi''$ witnesses $\calF,w_2 \models \eta(\form_1,\form_2)$.
  \end{proof}

\subsection{Proof of Lemma~\ref{lem:StrongLab}}\label{apx:prf:lem:StrongLab}$ $\\

\noindent
{\bf Lemma~\ref{lem:StrongLab}.}
{\em
Consider a finite poset model $\calF=(W,\preccurlyeq,\peval{\calF})$. Then for all $w_1, w_2 \in W$ the following holds: 
if $[w_1]_{\lcceq} \seq^{\posToltsA(\calF)} [w_2]_{\lcceq}$, 
then $\invpeval{\calF}(\SET{w_1}) = \invpeval{\calF}(\SET{w_2})$.
}

\begin{proof}
By Rule (PL), we have
$[w_1]_{\lcceq} \trans{\invpeval{\calF}(\SET{w_1})}[w_1]_{\lcceq}$ and, by hypothesis, we also have $[w_2]_{\lcceq} \trans{\invpeval{\calF}(\SET{w_1})}[w'_2]_{\lcceq}$,
for some $[w'_2]_{\lcceq} \seq [w_1]_{\lcceq}$. But then,
using again Rule (PL), we get $[w'_2]_{\lcceq} = [w_2]_{\lcceq}$ and
$\invpeval{\calF}(\SET{w_1}) = \invpeval{\calF}(\SET{w_2})$.
  \end{proof}

\section{3D Maze Example of Section~\ref{sec:Experiments}}\label{apx:AdditionalExamples}

Below, the spatial logic specification in \imgql{} is shown, that was used for model checking the various maze-variants in Table~\ref{tab:toolchain} in Section~\ref{sec:Experiments} with \polylogica.
\imgql{} is the input language of \polylogica{} in which spatial logic properties of \slcsE{} can be expressed. In the specification below, first the polyhedral model is loaded in {$\mathsf{json}$} format. 
After that, the atomic propositions {$\mathbf{green}$}, {$\mathbf{white}$} and {$\mathbf{corridor}$} are defined. This is followed by a number of properties for the maze that should be self-explanatory. They include the formulas for $\phi_1$ and $\phi_2$ that were introduced in Section~\ref{sec:Experiments}. Finally, the lines starting by {\sf save} are  defining which results to save in a file. Such files contain the name of a property and for each property a list of true/false items, one for each  cell in the polyhedral model and in the order in which these cells are defined in that polyhedral model.\\

{\tiny
\begin{verbatim}
load model = "polyInput_Poset.json"

let green       = ap("G")
let white       = ap("W")
let corridor    = ap("corridor")


let greenOrWhite		= (green | white)

let oneStepToWhite   = eta((green | eta(corridor,white)),white)
let twoStepsToWhite  = eta((green | eta(corridor,oneStepToWhite)), oneStepToWhite) & (!oneStepToWhite)
let threeStepsToWhite = eta((green | eta(corridor,twoStepsToWhite)), twoStepsToWhite) & 
                                             (!twoStepsToWhite) & (!oneStepToWhite)

let phi1 = eta((green | eta(corridor,white)),white)
let phi2 = eta((green | eta(corridor,oneStepToWhite)), oneStepToWhite)

save "green" green
save "white" white
save "corr" corridor
save "phi1" phi1
save "phi2" phi2

\end{verbatim}
}

$ $\\
Figure~\ref{fig:cubesLTS} shows the 3x5x3 maze and its minimised LTS. Note that in the LTS not all transition labels are shown in order to avoid cluttering of the image. However, states corresponding to corridors, green rooms and white rooms, are shown in grey, green and white, respectively. 

\begin{figure}
  \begin{center}
  \subfloat[\label{subfig:c3x5x3}Maze 3x5x3]{\phantom{AAAA}\includegraphics[width=0.4\textwidth]{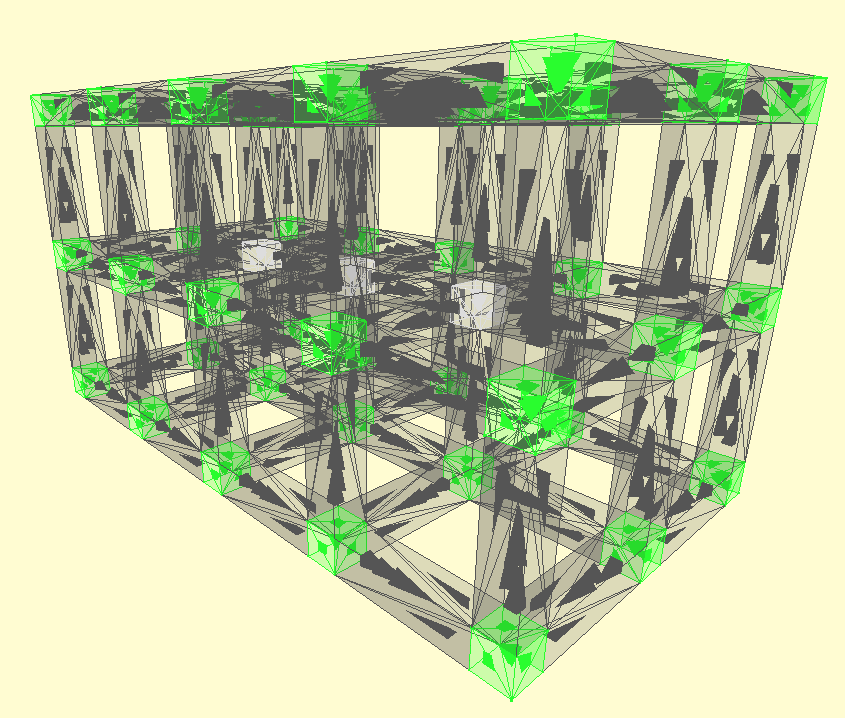}}
   \subfloat[\label{subfig:c3x5x3min}Minimised LTS]{
   \resizebox{2.5in}{!}{

   }
   }
  \end{center}
  \caption{Maze of dimension 3x5x3 (Fig.~\ref{subfig:c3x5x3}) and its respective minimal LTSs (Figs.~\ref{subfig:c3x5x3min}).}\label{fig:cubesLTS}
\end{figure}

\end{document}